\documentclass{article}

\usepackage{arxiv}
\usepackage{tikz}
\usetikzlibrary{shapes.geometric, arrows}

\usepackage{xcolor}

\usepackage[utf8]{inputenc} 
\usepackage[T1]{fontenc}    
\usepackage{hyperref}       
\usepackage{url}            
\usepackage{booktabs}       
\usepackage{amsfonts}       
\usepackage{nicefrac}       
\usepackage{microtype}      
\usepackage{lipsum}         
\usepackage{graphicx}
\usepackage{natbib}
\usepackage{doi}
\usepackage{hyperref}
\usepackage{float}
\usepackage{url}            
\usepackage{booktabs}       
\usepackage{amsfonts}       
\usepackage{amssymb}        
\usepackage{pifont}
\usepackage{nicefrac}       
\usepackage{microtype}      
\usepackage{bm}				
\usepackage{mathrsfs}		
\usepackage{cite}			
\usepackage{graphicx}		
\usepackage{mathtools}		
\usepackage{algpseudocode}  	
\usepackage{algorithm}
\usepackage{amsthm}			
\usepackage{enumitem}		
\usepackage{multirow}       
\usepackage{tabularx}	    
\usepackage{hhline}
\usepackage{arydshln}		
\usepackage{wrapfig}
\usepackage{lipsum}
\usepackage{diagbox}
\usepackage{soul}
\usepackage{amsthm}
\usepackage{enumitem}
\usepackage{titlesec}

\titleformat{\paragraph}
{\normalfont\normalsize\bfseries}{\theparagraph}{1em}{}
\titlespacing*{\paragraph}
{0pt}{3.25ex plus 1ex minus .2ex}{1.5ex plus .2ex}

\theoremstyle{plain} 
\newtheorem{proposition}{Proposition}

\newtheorem{theorem}{Theorem}
\newtheorem{lemma}{Lemma}
\newtheorem{corollary}{Corollary}
\newtheorem{assumption}{Assumption}

\title{Volatility Calibration via Automatic Local Regression}

\author{
    {\large Ruozhong Yang$^{\dagger}$, \large Hao Qin$^{\dagger}$,  \href{https://orcid.org/0009-0009-4556-8664}{\includegraphics[scale=0.06]{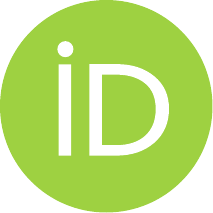}\hspace{1mm}Charlie Che$^{\dagger\ddagger}$},and Liming Feng$^{\dagger}$\thanks{Corresponding author. Email: fenglm@illinois.edu. First authors: Ruozhong Yang and Hao Qin}}\\[1ex]
    $^{\dagger}$Department of Industrial and Enterprise Systems Engineering, University of Illinois Urbana-Champaign, Illinois, United States\\
    $^{\ddagger}$JPMorgan Chase \& Co., New York, United States
}
\date{} 

\renewcommand{\shorttitle}{Volatility Calibration via Automatic Local Regression}

\begin{document}
\maketitle

\begin{abstract}
Managing exotic derivatives requires accurate mark‑to‑market pricing and stable Greeks for reliable hedging. The Local Volatility (LV) model distinguishes itself from other pricing models by its ability to match observable market prices across all strikes and maturities with high accuracy. However, LV calibration is fundamentally ill-posed: finite market observables must determine a continuously-defined surface with infinite local volatility parameters.  This ill-posed nature often causes spiky LV surfaces that are particularly problematic for finite-difference-based valuation, and induces high-frequency oscillations in solutions, thus leading to unstable Greeks. To address this challenge, we propose a pre‐calibration smoothing method that can be integrated seamlessly into any LV calibration workflow. Our method pre-processes market observables using local regression that automatically minimizes asymptotic conditional mean squared error to generate denoised inputs for subsequent LV calibration. Numerical experiments demonstrate that the proposed pre-calibration smoothing yields significantly smoother LV surfaces and greatly improves Greek stability for exotic options with negligible additional computational cost, while preserving the LV model’s ability to fit market observables with high fidelity.

\end{abstract}
\vspace{1em} 

\vspace{1em} 

\section{Introduction}

Originally proposed by Dupire \citep{dupire1994pricing}, the Local Volatility (LV) model offers a unique framework designed to match observable market prices with high accuracy across all strikes and maturities. Unlike other parametric models that calibrate a finite set of parameters to market observables, the LV model requires determining a continuous volatility surface $\sigma_{\mathrm{loc}}(K, T)$ as a function of strike $K$ and maturity $T$. This functional calibration problem is inherently ill-posed: finite market observables must uniquely determine an entire surface, yet infinitely many local volatility surfaces can reproduce the same set of option prices. Moreover, the LV model does not by itself remove static arbitrage embedded in the market observables; arbitrage-free constraints must therefore be enforced during calibration. This additional requirement increases implementation complexity, as any such implementation must simultaneously achieve precise fitting and rigorously enforce arbitrage-free conditions. 

The LV model's reliance on numerical partial derivatives of market prices through the Dupire formula creates fundamental implementation challenges. This derivative-based approach makes the model highly sensitive to perturbations in market observables. Even minor fluctuations in observed prices produce large variations in the computed derivatives, leading to numerical instability in the generated local volatility function $\sigma_{\text {loc }}(K, T)$. When calibrating to market data containing bid-ask spreads, noisy quotes, and interpolation errors, these spurious fluctuations become embedded in the volatility surface, resulting in high frequency oscillations and artificial spikes. Such instabilities directly influence option price sensitivities and yield unstable Greeks that compromise hedging reliability. While industry practice rightly prioritizes exact calibration to market observables, the need for a smooth and stable LV surface for reliable hedging is equally important and warrants the same level of emphasis.

Given these challenges, the calibration of the LV model has been a significant research topic for decades. Researchers aim to reformulate the original problem into a well-posed one while ensuring that the generated LV function is arbitrage-free and robust. Numerous implementation methods have been proposed — to name just a few — \citep{achdou2002volatility, benko2007extracting, andreasen2010volatility, lipton2011filling, reghai2012local, buehler2018discrete}. Beyond methods concentrating exclusively on LV calibration mechanics, the literature has pursued two primary strategies to address the ill-posedness and ensure smoothness.

The first strategy utilizes pre-calibration smoothing methods to process market inputs before applying the Dupire inversion. Global quote fitting methods typically assume that the price or implied volatility surface can be represented by a specific global parametric function. By fitting this function, practitioners obtain a continuous representation of the current market, enabling the safe application of the Dupire formula. Examples include parametric models for the implied volatility surface (e.g., \citealp{gatheral2014arbitrage}) and parametric stochastic volatility models (e.g., \citealp{heston1993closed, hagan2002managing}). As an alternative, local quote fitting methods assume that price or implied volatility values at given strikes and maturities are locally determined by surrounding observations. Practitioners typically construct a fine mesh over the strike-maturity domain and solve independent optimization problems for values at each mesh point. The resulting price or implied volatility representation consists of discrete points rather than a continuous function. Popular approaches include local fitting methods for the IV surface (e.g., \citep{benko2007extracting}) and for the price surface (e.g., \citep{glaser2012arbitrage}). This pre-processing approach offers clean separation of concerns by first denoising data, then calibrating. It ensures the Dupire formula receives smooth inputs, naturally leading to smooth LV surfaces through an intuitive and modular two-stage process. However, while global parameterizations can capture broad market structure, they may misrepresent localized features. Local methods adapt to nearby data but require manual tuning of bandwidths and lack clear optimality guarantees.

The second research direction embeds smoothness directly into the calibration process through regularization techniques that penalize high-order derivatives or project the volatility surface onto low-dimensional bases (e.g., \citep{achdou2002volatility}). These methods formulate single optimization problems that simultaneously pursue accuracy and smoothness objectives. They offer problem‑formulation efficiency by eliminating separate smoothing steps and naturally incorporate arbitrage constraints within the optimization framework. However, this approach involves inherent compromises in multi-objective optimization. Smoothness is achieved through strong assumptions about price or LV derivatives, and practical implementations often require parameter tuning to balance competing objectives.

To address these limitations, we propose a novel pre-calibration smoothing method that can be seamlessly integrated with cutting-edge LV calibration techniques. Our approach incurs negligible additional computational cost while preserving the model's ability to fit market observables precisely. Specifically, we denoise the IV surface via local regression that automatically minimizes asymptotic conditional mean squared error. By inverting the Black-Scholes formula on the smoothed volatilities, we generate a smoother sequence of option prices that are sufficiently close to raw market observables. Our work is distinguished by several fundamental advantages. We optimize the preprocessing stage with automatic parameter selection that ensures theoretically optimal smoothing levels. Our two-stage approach avoids the compromises inherent in multi-objective optimization by cleanly separating denoising from calibration. The smoothness is achieved through optimal bias-variance tradeoff in curve fitting, thus eliminating strong assumptions about price or LV derivatives. This framework maintains market structure while removing noise and requires no manual parameter tuning in practice.

The remainder of this paper is organized as follows. Section 2 explains existing approaches and their limitations in detail. Section 3 sets up the problem and presents the proposed two-stage implementation method. Section 4 provides numerical experiments illustrating the effectiveness of our proposed method. Finally, Section 5 concludes. Technical proofs are deferred to the Appendix section.

\section{Existing Methods}
The ill-posed nature of local volatility calibration and the requirement for smooth LV surfaces have motivated decades of research, and various approaches are proposed to address three main issues: arbitrage-freeness, smoothness of LV surface, and precision of fitting. Implementation methods fall into two branches: pre-calibration smoothing approaches that process market data before applying the Dupire inversion, and integrated calibration techniques that embed smoothness constraints directly within the optimization framework.

\subsection{Pre-Calibration Smoothing Methods}
Pre-calibration smoothing methods recognize that clean, smooth inputs to the Dupire formula naturally produce well-behaved local volatility surfaces. These approaches operate in two stages: first processing market data to remove some noise and ensure smoothness, then applying common LV calibration techniques. By separating smoothing from calibration, these methods make each stage single‑purpose: smoothing deals only with denoising and shape control, while calibration focuses on price fit and arbitrage constraints.

\subsubsection{Global Quote Fitting Approaches}

Global quote fitting methods either assume parametric forms for the IV or price surface (e.g., \citealp{gatheral2014arbitrage}), or for the underlying process (e.g., \citealp{heston1993closed, hagan2002managing}). The static arbitrage-free conditions $\frac{\partial C}{\partial t} \geq 0$ and $\frac{\partial^2 C}{\partial K^2} \geq 0$ can then be written as explicit inequality constraints on the parameter vector and enforced within the calibration optimization.

These methods offer analytical tractability and computational efficiency. Calibrated models provide closed-form implied volatilities for rapid pricing and stable Greek calculations. The parametric structure ensures smoothness and simplifies arbitrage constraint enforcement. However, because of the strong parameterization, the resulting fit is not always satisfactory.

\subsubsection{Local Quote Fitting Approaches}

Local quote fitting methods determine values at each point from surrounding market observables. Practitioners construct meshes and solve optimization problems at each point independently. Benko et al. \citep{benko2007extracting} apply local polynomial regression with arbitrage constraints to implied volatility surfaces. Glaser and Heider \citep{glaser2012arbitrage} use shape-preserving interpolation to maintain price surface convexity and time monotonicity. Fengler \citep{fengler2015simple} implements tensor product B-splines with smoothness penalties.

Local methods capture complex market structures effectively. They preserve sharp features in liquid regions and handle uneven strike spacing naturally. The independent estimation at each mesh point enables parallel implementation for efficient pre-processing. However, these methods typically require manual tuning of neighborhood size, polynomial order, and spline penalties, and the resulting surface can be quite sensitive to these choices.

\subsection{Integrated Calibration Methods}
Integrated calibration methods incorporate smoothness directly in the optimization. Regularization approaches exemplify this category. Achdou and Pironneau \citep{achdou2002volatility} minimize pricing errors while controlling smoothness through penalty terms. Tikhonov regularization \citep{geng2014non} penalizes second derivatives; total variation regularization preserves sharp features while limiting overall variation. Relative entropy \citep{cont2004recovering} and maximum entropy methods \citep{bouchouev1999uniqueness} incorporate prior volatility surface information that helps generate the smoothness.

These integrated approaches combine calibration and smoothness objectives in a single optimization framework. The regularization parameters directly control the balance between fitting precision and surface regularity. This unified formulation addresses the ill-posed nature of the inverse problem and incorporates arbitrage constraints alongside smoothness requirements. The mathematical consistency throughout the optimization process enhances numerical stability. Beyond regularization techniques commonly used for local volatility, some recent works have also applied neural networks to LV calibration\citep{wang2021deep}, and the strategies developed in related volatility frameworks \citep{cuchiero2020generative},\citep{fu2022solving} offers potential insights for stable local volatility implementations. However, combining smoothing and LV calibration into a single program increases problem size and nonlinearity, making the optimization harder to solve and more sensitive to initialization. Moreover, regularization inevitably embeds prior assumptions (e.g., about derivatives or surface shape); inappropriate weights or priors can bias the solution, so hyperparameter tuning remains nontrivial.

\subsection{Our Contribution}
We advance the pre-calibration smoothing approach by automating local regression for noisy market observables. Specifically, the polynomial order $p$ and bandwidth $h$ are chosen by minimizing the asymptotic conditional mean-squared error (ACMSE). This automation eliminates manual tuning entirely and makes the approach practical for production environments. Where processing hundreds of surfaces without asset-specific parameter adjustments is a must.

Our framework follows the established two-stage pre-calibration smoothing methodology. The first stage produces optimally smoothed market observables through local regression with automatic bandwidth and polynomial selection. The second stage uses these smoothed observables as direct inputs to LV calibration methods. This separation maintains mathematical integrity while ensuring each stage focuses on its primary objective. The denoising process concentrates solely on optimal smoothness without arbitrage considerations, while the subsequent calibration ensures arbitrage-free dynamics.

Our approach operates directly on market observables without requiring specific parametric assumptions. Crucially, the smoothing stage involves no numerical optimization: local regression yields closed-form solutions, so it is fast and deterministic. The local regression methodology adapts naturally to the data structure and captures complex patterns across diverse asset classes and market conditions. This flexibility stems from the nonparametric nature of our method. The statistical foundation ensures smoothness quality through theoretically optimal bias–variance trade-offs. Each surface receives customized treatment based on its own characteristics rather than predetermined model parameters. Because no regularization terms are introduced here, we avoid injecting potentially wrong priors and keep the downstream optimization problem small and simple.

The automatic parameter selection represents a key advancement in pre-smoothing calibration. Our method determines optimal smoothing parameters through data-driven optimization for each asset independently. This eliminates the need for manual parameter tuning that traditionally requires extensive experimentation and market expertise. Empirical results demonstrate significantly smoother local volatility surfaces with stable Greeks while maintaining calibration accuracy. The absence of regularization in the calibration objective further simplifies the problem and improves runtime. The automation makes pre-calibration smoothing practical for large-scale applications where practitioners process hundreds of surfaces across different assets daily. This systematic approach ensures consistent quality across different market environments and reduces operational complexity in trading systems.

\section{Methodology}
Our goal is to generate a smooth, arbitrage-free LV surface that matches market observables with high accuracy. We achieve this through a two-stage approach that separates noise removal from arbitrage enforcement:

\begin{enumerate}
\item \textbf{Stage 1: From Noisy IV to Smooth IV.} We apply automatic local regression to remove noise from market IV while preserving market structure. The algorithm minimizes asymptotic conditional MSE by automatically selecting optimal polynomial order and bandwidth parameters adapted to local market characteristics (e.g., trading volume).

\item \textbf{Stage 2: From Smooth IV to Arbitrage-free LV.} We transform the smoothed IV into an arbitrage-free local volatility surface. This involves converting smoothed IVs to option prices via Black-Scholes, then applying finite difference methods with explicit arbitrage constraints through the Dupire formula.
\end{enumerate}

This separation ensures each stage focuses on its primary objective without compromise: Stage 1 achieves optimal smoothness through statistical principles, while Stage 2 enforces arbitrage-free dynamics on already-smooth inputs. Note that both the pre-calibration smoothing method and the LV fitting method can be applied to either the IVs or prices of the market observables. In this paper, we choose to apply our proposed method to the IVs, then convert the smoothed IVs to prices using the Black-Scholes formula, and finally apply the price version of the Dupire formula.

\subsection{Stage 1: From Noisy IV to Smooth IV} 
\label{sec:From Noisy IV to Smooth IV}
To filter market IV noise while preserving market structure, we employ local regression: a nonparametric technique that adapts to local data without imposing global functional forms. The smoothness is controlled by bandwidth $h$ and polynomial order $p$, and they mutually determine the bias-variance tradeoff: bias measures systematic deviation from the true function, while variance captures noise sensitivity. 
    
This tradeoff directly governs the model's performance, yet traditional implementations require manual parameter tuning—a subjective and time-consuming process. Our framework addresses this by automatically selecting both polynomial order and bandwidth to minimize the asymptotic conditional mean squared error at each point, thus eliminating manual parameter choices entirely.

In the following, we first review the definition of the general local regression model in section \ref{sec:LR} and its bias–variance tradeoff in section \ref{sec:BVT}. We then introduce our automatic local regression framework with necessary assumptions in section \ref{sec: ALR FRAMEWORK}, and shows it can optimally balances the tradeoff to achieve minimal asymptotic conditional MSE. 

To apply our framework, we partition the market observables by maturity and process each maturity group independently. Within each group, we iteratively apply the polynomial order selector (Section \ref{sec:Data-Driven Optimal Polynomial Order Selector}) and bandwidth selector (Section \ref{sec:Data-Driven Optimal Bandwidth Selector}) until convergence to the optimal parameter set. Section \ref{sec:From Asymptotic Conditional MSE to True MSE} gives advanced theoretical results that can be achieved through this framework.

\subsubsection{General Local Regression}\label{sec:LR}

We model the observed noisy IV sequence $\{\sigma_i\}_{i=1}^{n}$ and its corresponding strike sequence $\{K_i\}_{i=1}^{n}$ of a given maturity as:

\[
\sigma_i = f(K_i) + \varepsilon_i, \quad i=1,\ldots,n,
\]

where $f(\cdot)$ denotes the true IV function, and whose value is uniquely determined by strike. The term $\varepsilon_i$ represents random noise, satisfying $E\left(\varepsilon_i \mid K_i\right) = 0$ and $\operatorname{Var}\left(\varepsilon_i \mid K_i\right) = \tau_i^2$. Our goal is to estimate the value of the true function $f(\cdot)$ at a given strike $k$ using nearby noisy observations.

Following the notation in \citep{fan2018local}, the general local regression model of polynomial order $p$ is formulated as a minimization problem centered at strike $k$:

\begin{equation}
\min_{\alpha} \sum_{i=1}^{n}\left[\sigma_i - \sum_{j=0}^{p} \alpha_j (K_i - k)^j\right]^2 \kappa_h(K_i - k),
\end{equation}

where $\kappa_h(\cdot)$ is a regular kernel function with bandwidth $h$. Define the $n \times (p+1)$ Vandermonde matrix 

\[
\mathbf{X} = \mathcal{V}(K_1 - k,\, K_2 - k,\, \ldots,\, K_n - k),
\]

the vector $\sigma = (\sigma_1, \sigma_2, \ldots, \sigma_n)^T$, and the polynomial coefficient vector 

\[
\alpha = (\alpha_0, \alpha_1, \ldots, \alpha_p)^T.
\]

Let $W$ be the diagonal weight matrix given by

\begin{equation}
\mathbf{W} = \operatorname{diag}\{\kappa_h(K_i - k)\}.
\end{equation}

The weighted least squares problem then becomes

\begin{equation}
\min_{\alpha} (\sigma - \mathbf{X}\alpha)^T W (\sigma - \mathbf{X}\alpha),
\end{equation}

with the solution denoted as $\hat{\alpha} = (\hat{\alpha}_0, \hat{\alpha}_1, \ldots, \hat{\alpha}_p)^T$
\begin{equation}
\hat{\alpha} = (\mathbf{X}^T \mathbf{W} \mathbf{X})^{-1} \mathbf{X}^T \mathbf{W} \sigma.
\end{equation}

Here, $\hat{\alpha}_0$ is the estimated noise-free true IV at strike $k$.

\subsubsection{Bias-Variance Tradeoff}\label{sec:BVT}

As mentioned earlier, the bandwidth $h$ and polynomial order $p$ jointly determine the bias-variance tradeoff in local regression. To develop our automatic parameter selection framework, we need to quantify this tradeoff mathematically and understand how these parameters affect the asymptotic conditional mean squared error.

Managing the bias-variance tradeoff plays a key role in the performance of the local regression model. However, the true bias and variance depend on unknown quantities: the design density $g(\cdot)$, true IV function $f(\cdot)$, and noise variance $\tau(\cdot)$ \citep{fan2018local}(chapter 3, chapter 4), and cannot be obtained. As a result, we focus on analyzing and controlling the asymptotic conditional bias and variance. According to Fan's paper, the asymptotic bias and variance of $\hat{\alpha}_0$ at strike $k$ conditionally upon $\mathbb{K}=\{K_1,\, K_2,\, \ldots,\, K_n\}$ are given by:

\begin{equation}
\label{eq:ACMSE VAR}
    \operatorname{Var}\{\hat{\alpha}_0(k) \mid \mathbb{K} \} =
    e_{1}^{T} S^{-1} S^* S^{-1} e_{1} 
    \frac{\hat{\tau}^2(k)}{\hat{g}(k) nh}
    + \mathcal{R}_n(p,h).
\end{equation}

For odd $p$, the asymptotic conditional bias is

\begin{equation}
\label{eq:ACMSE ODD BIAS}
    \operatorname{Bias} \{\hat{\alpha}_0(k) \mid \mathbb{K} \} =
    e_{1}^{T} S^{-1} c_p \frac{1}{(p+1)!} \hat{f}^{(p+1)}(k) h^{p+1} + \mathcal{R}_n(p,h).
\end{equation}

For even $p$, the asymptotic conditional bias is:

\begin{equation}
\label{eq:ACMSE EVEN BIAS}
    \operatorname{Bias} \{\hat{\alpha}_0(k) \mid \mathbb{K} \} =
    e_{1}^{T} S^{-1} \tilde{c}_p \frac{1}{(p+2)!} 
    \left\{ \hat{f}^{(p+2)}(k) + (p+2) \hat{f}^{(p+1)}(k) \frac{\hat{g}'(k)}{\hat{g}(k)} \right\} h^{p+2}+ \mathcal{R}_n(p,h).
\end{equation}

Also, the asymptotic conditional MSE can be written as:

\begin{equation}
    Z\triangleq\operatorname{MSE}=(\operatorname{Bias})^2+\operatorname{Var}.
\end{equation}

Where $S$, $S^*$, $c_p$, $\tilde{c}_p$ are given and related to the selected regular kernel function, $e_1 = (1,0,\ldots,0)^T$ is the unit vector with $1$ at the first position, and $\mathcal{R}_n(p,h)$ represent higher order terms related to $p$ and $h$. As for $\hat{g}(\cdot)$, $\hat{f}(\cdot)$, and $\hat{\tau}(\cdot)$, they are estimations of the unknown quantities $g(\cdot)$, $f(\cdot)$, and $\tau(\cdot)$ in the information space, and converge to the true value in probability under appropriate assumptions.

These equations clearly illustrate the tradeoff: reducing the bandwidth $h$ decreases the bias but increases the variance, and vice versa. In contrast, unlike this clear relationship for bandwidth, the effect of the polynomial order $p$ on this tradeoff is more nuanced. It impacts calculation of $S,S^*,c_p,\tilde{c}_p$, and also other terms, making its effect more subtle and thus harder to control.

Intuitively, an improper selection of the polynomial order and bandwidth may lead to either overfitting (excessive variance) or over-smoothing (excessive bias). Traditional implementations often rely on subjective choices for them \citep{benko2007extracting, glaser2012arbitrage}, making the model performance heavily dependent on the user's experience. However, human judgment is inherently imprecise and cannot guarantee an optimal balance. Moreover, in practice, manually achieving a well-tuned regression typically requires a time-consuming iterative process: repeatedly applying local regression, adjusting the bandwidth and polynomial order based on experience, and then visually assessing whether the performance has improved.

This trial-and-error approach presents three major issues. First, when the differences between results are subtle, it becomes difficult—if not impossible—to visually determine which result is preferred. Second, manually selecting parameters typically restrict choices to a limited set of predefined values, preventing a thorough search for the true optimal combination. Consequently, we may settle for a suboptimal balance between bias and variance, resulting in limited performance. Third, and perhaps most importantly, manual tuning is impractical in real-world applications. For example, a trader may need to fit hundreds of LV surfaces for pricing needs. In such cases, manually fine-tuning parameters for each surface is not feasible. 

Given these limitations, a more reliable framework is needed—one that systematically determines the optimal polynomial order and bandwidth without human intervention. By replacing the subjective tuning process with an automatic framework, we can ensure that the observable market IVs are optimally denoised. After that, these denoised IVs are then used as input for the LV fitting stage.

\subsubsection{Automatic Local Regression Framework}\label{sec: ALR FRAMEWORK}

In this section, we provide an overview of our automatic local regression framework for a fixed strike $k$. This framework consists of two key components, which are applied iteratively: a data-driven optimal polynomial order selector and a data-driven optimal bandwidth selector. For a given input IV sequence with the same maturity, the asymptotic conditional MSE for a fixed strike $k$, $\hat{Z}_k(p,h)$, is a function of the polynomial order $p$ and the bandwidth $h$. We begin by selecting an initial pair $(p_0,h_0)$ as the starting point. Then, we iteratively apply the order selector followed by the bandwidth selector. As demonstrated in proposition \ref{prop:ACMSE convergece in algo}, this iterative framework converges to the minimum asymptotic conditional MSE as $n\rightarrow\infty$, where $n$ here refers to the iteration numbers. In practical applications, it is advisable to terminate the algorithm after a predefined number of iterations or when the difference between two consecutive iterations falls below a specified threshold.

\begin{algorithm}[H]
\caption{Optimal Denoising Parameter Selection}
\label{alg:Optimal Denoising Parameter Selection}
\begin{algorithmic}
\State \textbf{Initialize:} Select an initial pair \((p_0, h_0)\)
\For{$n = 0, 1, 2, \dots$}
    \State \textbf{Step 1: Optimal Polynomial Order Selection}
    \State \quad $p_{n+1} = \arg\min\limits_{p\in\mathcal{P}} \hat{Z}_k(p, h_n)$
    \State \textbf{Step 2: Optimal Bandwidth Selection}
    \State \quad $h_{n+1} = \arg\min\limits_{h} \hat{Z}_k(p_{n+1}, h)$
    \If{$|\hat{Z}_k(p_{n+1}, h_{n+1}) - \hat{Z}(p_n, h_n)| < \epsilon$}
        \State \textbf{Break}
    \EndIf
\EndFor
\State \Return $p_{n+1}$, $h_{n+1}$
\end{algorithmic}
\end{algorithm}

After outlining the iterative algorithm, we now turn our attention to its performance characteristics. Specifically, we claim the following:  

\begin{enumerate}
    \item \textbf{Global Optimality of the Selectors:} 
    Both the polynomial order selector and the bandwidth selector achieve the global optimal at each iteration. Further detail of this claim is provided in the respective sections.  
    
    \item \textbf{Convergence to the Minimal Asymptotic Conditional MSE:}     
    The objective function value in this iterative framework converges to the global minimum of the asymptotic conditional MSE $\hat{Z}^*$($\hat{Z}^*$ is derived based on the finite market observables). Formal proof of this convergence is provided in Proposition \ref{prop:ACMSE convergece in algo}. 
        
    \item \textbf{Convergence of Asymptotic Conditional MSE to MSE:} 
    Under the same search space, and with appropriate assumptions, the minimum asymptotic conditional MSE $\hat{Z}^*$ converges to the true minimum MSE $Z^*$($Z    ^*$ is derived on theoretical assumption that there are infinite market observables). This convergence is analyzed in Section \ref{sec: IV Denoising theory conclusion}.

\end{enumerate}

From these claims, we conclude that our framework effectively searches for the optimal asymptotic conditional MSE, which, under suitable conditions and assumptions, is asymptotically efficient with respect to the true MSE \citep{fan2018local}. In the following paragraphs, we provide a detailed explanation of both selectors and their interaction within our iterative framework.

We now give a set of assumptions for the propositions in Section \ref{sec:From Noisy IV to Smooth IV}. Besides, those assumptions are also applied in equation \ref{eq:ACMSE VAR}, \ref{eq:ACMSE ODD BIAS}, and \ref{eq:ACMSE EVEN BIAS} \citep{fan2018local}.

\begin{assumption}
\label{assup:convergence analysis}
Denote $p$ as the polynominal order selected for local regression on fixed point $k_0$, we assume the following:
\begin{enumerate}

    \item[1]  The symmetric kernel $K$ is a continuous density function having bounded support;
    \item[2]  $g\left(k_0\right)>0$ and $g^{\prime \prime}(k)$ is bounded in a neighborhood of $k_0$;
    \item[3]  $f^{(p+3)}(\cdot)$ exists and is continuous in a neighborhood of $k_0$, and $f^{(p+1)}\left(k_0\right) \neq 0$;
    \item[4]  $\tau^2(\cdot)$ has a bounded second derivative in a neighborhood of $k_0$;
    \item[5]  Estimated function $\hat{g}(\cdot),\hat{f}(\cdot)$ and $\hat{\tau}(\cdot)$ also share property 2,3,4 respectively.
\end{enumerate}
\end{assumption}

\begin{proposition}
\label{prop:ACMSE convergece in algo}
    Besides the assumption \ref{assup:convergence analysis}, we assume global optimality holds for each parameter selection step. Let $\hat{Z}_k^* = \min\limits_{p,h} \hat{Z}_k(p,h)$ be the global minimum value of $\hat{Z}_k(p,h)$ and $\hat{Ph}^* = \{(p,h) \mid \hat{Z}_k(p,h) = \hat{Z}_k^*\}$ be the set of all parameter pairs achieving this minimum value. Then the iterative algorithm will converge to this global minimum value:
\begin{align*}
\lim_{n \to \infty} \hat{Z}_k(p_n, h_n) = \hat{Z}_k^*
\end{align*}
Furthermore, there exists an integer $N$ such that for all $n \geq N$, the algorithm produces parameter pairs satisfying $|\hat{Z}_k(p_{n+1}, h_{n+1}) - \hat{Z}_k(p_n, h_n)| < \epsilon$, where $\epsilon>0$.
\end{proposition}

\subsubsection{Data-Driven Optimal Polynomial Order Selector}
\label{sec:Data-Driven Optimal Polynomial Order Selector}






Beginning with a fixed bandwidth $h$ and strike $k$, we now search for the optimal polynomial order $p$ that minimizes the asymptotic conditional MSE. Assuming we have the true value of $\alpha$ up to its $\bar{p}$-th order, and $a = \bar{p} - p$. For a local regression of order $p$, the asymptotic conditional bias of $\hat\alpha$ can be estimated as:

\begin{equation}
\label{eq:bias_polynomialselect}
\operatorname{Bias} \{\hat{\alpha}(k) \mid \mathbb{K} \} = S_n^{-1}b,
\end{equation}

where 
\begin{equation}
S_n = \mathbf{X}^T\mathbf{W}\mathbf{X}
\end{equation}

is the weighted moment matrix. To understand the bias term $b$, we follow Fan's notation $S_{n,j}$ for the $j$-th order weighted moment:
\begin{equation}
S_{n,j} = \sum_{i=1}^n \kappa_h(K_i-k)(K_i-k)^j.
\end{equation}

Note that the $(i,j)$-th element of matrix $S_n$ equals $S_{n,i+j}$, i.e., 
$[S_n]_{i,j} = \sum_{k=1}^n \kappa_h(K_k-k)(K_k-k)^{i+j}$ for $i,j = 0,1,\ldots,p$.

When the true model has order $\bar{p}$ but we fit with order $p < \bar{p}$, the omitted higher-order terms $\sum_{j=p+1}^{\bar{p}} \alpha_j (K_i - k)^j$ contribute to the bias. Through the normal equations, this contribution takes the form:
\begin{equation}
b=\begin{pmatrix}
\alpha_{p+1}S_{n,p+1}+\cdots+\alpha_{p+a}S_{n,p+a} \\
\vdots \\
\alpha_{p+1}S_{n,2p+1}+\cdots+\alpha_{p+a}S_{n,2p+a},
\end{pmatrix}
\end{equation}


where vector $b$ is $(p+1)$-dimensional, and its $\ell$-th element (for $\ell = 0,1,\ldots,p$) captures how the omitted terms project onto the $\ell$-th basis function through the weighted inner product.

The asymptotic conditional variance matrix can be written as:
\begin{equation}
\label{eq:var_polynomialselector}
\operatorname{Var}\{\hat{\alpha}(k) \mid \mathbb{K} \} = S_n^{-1}\left(\mathbf{X}^T \Sigma \mathbf{X}\right)S_n^{-1},
\end{equation}

where
\begin{equation}
        \Sigma=\textbf{diag}\{\kappa^2_h(K_i-k)\tau^2(K_i)\}.
\end{equation}

As we already demonstrated, the calculation of asymptotic conditional MSE requires estimations of some unknown quantities. Based on Theorem 3.1 in \citep{fan2018local}, for a general local regression, its solution $\hat{\alpha} = (\hat{\alpha}_0, \hat{\alpha}_1, \ldots, \hat{\alpha}_p)^T$ converge to the true $\alpha$, under the assumption \ref{assup:convergence analysis} with $h\rightarrow 0, nh\rightarrow \infty, n\rightarrow \infty$. Therefore, a pilot estimator approach can be applied here to estimate those unknown quantities. Under the condition above, those estimations will converge to the true value.

After establishing the pilot estimator, we fit a sequence of local regressions with different polynomial order $p$. Their asymptotic conditional MSE can be obtained using the estimations from the pilot estimator and the formulas above. We then select the global optimal polynomial order of the current iteration, based on asymptotic conditional MSE minimization.

\paragraph{Construct Pilot Estimator}

We now fit a local regression with $\bar{p}$ and bandwidth $\bar{h}$ as the pilot estimator, which provides an approximation of $f(\cdot)$ within a small neighborhood of $k$. Although $\bar{p}$ can be picked discretionarily, we still need to determine the bandwidth $\bar{h}$ for the pilot estimator.

The numerical procedure for selecting the optimal pilot bandwidth $\bar{h}$ closely follows the leave-one-out cross-validation (LOO-CV) method commonly used in kernel smoothing. For each candidate bandwidth $h$ in a predefined grid $\mathcal{H}$, we fit a local regression while systematically leaving out one observation at a time. The resulting prediction errors are squared, and then aggregated into the cross-validation score:

\begin{equation}
\mathrm{CV}(h)=\frac{1}{n} \sum_{i=1}^n\left\{\sigma_i-\widehat{f}_{h,-i}\left(K_i\right)\right\}^2,
\end{equation}

where $\widehat{f}_{h,-i}\left(K_i\right)$ represents the local polynomial estimate at $K_i$, obtained using bandwidth $h$ while excluding the $i$-th data point. The optimal pilot bandwidth $\bar{h}$ is chosen as the value that minimizes $\mathrm{CV}(h)$.

This cross-validation approach eliminates the need for explicit knowledge of the underlying error structure, ensuring a fully data-driven selection of $\bar{h}$. By optimizing the bandwidth in this manner, we enhance the stability of higher-order coefficient estimation, thereby improving the accuracy and robustness of the pilot estimator.

\begin{algorithm}[H]
\caption{Cross-Validated Pilot Bandwidth Selection}
\label{alg:pilot_bandwidth}
\begin{algorithmic}
\Require Data: $\{(K_i, \sigma_i)\}_{i=1}^n$, strike $k$, candidate bandwidth grid $\mathcal{H} = \{h_1, h_2, \ldots, h_M\}$, pilot estimator polynomial order $\bar{p}$, kernel function $\kappa(\cdot)$.
\State Initialize $\mathrm{CV}(h) \gets 0$ for every $h\in \mathcal{H}$.
\For{each $h \in \mathcal{H}$}
  \For{$i=1$ to $n$}
    \State Exclude $(K_i,\sigma_i)$ from the sample.
    \State Fit a local polynomial of order $\bar{p}$ at $k$ using bandwidth $h$ with the remaining data.
    \State Predict the response $\widehat{f}_{h,-i}(K_i)$.
    \State Update: $\mathrm{CV}(b) \gets \mathrm{CV}(h) + \bigl[\sigma_i - \widehat{f}_{h,-i}(K_i)\bigr]^2$.
  \EndFor
  \State Compute the average error: $\mathrm{CV}(h) \gets \mathrm{CV}(b)/n$.
\EndFor
\State Select the pilot bandwidth: 
\[
\bar{h} \;=\; \arg\min_{h\in \mathcal{H}} \mathrm{CV}(h).
\]
\Return $\bar{h}$
\end{algorithmic}
\end{algorithm}

\paragraph{Bias and Variance Calculation}

we now show how to obtain practical values of bias and variance in equation \ref{eq:bias_polynomialselect},\ref{eq:var_polynomialselector}. For bias calculation, by using the pilot estimator, we have:
\begin{equation}
\hat{b}=\begin{pmatrix}
\hat\alpha_{p+1}S_{n,p+1}+\cdots+\hat\alpha_{p+a}S_{n,p+a} \\
\vdots \\
\hat\alpha_{p+1}S_{n,2p+1}+\cdots+\hat\alpha_{p+a}S_{n,2p+a}
\end{pmatrix},
\end{equation}
where the coefficients $\hat\alpha_{p+1},\ldots,\hat\alpha_{p+a}$ are now given by the pilot estimator.

For variance calculation, based on the $p$ order estimation $\hat{f}(\cdot)$ of the true IV function $f(\cdot)$, we compute the residuals as:

\begin{equation}
\hat{\varepsilon}_i=\sigma_i - \hat{f}(K_i).
\end{equation}

To determine the appropriate variance $\tau^2(K_i), i=1,\ldots,n$, the user must first choose between the homoscedasticity and heteroscedasticity assumptions. 

Under the local homoscedasticity assumption (i.e.,$\tau^2(K_i)=\tau^2$ for all $i$), we estimate the variance using a pseudo-Nadaraya–Watson estimator that we propose:

\begin{equation}
\hat{\tau}^2 = \frac{\sum_{j=1}^{n}\kappa\Bigl(\frac{K_j-k}{h}\Bigr)\hat{\varepsilon}_j^2}{\sum_{j=1}^{n}\kappa\Bigl(\frac{K_j-k}{h}\Bigr)}.
\end{equation}

Under the local heteroscedasticity assumption, where the variance may vary with $K_i$, the pseudo-Nadaraya–Watson estimator become: 

\begin{equation}
\hat{\tau}^2(K_i) = \frac{\sum_{j=1}^{n}\kappa\Bigl(\frac{K_j-K_i}{h}\Bigr)\hat{\varepsilon}_j^2}{\sum_{j=1}^{n}\kappa\Bigl(\frac{K_j-K_i}{h}\Bigr)}.
\end{equation}

Our estimators differs from the traditional Nadaraya-Watson estimator in several key aspects. While the classic Nadaraya-Watson estimator is primarily used to estimate the unknown quantity $f(\cdot)$, it is not designed for noise variance estimation. Moreover, our approach estimates the residuals using $\hat{\varepsilon}_j$, whereas the traditional Nadaraya-Watson method assumes access to the true residuals $\varepsilon_j$, which are unknown in our case. Despite these differences, we can show that this estimator converges almost surely to the true residual variance at a specific rate:

\begin{proposition}
\label{prop:tau convergence rate}
    Under assumption \ref{assup:convergence analysis}, the convergence rate of the estimator $\hat{\tau}(K)^2$ to the true variance $\tau^2(K)$ is:
\begin{align*}
|\hat{\tau}(K)^2 - \tau(K)^2| &= O_p\left(\frac{1}{\sqrt{nh}}\right) + O_p(h^{p+1}) + O(h^2) + O\left(\frac{1}{nh}\right),
\end{align*}
where $p \geq 0$.
\end{proposition}

\subsubsection{Data-Driven Optimal Bandwidth Selector}
\label{sec:Data-Driven Optimal Bandwidth Selector}

Beginning with a fixed polynomial order $p$ and strike $k$, we now search for the optimal bandwidth $h$ that minimizes the asymptotic conditional MSE. To determine it, we start by considering the theoretical form of asymptotically optimal bandwidth selection \citep{fan2018local}:

\begin{equation}
h_{\mathrm{opt}}\left(k\right)=C_{\nu, p}(\kappa)\left[\frac{\tau^2\left(k\right)}{\left\{f^{(p+1)}\left(k\right)\right\}^2 g\left(k\right)}\right]^{\frac{1}{2 p+3}} n^{-\frac{1}{2 p+3}},
\end{equation}

where

\begin{equation}
C_{\nu, p}(\kappa)=\left[\frac{(p+1)!^2(2 \nu+1) \int \kappa_\nu^{* 2}(t) d t}{2(p+1-\nu)\left\{\int t^{p+1} \kappa_\nu^*(t) d t\right\}^2}\right]^{\frac{1}{2 p+3}}
\end{equation}

 follows the same notation in Fan's book, and one can refer them in section \ref{proposition: cvp}. There are three quantities that must be estimated in this formula: the design density $g(\cdot)$, the higher-order derivative $f^{(p+1)}(\cdot)$, and the variance $\tau^2(\cdot)$. Traditionally, these quantities are difficult to estimate directly in local regression, limiting its practical applicability. 

In our framework, both $f^{(p+1)}(\cdot)$ and $\tau^2(\cdot)$ can be easily obtained from the pilot estimator. However, the estimation of the design density $g(\cdot)$ is more challenging. In many local regression applications, the independent variable follows a continuous distribution, allowing standard kernel density estimation (KDE) techniques to approximate the design density effectively. In contrast, the strike prices of standardized option contracts are inherently discrete and separated by fixed intervals. Consequently, their distribution can be viewed as a discrete uniform distribution. Applying KDE techniques to such a distribution typically yields a uniform distribution, which provides no useful information.

To overcome this limitation, we reinterpret each individual trade as a data point. In other words, rather than having a single observation per strike, the number of observations at a given strike is equal to its trading volume. This reinterpretation enables us to safely apply KDE techniques and obtain meaningful information about market participation and conviction. Moreover, the optimal bandwidth selection based on this interpretation method naturally aligns with common market practices: a larger (smaller) trading volume implies a shorter (longer) bandwidth, equivalent to assigning more (less) weight to the current strike.

With our specific estimation of those quantities in hand, we substitute them into the bandwidth formula, and are able to show later that this formula of selection bandwidth: 

\begin{equation}
\hat{h}_{\mathrm{opt}} \;=\; C_{v,p}\; \left[ \frac{\hat{\tau}^2(k)}{\bigl[\widehat{f}^{(p+1)}(k)\bigr]^2\,\hat{g}(k)} \right]^{\frac{1}{2p+3}} \, n^{-\frac{1}{2p+3}},
\end{equation}

is actually asymptotically optimal for minimal ACMSE.

Rather than relying on traditional 'rule-of-thumb' (ROT) methods—which depend on generic assumptions about function smoothness—our approach incorporates market volume information to achieve a more precise balance between bias and variance.

\begin{proposition}
\label{prop:hp are global optimal}
Under assumption \ref{assup:convergence analysis}, if for $n^{th}$ iteration in Algorithm \ref{alg:Optimal Denoising Parameter Selection}, $h_{n+1}=\hat{h}_{opt}$ and  $p_{n+1}$ is obtained following section \ref{sec:Data-Driven Optimal Polynomial Order Selector}, then $p_{n+1}, h_{n+1}$ are indeed the global optimal selectors for fixed $k$ and corresponding $h, p:$
\begin{enumerate}
    \item[] $p_{n+1}=\arg \min _{p \in \mathcal{P}} \hat{Z}_k\left(p, h_n\right)$,
    \item[] $h_{n+1}=\arg \min _h \hat{Z}_k\left(p_{n+1}, h\right)$,
\end{enumerate}
where $\mathcal{P}$ is the sample space for polynomial order selector, and $\hat{Z}_k$ is asymptotic conditional MSE for fixed $k$.
\end{proposition}

\subsubsection{From Asymptotic Conditional MSE to True MSE}
\label{sec:From Asymptotic Conditional MSE to True MSE}

With appropriate assumptions and conditions, we prove that the minimum asymptotic conditional MSE $\hat{Z}^*$ obtained in our framework converges to the true minimum MSE $Z^*$. This result provides a strong theoretical foundation for the efficiency and practical utility of our automatic local regression framework, thereby validating its use for optimal denoising in LV model implementation.

\label{sec: IV Denoising theory conclusion}
\begin{proposition}
\label{prop:ACMSE to TRUE MSE}
    Under assumption \ref{assup:convergence analysis} and the same search space $Ph=\{(p,h)\mid p \in \mathcal{P}\}$, let's denote $Z_k^*$ as the optimal true MSE value under unknown quantities $f^{(p+1)}(\cdot), g(\cdot),\tau(\cdot)$. Then one can show the following convergence result:
    $$\lim _{\substack{n \rightarrow \infty \\ n h \rightarrow \infty \\ h \rightarrow 0}} \hat{Z}_k^*=Z_k^*.$$
\end{proposition}

\subsection{Stage 2: From Smooth IV to Arbitrage-free LV}

The second stage transforms the smoothed IVs from Stage 1 into an arbitrage-free LV surface. The key insight is that smoothed input naturally lead to smooth LV surfaces through the Dupire formula, as the LV model depends on partial derivatives of observable market prices. By starting with optimally smoothed data, we avoid the spurious oscillations that typically plague local volatility calibration.

\subsubsection{From Smooth IV to Smooth Prices}

Having obtained the denoised implied volatility surface $\hat{\sigma}_{IV}(K,T)$ from Stage 1, we first convert these values to option prices using the Black-Scholes formula:

\begin{equation}
\hat{C}(K,T) = BS(S_0, K, T, r, d, \hat{\sigma}_{IV}(K,T)),
\end{equation}

where $BS(\cdot)$ denotes the Black-Scholes pricing formula. This transformation preserves the smoothness properties: since the Black-Scholes formula is smooth in the implied volatility parameter, our optimally smoothed implied volatilities yield correspondingly smooth option prices. These smooth prices serve as the target for our local volatility calibration.

\subsubsection{Local Volatility Calibration via Finite Differences}

The Dupire formula establishes the fundamental relationship between option prices and local volatility:

\begin{equation}
\frac{\partial C}{\partial t} = \frac{1}{2}{\sigma_{LV}^2}(t,K)K^2 \frac{\partial^2 C}{\partial K^2} - (r-d)K\frac{\partial C}{\partial K}-dC,
\end{equation}

where $C$ is the option price, $K$ is the strike price, $t$ is the time, $r$ is the risk-free rate, $d$ is the dividend rate, and $\sigma_{LV}^2(t,K)$ represents the local volatility function.

To ensure numerical stability and maintain smoothness throughout the calibration, we work with normalized variables. Define the forward price $F(t,T) = S_0e^{(r-d)(T-t)}$, normalized strike $k(t) = K(t)/F(0,t)$, and normalized price $\tilde{C}(t,k(t)) = C(t,K(t))/F(0,t)$. The Dupire formula then simplifies to:

\begin{equation}
\frac{\partial \tilde{C}}{\partial t} = \frac{1}{2}{\sigma_{LV}^2}(t,k)k^2 \frac{\partial^2 \tilde{C}}{\partial k^2}.
\end{equation}

This normalized form eliminates the drift term and improves numerical conditioning, particularly important for maintaining smoothness in the discretized system.

\subsubsection{Optimization Framework}

We implement the local volatility calibration through an implicit finite difference scheme that naturally enforces arbitrage-free conditions \citep{andreasen2010volatility}. The key is to formulate the problem inversely: rather than computing prices forward given local volatilities, we determine the local volatilities that best reproduce our smooth target prices.

At each time step, we solve the optimization problem:

\begin{equation}
\sigma_{LV}^{i-1}(k) = \arg\min_{\sigma \geq 0} \sum_{j} \left(C_i^{\text{smooth}}(k_j) - C_i^{\text{model}}(k_j; \sigma)\right)^2,
\end{equation}

subject to the finite difference evolution:

\begin{equation}
C_i^{\text{model}} = f_{\text{FD}}(C_{i-1}, \sigma_{LV}^{i-1}),
\end{equation}

where $f_{\text{FD}}$ represents the implicit finite difference operator. The non-negativity constraint on $\sigma$ ensures arbitrage-free conditions, as negative local volatilities would violate the no-arbitrage principle.

\subsubsection{Algorithm Implementation}

The complete algorithm proceeds backward in time from the longest maturity to ensure consistency across the entire surface:

\begin{algorithm}[H]
\caption{Finite-Difference Local Volatility Calibration}\label{alg:lv_calibration}
\begin{algorithmic}
\State \textbf{Input:} Smooth prices $\{C^{\text{smooth}}(k_j, T_i)\}$ from Stage 1
\State \textbf{Initialize:} $C_0(k) = \max(S_0 - k, 0)$ \Comment{Terminal payoff}
\For{$i = 1$ to $N$} \Comment{Iterate through maturities}
    \State Solve: $\sigma_{LV}^{i-1} = \arg\min_{\sigma \geq 0} \sum_j \left(C_i^{\text{smooth}}(k_j) - C_i^{\text{model}}(k_j)\right)^2$
    \State Update: $C_i = f_{\text{FD}}(C_{i-1}, \sigma_{LV}^{i-1})$ \Comment{Finite difference step}
\EndFor
\State \textbf{Output:} Local volatility surface $\{\sigma_{LV}(k, T)\}$
\end{algorithmic}
\end{algorithm}

The smoothness of the input prices from Stage 1 ensures that the optimization problem at each time step is well-conditioned. Since the finite difference operator preserves smoothness properties and the optimization respects the smooth structure of the target prices, the resulting local volatility surface inherits the smoothness characteristics of our preprocessed data. This cascading of smoothness from implied volatilities through prices to local volatilities is the key mechanism by which our two-stage approach achieves stable Greeks for exotic option pricing.

\section{Numerical Experiments}

In this section, we demonstrate the effectiveness of our proposed pre-calibration smoothing method through comprehensive numerical experiments conducted on two distinct datasets: model-generated data based on the SVI parameterization (referred to as the SVI market), and real-world market data characterized by a W-shaped IV surface (referred to as the real-world W-shaped market). These datasets were selected to evaluate the method's performance across both controlled model-generated environments and challenging real-world scenarios.

\subsection{The SVI Market}

In the first part of numerical experiments, we conduct an experiment using the SVI market. In an ideal market, the IV function is expected to be both continuous and smooth. However, real-world markets provide only discrete market observables. Moreover, those observables are subject to various sources of random noise introduced by market imperfections. These imperfections cause the observed IVs to deviate from their true values. As demonstrated in Figure \ref{fig:SVI_IV}, to simulate this scenario, we consider an arbitrage-free SVI curve using the following parameters: $a=0.030358$, $b=0.0503815$, $\rho=-0.1$, $m=0.3$, and $\sigma=0.048922$. The time to maturity is one year, the underlying stock price $S_0$ is $1$, and both the risk-free rate and dividend rate are zero. For our numerical experiments, we first take a discrete sequence of IVs from this SVI curve (referred to as the ideal data) and then add i.i.d normally distributed noise with $\mu = 0$ and $\sigma = 0.001$ to the ideal data (referred to as the noisy data). For reproducibility, we set the random seed to $11$.

\begin{figure}[H]
\centering
\begin{minipage}{0.45\textwidth}
  \centering
  \includegraphics[width=\linewidth]{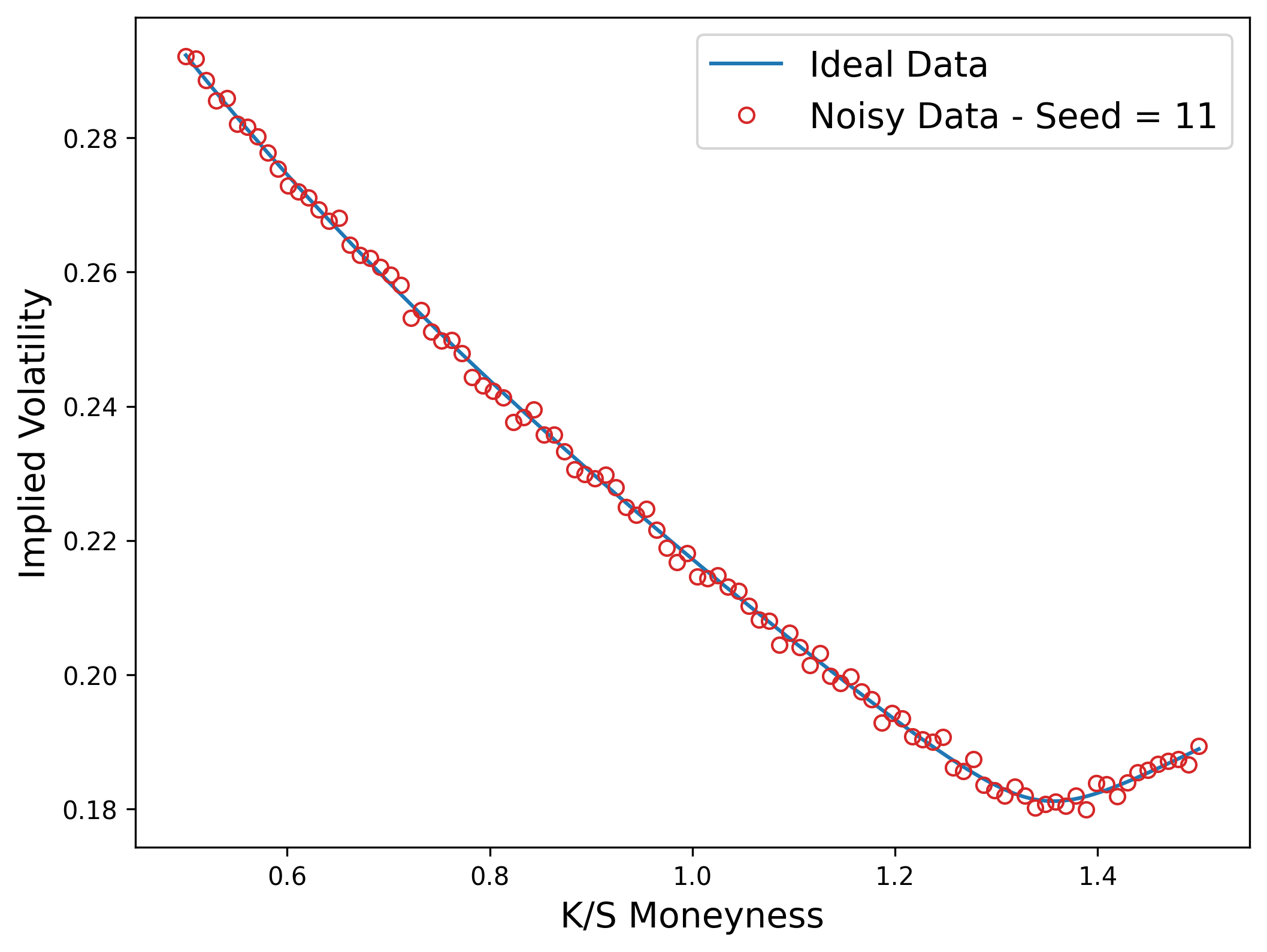}
\end{minipage}
\caption{SVI IV Curve}
\label{fig:SVI_IV}
\end{figure}

Through the Black-Scholes formula, we calculated the corresponding European call option prices from both ideal and noisy data. In Figure \ref{fig:SVI_Market_C_kk_Ideal} and \ref{fig:SVI_Market_C_kk_Noisy}, we compared the second-order strike derivative, $\frac{\partial^2 C}{\partial k^2}$, calculated from those prices. As an important component in the Dupire formula, a stable $\frac{\partial^2 C}{\partial k^2}$ is of vital importance for a stable and smooth LV function. However, as we demonstrated, even minimal noise in market observables can lead to substantial fluctuations in the computed partial derivatives. 

Moreover, without proper pre-calibration smoothing, the calibrated LV function will try its best to match this instability. This is because most of the objective functions are solely focused on the price difference, which means the LV function whose price function's $\frac{\partial^2 C}{\partial k^2}$ is closest to this spiky one is preferred. We will further demonstrate this problem later.

\begin{figure}[H]
\centering
\begin{minipage}{0.45\linewidth}
  \centering
  \includegraphics[width=\linewidth]{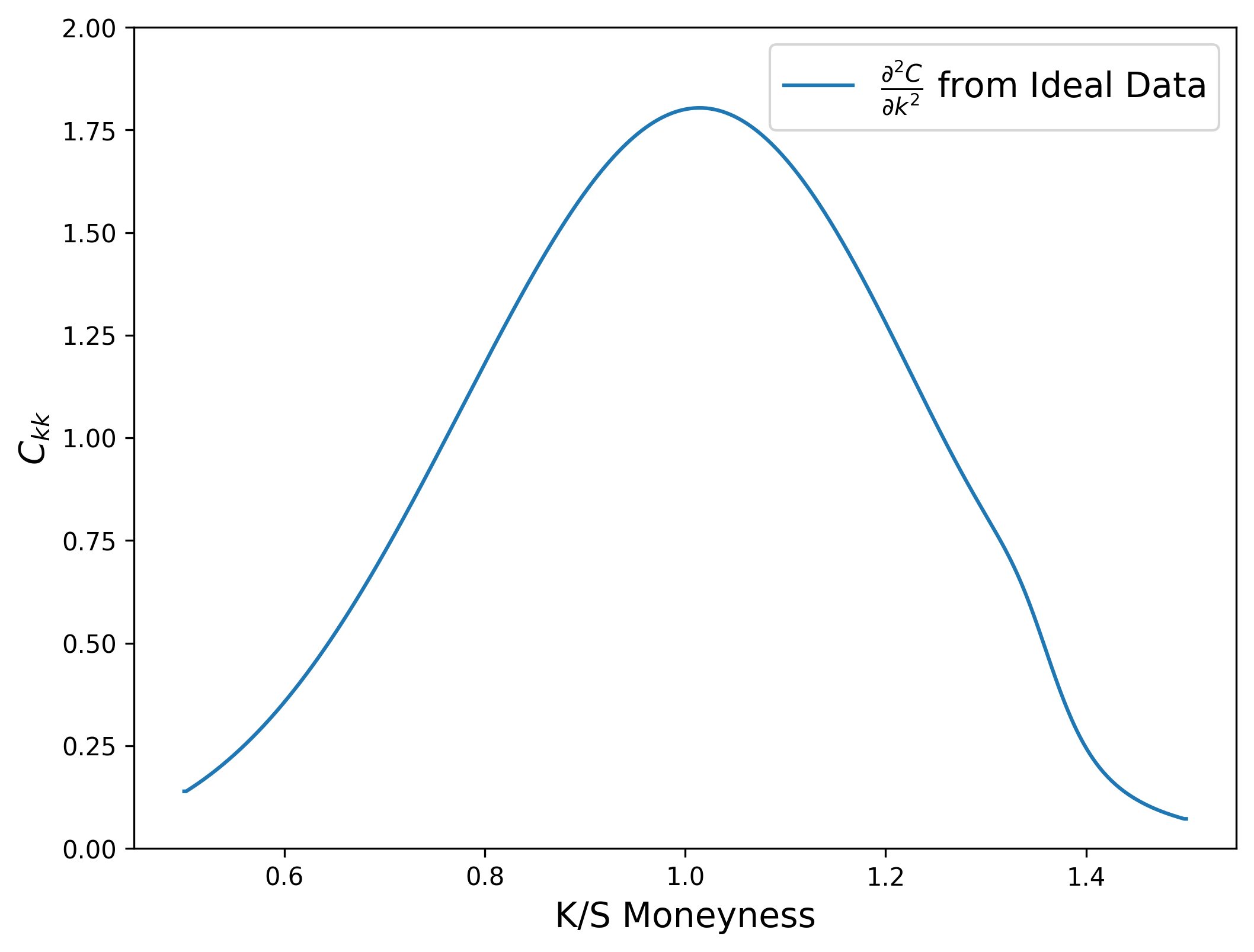}
  \caption{Market Price $\frac{\partial^2 C}{\partial k^2}$ from Ideal Data}
  \label{fig:SVI_Market_C_kk_Ideal}
\end{minipage}
\hfill
\begin{minipage}{0.45\linewidth}
  \centering
  \includegraphics[width=\linewidth]{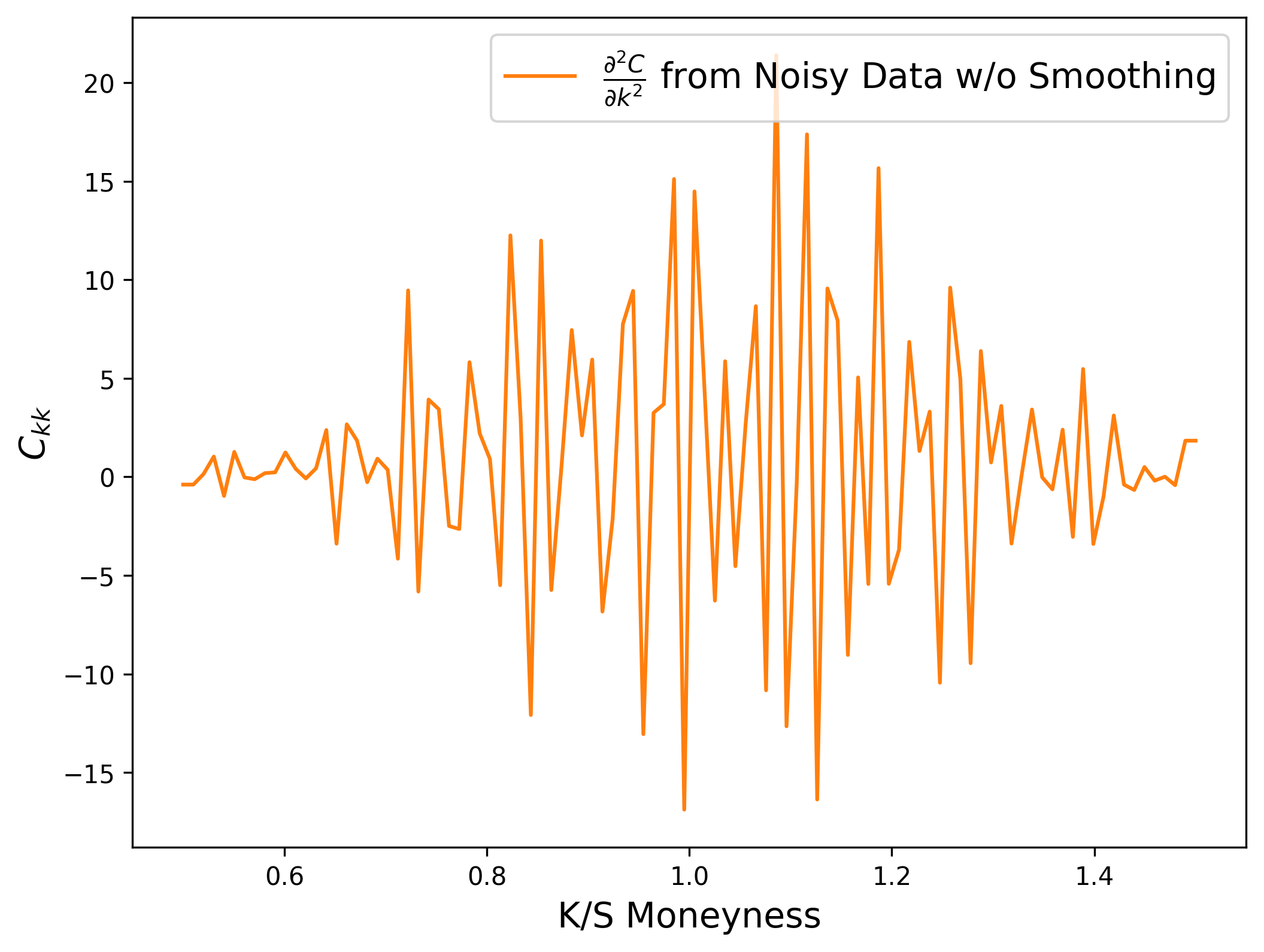}
  \caption{Market Price $\frac{\partial^2 C}{\partial k^2}$ from Noisy Data}
  \label{fig:SVI_Market_C_kk_Noisy}
\end{minipage}
\end{figure}

In the following paragraphs, we first apply a standard F-D based LV fitting method directly to both the ideal and noisy data and illustrate the resulting issues. We then apply our proposed method prior to the F-D fitting step, demonstrating how it effectively resolves those problems, and eventually lead to stable Greeks.

\subsubsection{Calibrate with only the F-D Based LV Fitting Method}

We begin by applying the F-D based LV fitting method to both the ideal and noisy data. Our numerical results in Figure \ref{fig:SVI_LV_FD_one_seed} show that the LV function generated from the noisy data inherits the fluctuations in the partial derivatives, becoming spiky. In contrast, the LV function generated from the ideal data remains smooth.

\begin{figure}[H]
\centering
\begin{minipage}{0.45\textwidth}
  \centering
  \includegraphics[width=\linewidth]{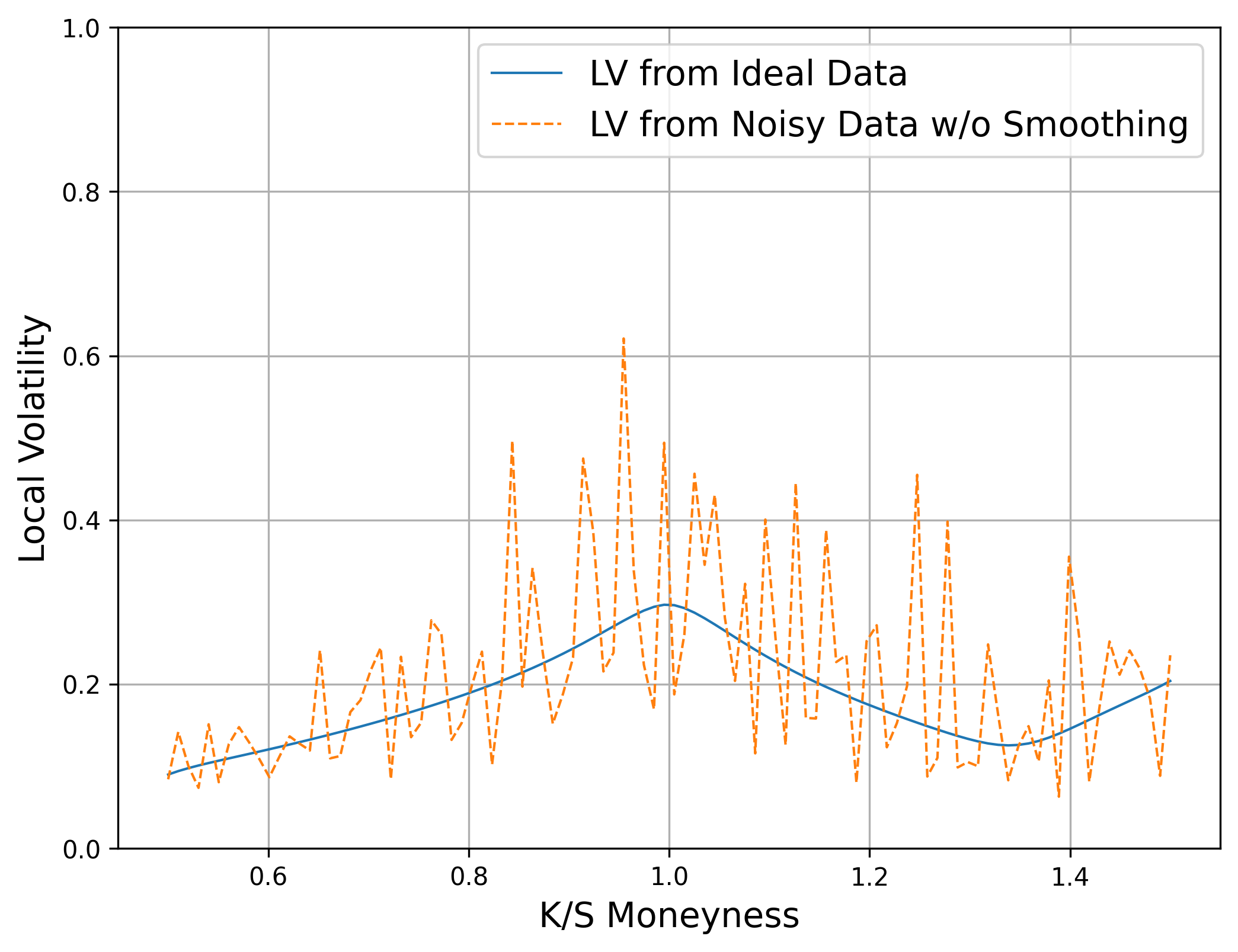}
\end{minipage}
\caption{SVI LV Curve}
\label{fig:SVI_LV_FD_one_seed}
\end{figure}

To quantify the discrepancy between the model and the market, we calculate the calibration error as the average absolute relative difference between the model IV and the market IV, where the model IV is calculated using the calibrated LV function to price European call options and then solve for the IV:

\begin{equation}
\text{Calibration Error} = \frac{1}{N} \sum_{i=1}^{N} \left| \frac{\sigma_i^{\text{model}} - \sigma_i^{\text{market}}}{\sigma_i^{\text{market}}} \right| \times 100\%
\end{equation}

Notably, despite the spikes observed in the LV function from the noisy data, the overall calibration error remains small, and the arbitrage-free condition is still satisfied. This suggests that a small calibration error alone does not guarantee a smooth LV function. In other words, solely focusing on minimizing the price discrepancy is insufficient to ensure smoothness, as the generated LV function may still exhibit undesirable spikes.

\begin{table}[H]
\centering
\textbf{Absolute Calibration Error in \% for Noisy Data} \\[1ex]
\begin{tabular}{|c|c|c|c|}
\hline
             & Moneyness $< 0.95$ &  $0.95 <$ Moneyness $< 1.05$ & $1.05<$ Moneyness \\ \hline
Error in \%  &          0.21\%      &                   0.34\%         &               0.17\% \\ \hline
\end{tabular}
\caption {Absolute Calibration Error in \% for Noisy Data}
\end{table}

Moreover, as demonstrated in Figure \ref{fig:SVI_LV_FD_two_seed}, by altering the random seed of the noisy data, and repeat the LV calibration process, we further demonstrate that the LV function generated through F-D based LV fitting method is unstable when small changes occur in the input data. We conclude that, in the presence of random market noise, directly applying this method can lead to unpredictable and undesirable spikes in the LV function. 

\begin{figure}[H]
\centering
\begin{minipage}{0.45\textwidth}
  \centering
  \includegraphics[width=\linewidth]{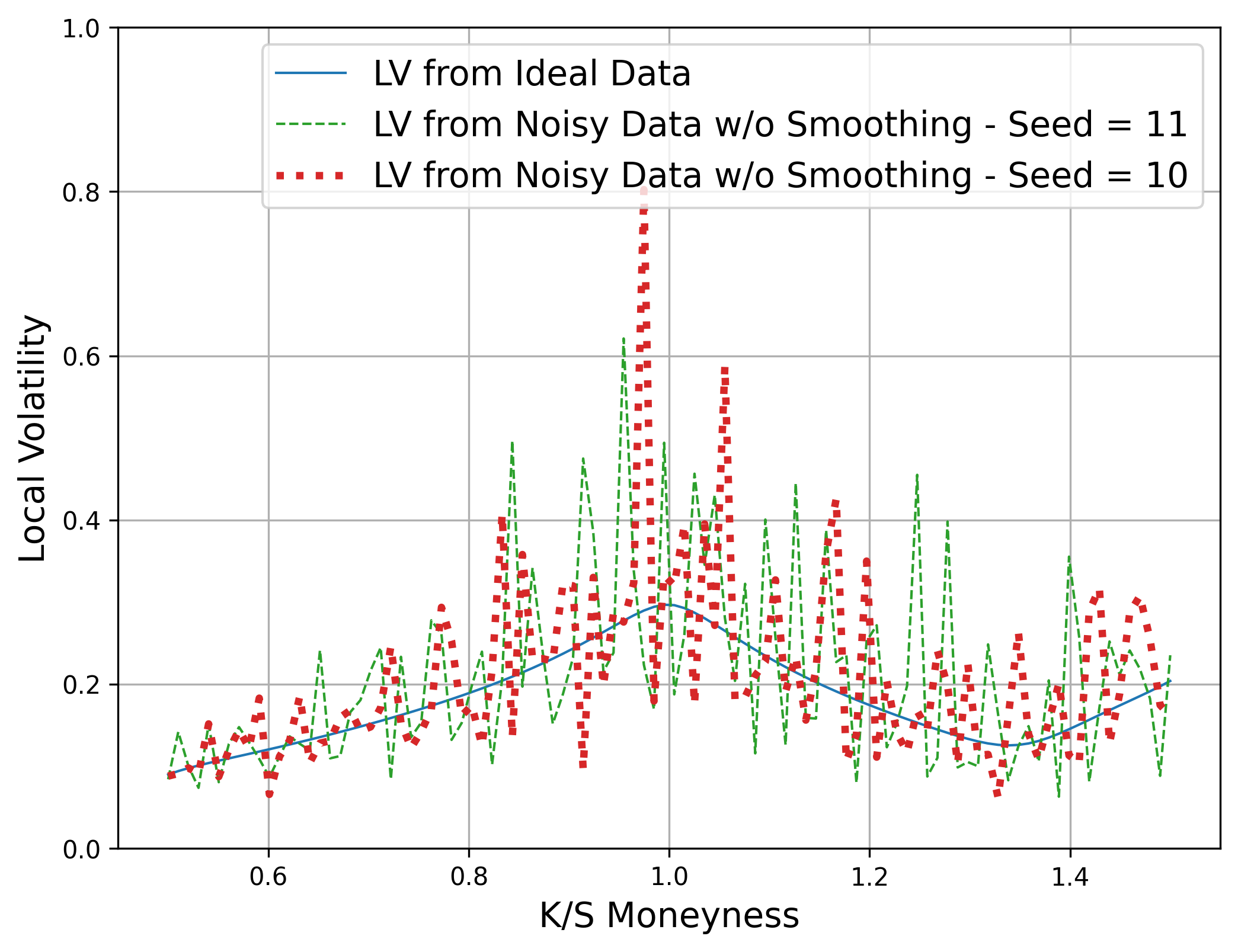}
\end{minipage}
\caption{SVI LV Curve with Different Random Seed}
\label{fig:SVI_LV_FD_two_seed}
\end{figure}

As discussed earlier, without proper pre-calibration smoothing, the instability in market observable will be inherited by the calibrated LV function. Here we price the European call options using both the spiky and smooth LV function, and then calculate the $\frac{\partial^2 C}{\partial k^2}$. As we demonstrated in Figure \ref{fig:SVI_Model_C_kk_Ideal} and \ref{fig:SVI_Model_C_kk_Noisy}, the $\frac{\partial^2 C}{\partial k^2}$ from spiky LV function's model price function inherit the instability in Figure \ref{fig:SVI_Market_C_kk_Noisy}. 

\begin{figure}[H]
\centering
\begin{minipage}{0.45\linewidth}
  \centering
  \includegraphics[width=\linewidth]{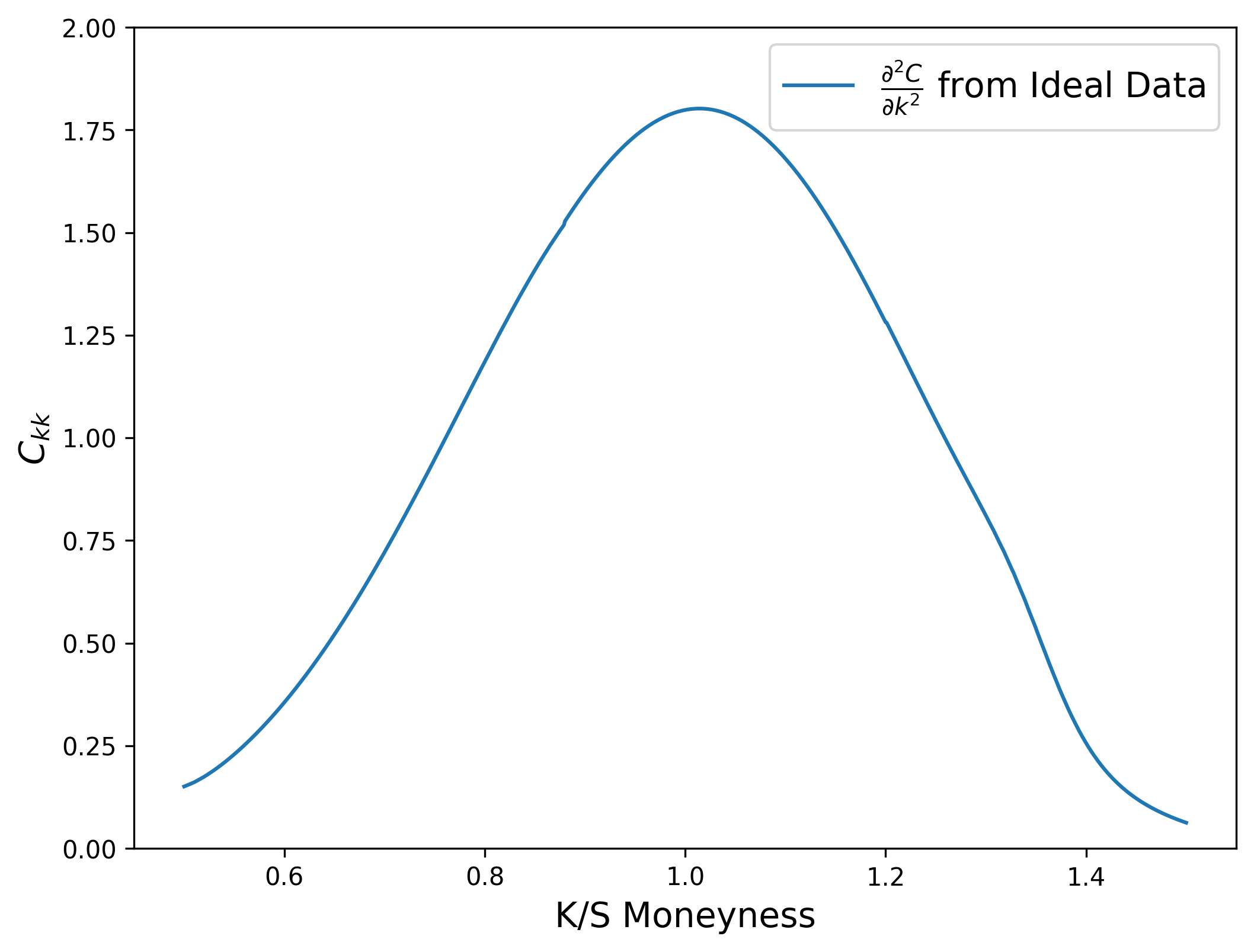}
  \caption{Model Price $\frac{\partial^2 C}{\partial k^2}$ from Ideal Data}
  \label{fig:SVI_Model_C_kk_Ideal}
\end{minipage}
\hfill
\begin{minipage}{0.45\linewidth}
  \centering
  \includegraphics[width=\linewidth]{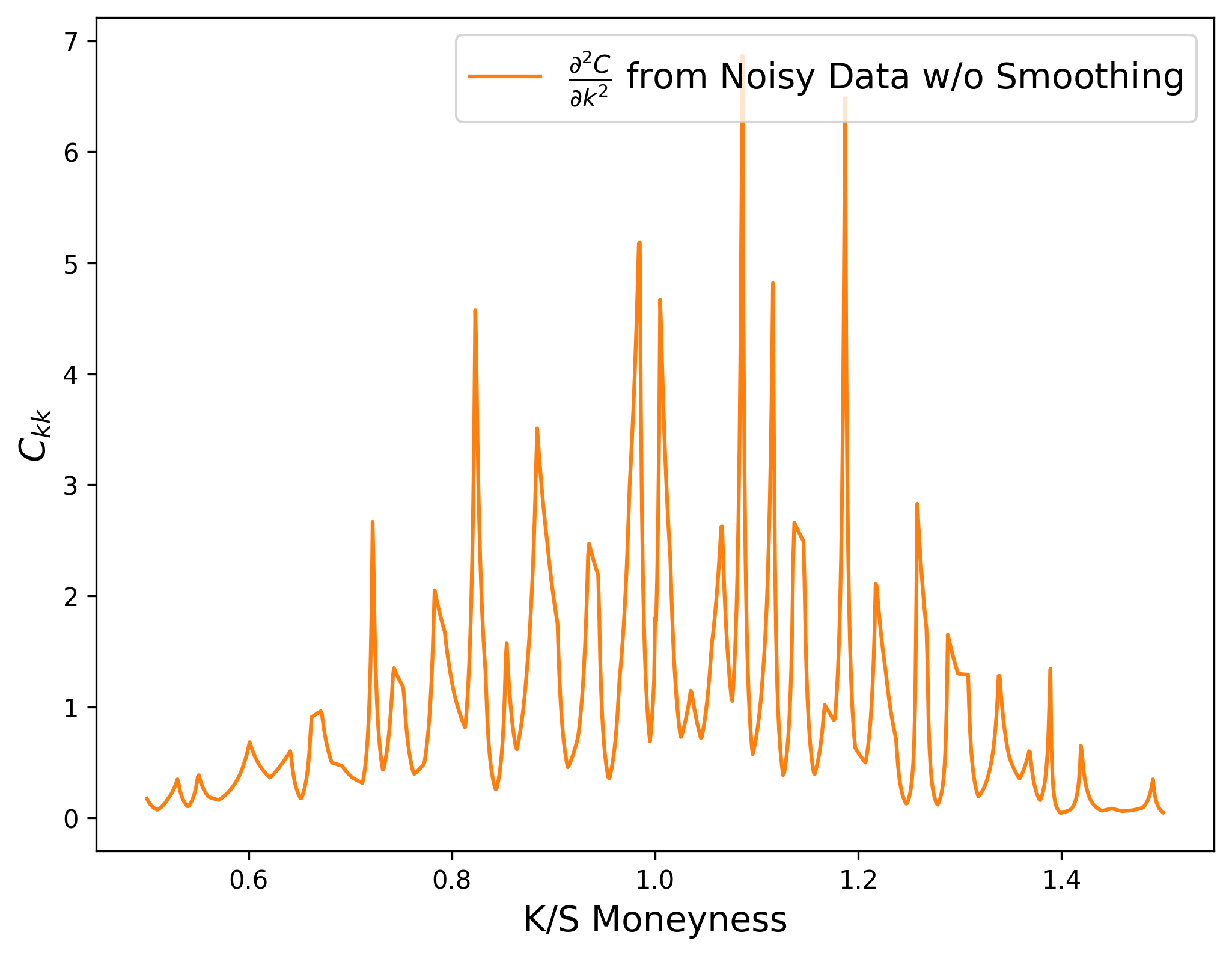}
  \caption{Model Price $\frac{\partial^2 C}{\partial k^2}$ from Noisy Data}
  \label{fig:SVI_Model_C_kk_Noisy}
\end{minipage}
\end{figure}

Furthermore, to illustrate the unreliable Greeks problem, we calculate the Delta and Gamma using both ideal and noisy data under the sticky-strike assumption with different spot price levels for a selected European call option. Specifically, we here assume the spot price $S_0$ varies from $0.5$ to $1.5$, while the strike is fixed at $1$. For every spot level, we recalibrate the corresponding LV function using ideal and noisy data, and calculate the Greeks using the generated LV functions, respectively. As demonstrated in Figure \ref{fg:SVI_Delta_no_smoothing} and \ref{fg:SVI_Gamma_no_smoothing}, without proper pre-calibration smoothing, the Gamma value from the noisy data is unreliable compared to the Gamma value from the ideal data. Though the Delta value is relatively reliable due to the simplicity of the contract.

\begin{figure}[H]
\centering
\begin{minipage}{0.45\linewidth}
  \centering
  \includegraphics[width=\linewidth]{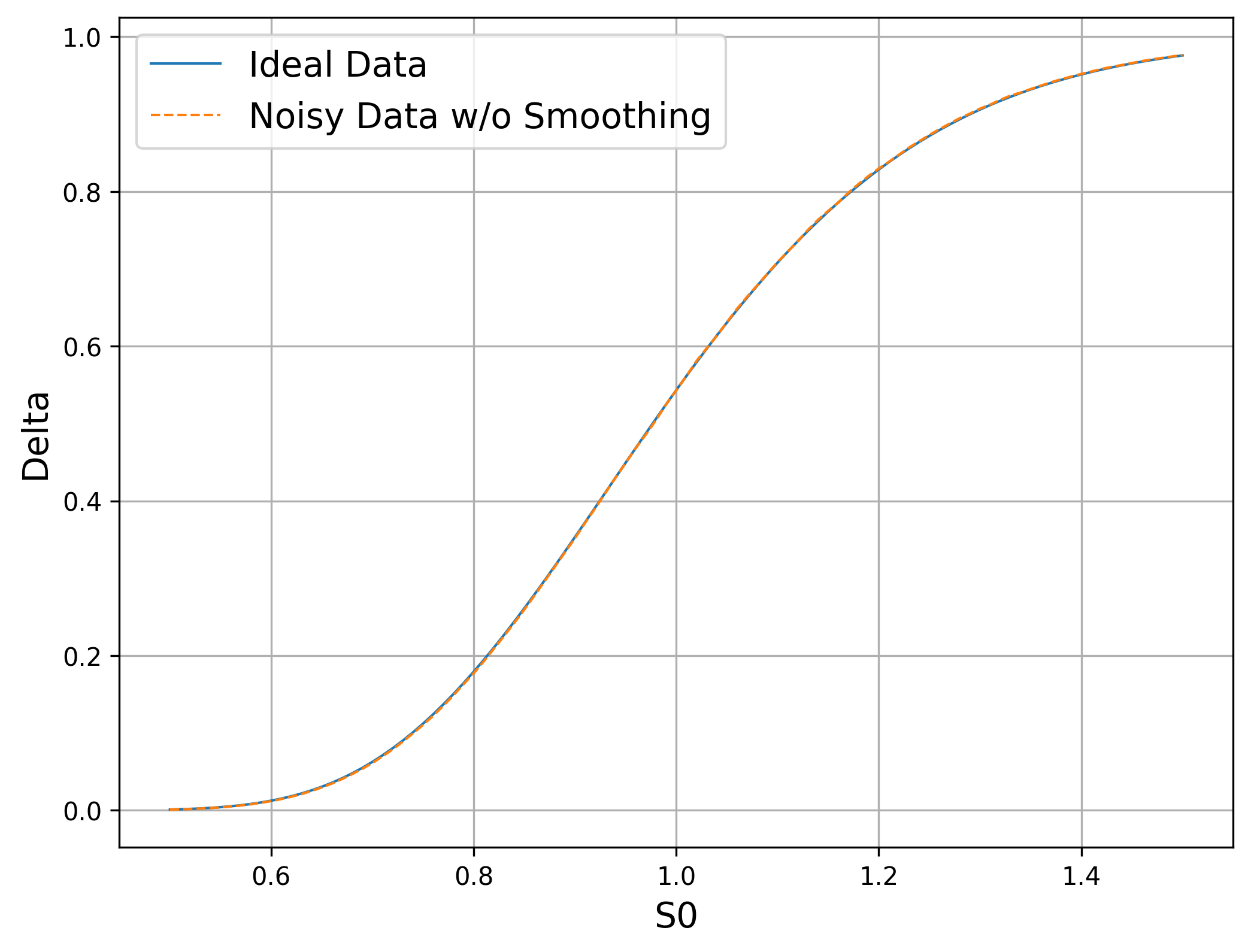}
  \caption{Delta Value}
  \label{fg:SVI_Delta_no_smoothing}
\end{minipage}
\hfill
\begin{minipage}{0.45\linewidth}
  \centering
\includegraphics[width=\linewidth]{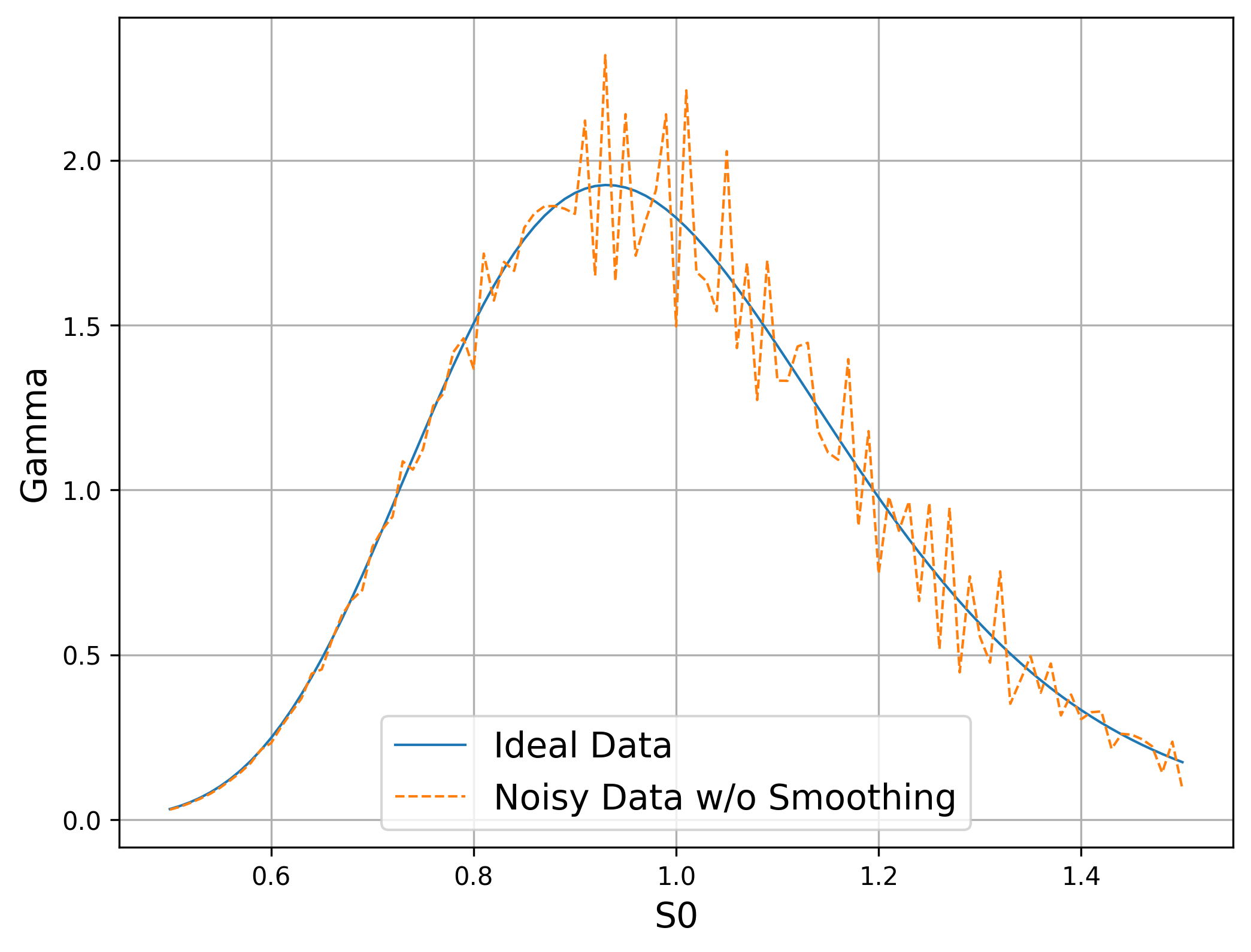}
  \caption{Gamma Value}
  \label{fg:SVI_Gamma_no_smoothing}
\end{minipage}
\end{figure}

\subsubsection{Calibrate with the Proposed Method}

Using the identical SVI market data described above, we demonstrate the effectiveness of our proposed methodology through a systematic two-stage process. First, we apply our pre-calibration smoothing technique to process the noisy data. Second, we pass the denoised data to a standard F-D based LV fitting method to generate a smooth and robust LV function while preserving both accuracy and arbitrage-free conditions. We then conduct a comprehensive comparison between the model price functions and Greeks obtained through our method against their true values derived from ideal data, demonstrating that our approach successfully filters out noise and generates a desirable arbitrage-free LV function that ensures a stable model price function and Greeks.

The effectiveness of our proposed method is clearly evidenced through examination of the second-order strike derivatives ($\frac{\partial^2 C}{\partial k^2}$) of the European call options. Using the same setup as for Figure \ref{fig:SVI_Market_C_kk_Ideal} and \ref{fig:SVI_Market_C_kk_Noisy}, we compute these numerical derivatives from both the original noisy data and the data processed through our proposed method. As demonstrated in Figure \ref{fig:SVI_Ckk_Proposed_one_seed}, the comparison reveals that the results obtained after applying our method exhibit significantly improved stability and demonstrate close alignment with the theoretical partial derivatives computed from ideal data.

\begin{figure}[H]
\centering
\begin{minipage}{0.45\textwidth}
  \centering
  \includegraphics[width=\linewidth]{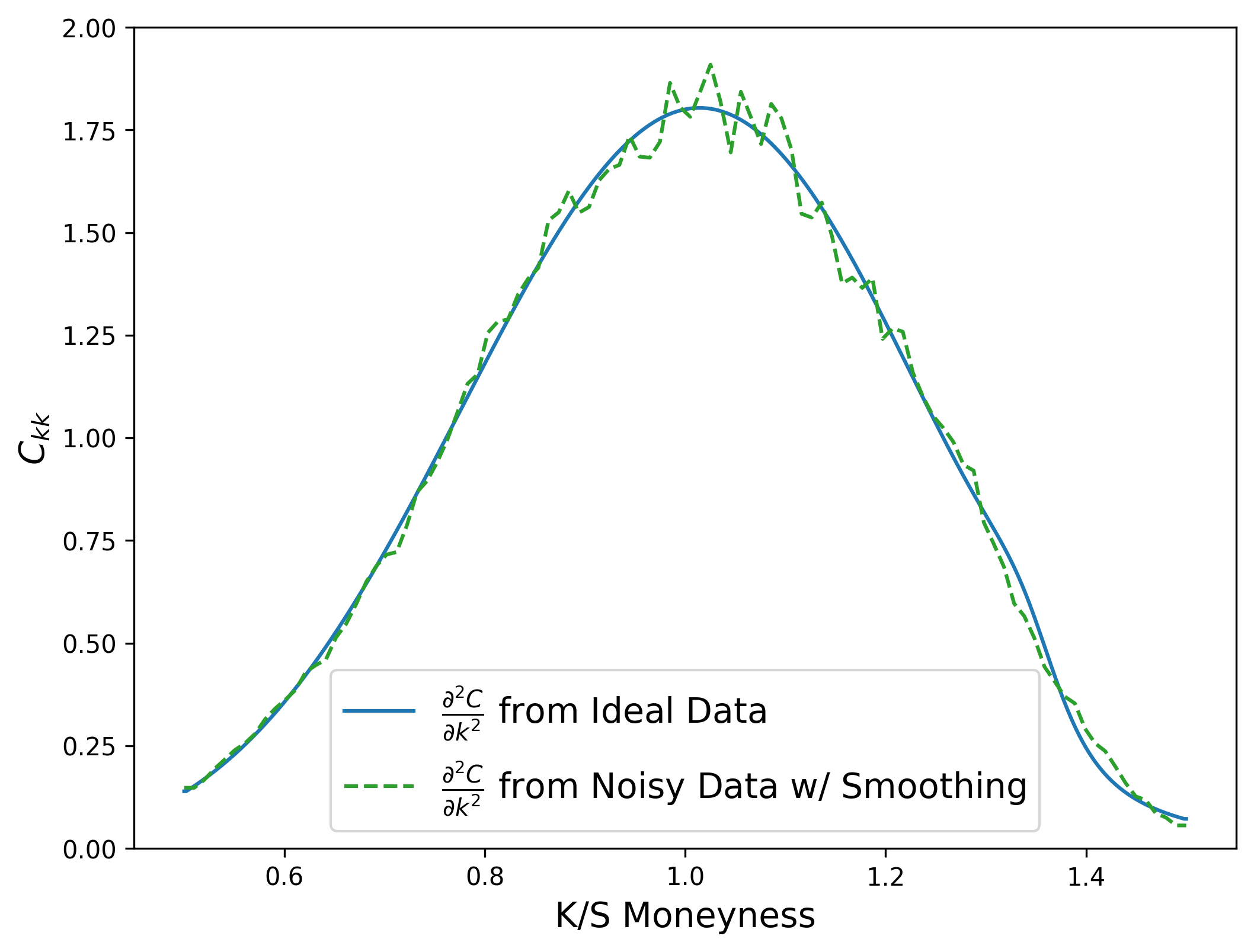}
\end{minipage}
\caption{Market Price $\frac{\partial^2 C}{\partial k^2}$ from Noisy Data with Proposed Method}
\label{fig:SVI_Ckk_Proposed_one_seed}
\end{figure}

We then apply the F-D based LV fitting method to find the optimal LV function whose model price function best approximates the processed noisy data. As demonstrated in Figure \ref{fig:SVI_LV_Proposed_one_seed}, the generated LV function is much smoother and closer to the LV function from the ideal data. Additionally, the calibration accuracy is also improved. 

\begin{figure}[H]
\centering
\begin{minipage}{0.45\textwidth}
  \centering
  \includegraphics[width=\linewidth]{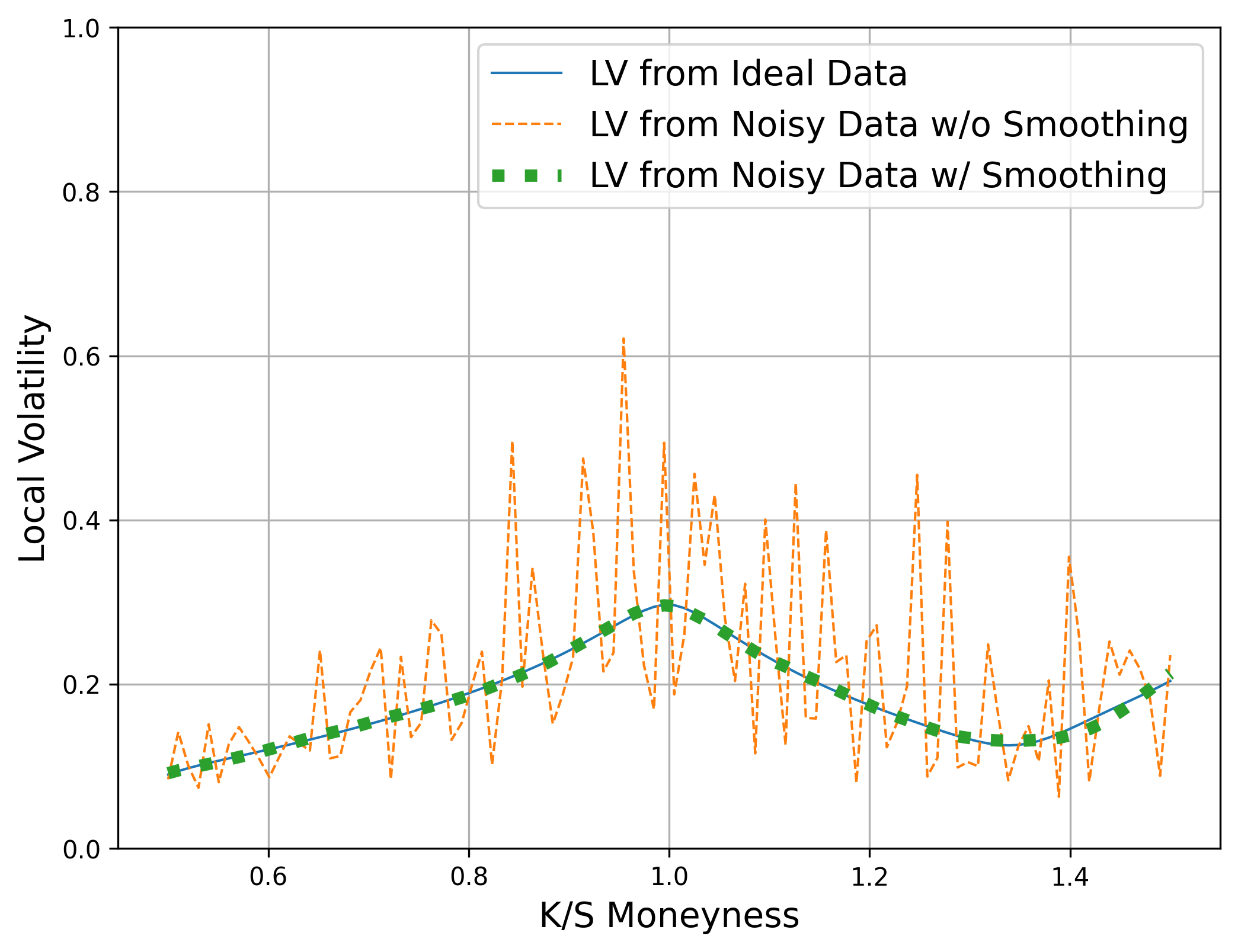}
\end{minipage}
\caption{SVI LV Curve}
\label{fig:SVI_LV_Proposed_one_seed}
\end{figure}

\begin{table}[H]
\centering
\textbf{Absolute Calibration Error in \% for Noisy Data with Proposed Method} \\[1ex]
\begin{tabular}{|c|c|c|c|}
\hline
             & Moneyness $< 0.95$ &  $0.95 <$ Moneyness $< 1.05$ & $1.05<$ Moneyness \\ \hline
Error in \%  &          0.11\%      &                   0.18\%         &               0.12\% \\ \hline
\end{tabular}
\caption {Absolute Calibration Error in \% for Noisy Data with Proposed Method}
\end{table}

We also apply our method to two sets of noisy data generated with different random seeds. As demonstrated in Figure \ref{fig:SVI_LV_Proposed_two_seed} our numerical result shows that the generated LV functions are similar to each other, which emphasizes the robustness of our method for LV model implementation.

\begin{figure}[H]
\centering
\begin{minipage}{0.45\textwidth}
  \centering
  \includegraphics[width=\linewidth]{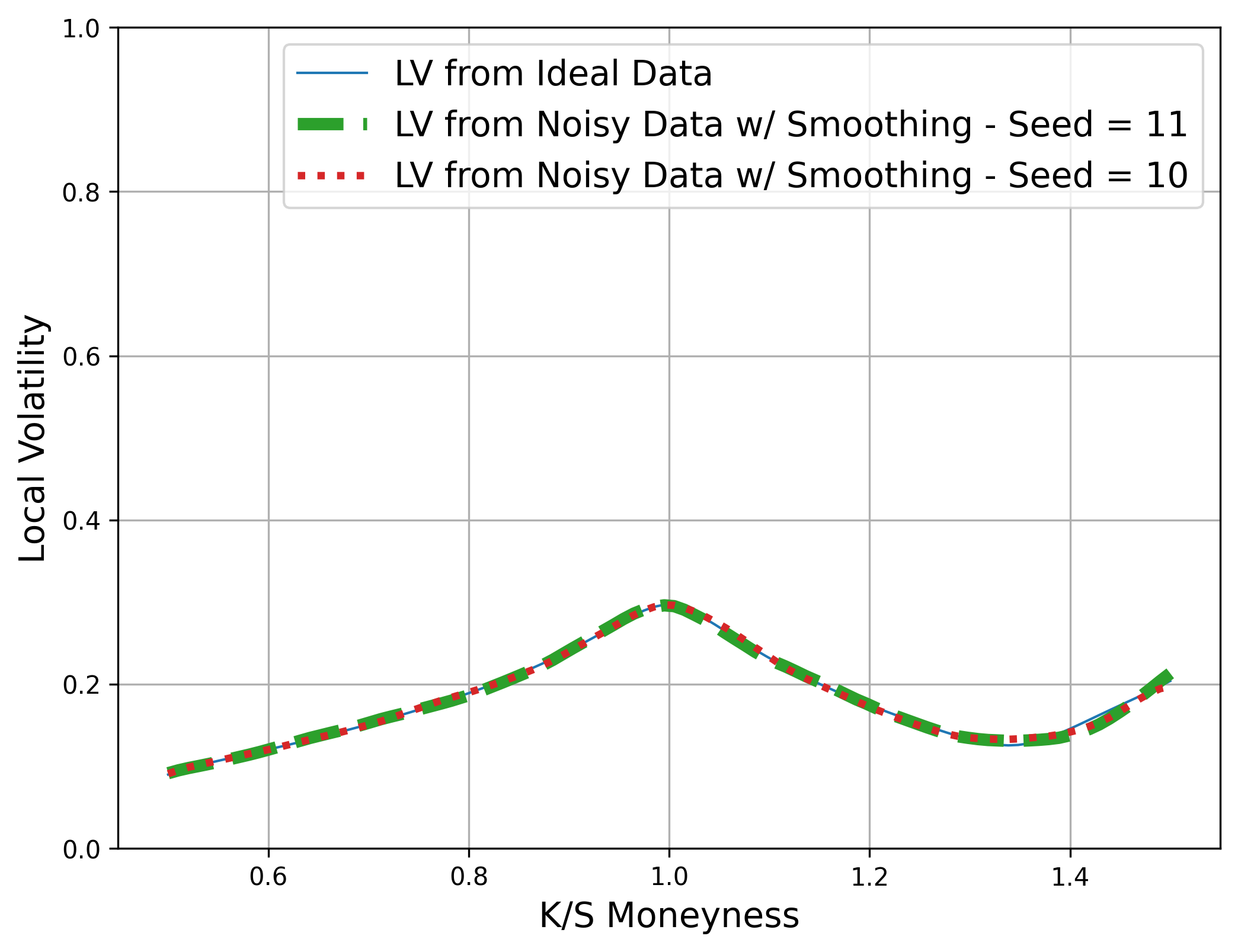}
\end{minipage}
\caption{SVI LV Curve with Different Seeds}
\label{fig:SVI_LV_Proposed_two_seed}
\end{figure}

We now investigate the numerical stability issues of the model price function and the Greeks. As we already demonstrated in the previous numerical experiments, our method efficiently reduces the noise-induced irregularities in the generated LV function and can generate a smooth LV function. Based on this smooth LV function, we now produce the $\frac{\partial^2 C}{\partial k^2}$ using the same setup as for Figure \ref{fig:SVI_Model_C_kk_Ideal} and \ref{fig:SVI_Model_C_kk_Noisy}. We show that by applying our method, as demonstrated in Figure \ref{fig:output_C_kk_Proposed}, we can now get a stable and close to true value $\frac{\partial^2 C}{\partial k^2}$ from our market price function.

\begin{figure}[H]
\centering
\begin{minipage}{0.45\textwidth}
  \centering
  \includegraphics[width=\linewidth]{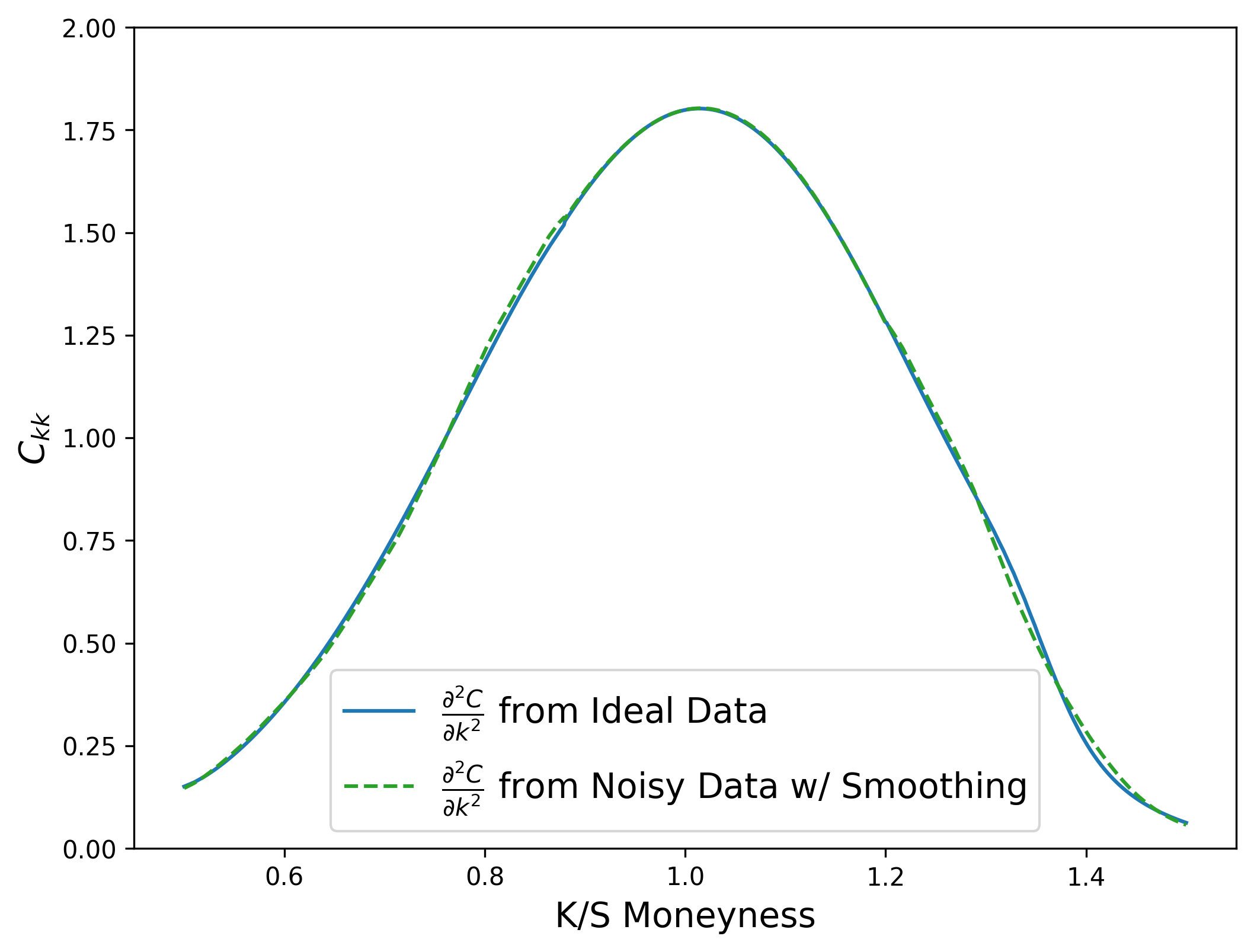}
\end{minipage}
  \caption{Model Price $\frac{\partial^2 C}{\partial k^2}$ from Noisy Data with Proposed Method}
\label{fig:output_C_kk_Proposed}
\end{figure}

Furthermore, to illustrate the unreliable Greeks problem, we calculate the Delta and Gamma using the same setup as for Figure \ref{fg:SVI_Delta_no_smoothing} and \ref{fg:SVI_Gamma_no_smoothing}. As demonstrated in Figure \ref{fg:SVI_Delta_smoothing} and \ref{fg:SVI_Gamma_smoothing}, both our method and the direct LV fitting method produce reliable Delta. However, only our proposed method yields a reliable Gamma. This finding highlights the importance of our proposed method.

\begin{figure}[H]
\centering
\begin{minipage}{0.45\linewidth}
  \centering
  \includegraphics[width=\linewidth]{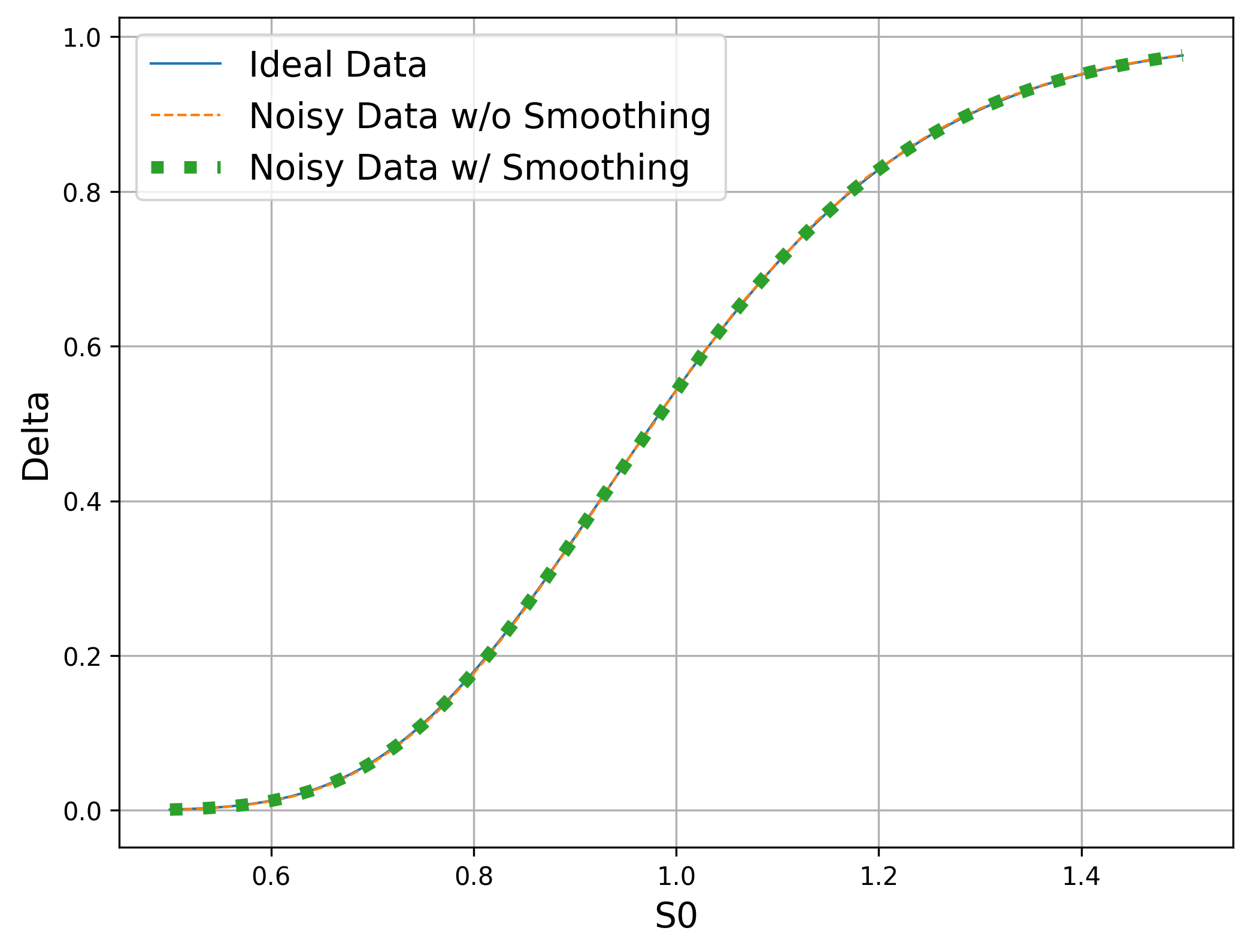}
  \caption{Delta Value}
  \label{fg:SVI_Delta_smoothing}
\end{minipage}
\hfill
\begin{minipage}{0.45\linewidth}
  \centering
  \includegraphics[width=\linewidth]{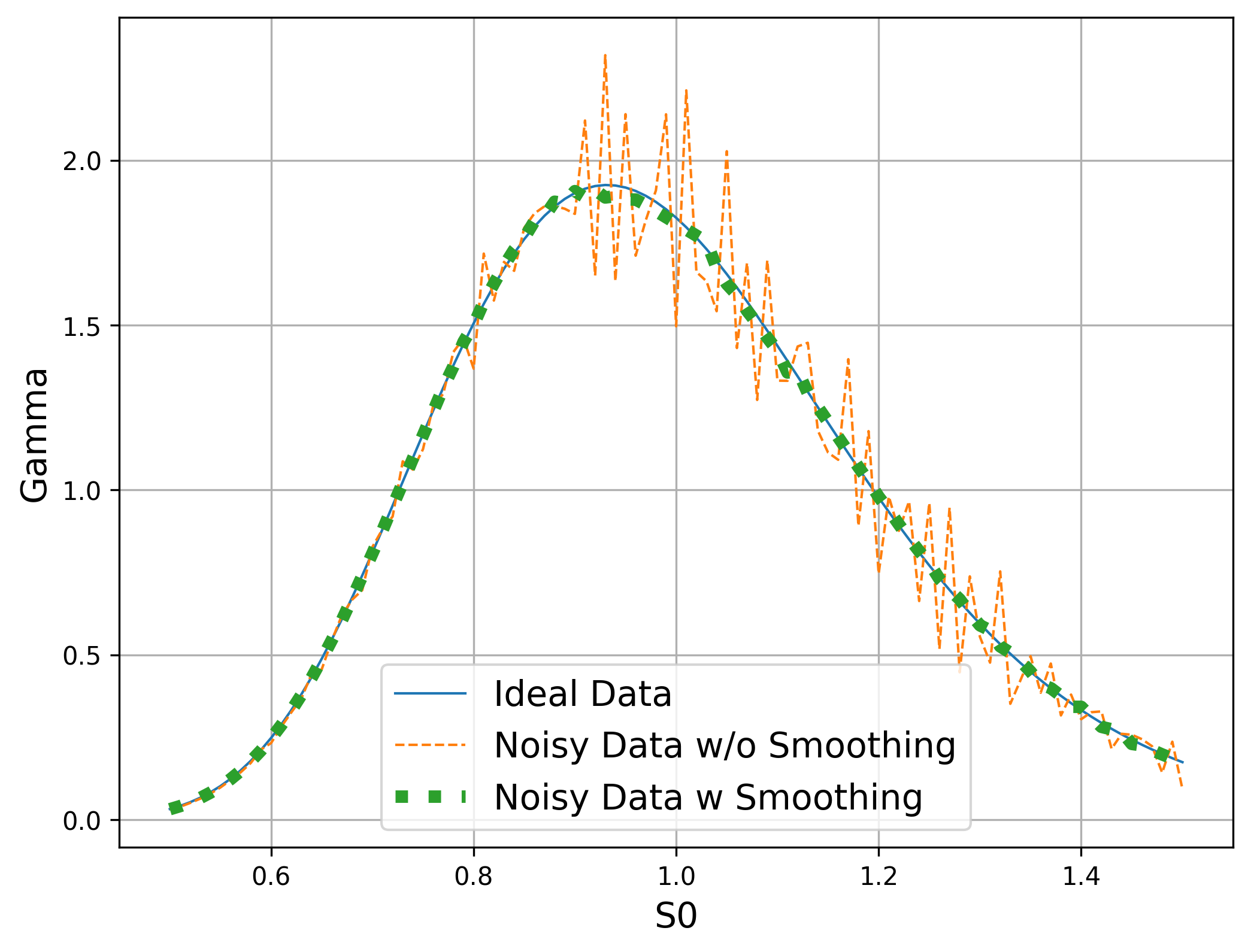}
  \caption{Gamma Value}
  \label{fg:SVI_Gamma_smoothing}
\end{minipage}
\end{figure}

Given that LV model plays a crucial role in the accurate pricing and risk management of exotic derivatives, we extend our analysis to demonstrate the practical significance of our proposed method in this context. To illustrate the impact on exotic option valuation, we conduct the same experimental procedure using an Asian arithmetic average call option, which represents a path-dependent derivative commonly encountered in financial markets.

An Asian arithmetic average call option is a path-dependent derivative whose payoff depends on the arithmetic average of the underlying asset price over a predetermined monitoring period. The payoff at maturity is given by:

\begin{equation}
\text{Payoff} = \max\left(\frac{1}{M}\sum_{i=1}^{M} S_{t_i} - K, 0\right)
\end{equation}

where $S_{t_i}$ represents the underlying asset price at monitoring time $t_i$, $M$ is the number of monitoring periods, and $K$ is the strike price. The option value is calculated as the risk-neutral expectation of the discounted payoff:

\begin{equation}
V_0 = e^{-rT} \mathbb{E}^{\mathbb{Q}}\left[\max\left(\frac{1}{M}\sum_{i=1}^{M} S_{t_i} - K, 0\right)\right]
\end{equation}

For this numerical experiment, we used the similar setup as for Figure \ref{fg:SVI_Delta_no_smoothing}, \ref{fg:SVI_Gamma_no_smoothing}, \ref{fg:SVI_Delta_smoothing} and \ref{fg:SVI_Gamma_smoothing}. For Asian option, we let $M = 12$, which means monthly monitoring. We here assume the spot price $S_0$ varies from $0.8$ to $1.2$, while the strike is fixed at $1$. Through Figure \ref{fig:SVI_Delta_smoothing_AC}, \ref{fig:SVI_Gamma_smoothing_AC2}, and \ref{fig:SVI_Gamma_smoothing_AC}, we can tell that our proposed method yields reliable Delta and Gamma. Compared to Figure \ref{fg:SVI_Delta_no_smoothing}, as the complicity of the contract increased, the stability issue of Delta value also appeared.

\begin{figure}[H]
\centering
\begin{minipage}{0.45\linewidth}
  \centering
  \includegraphics[width=\linewidth]{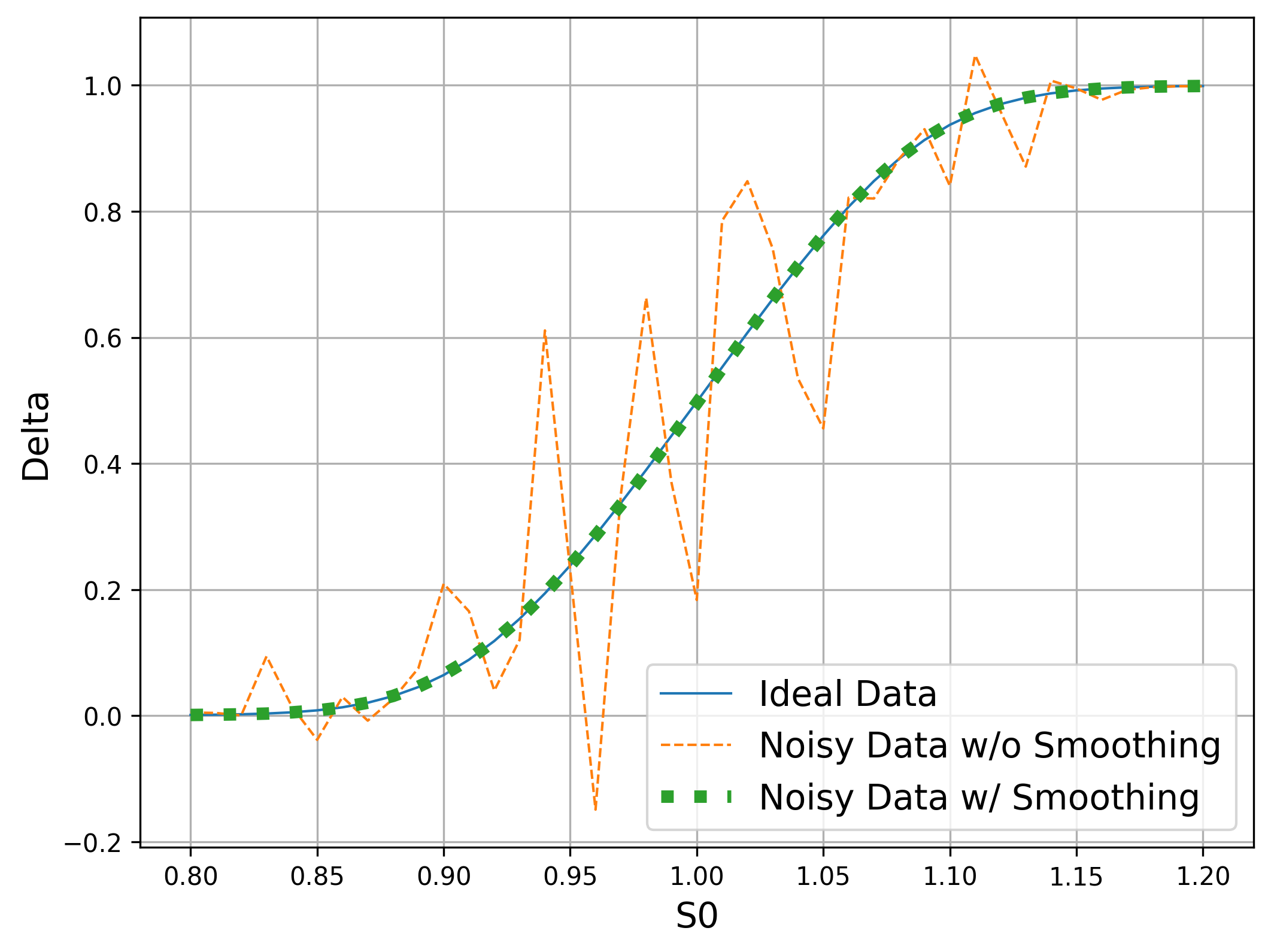}
  \caption{Delta Value}
  \label{fig:SVI_Delta_smoothing_AC}
\end{minipage}
\hfill
\begin{minipage}{0.45\linewidth}
  \centering
  \includegraphics[width=\linewidth]{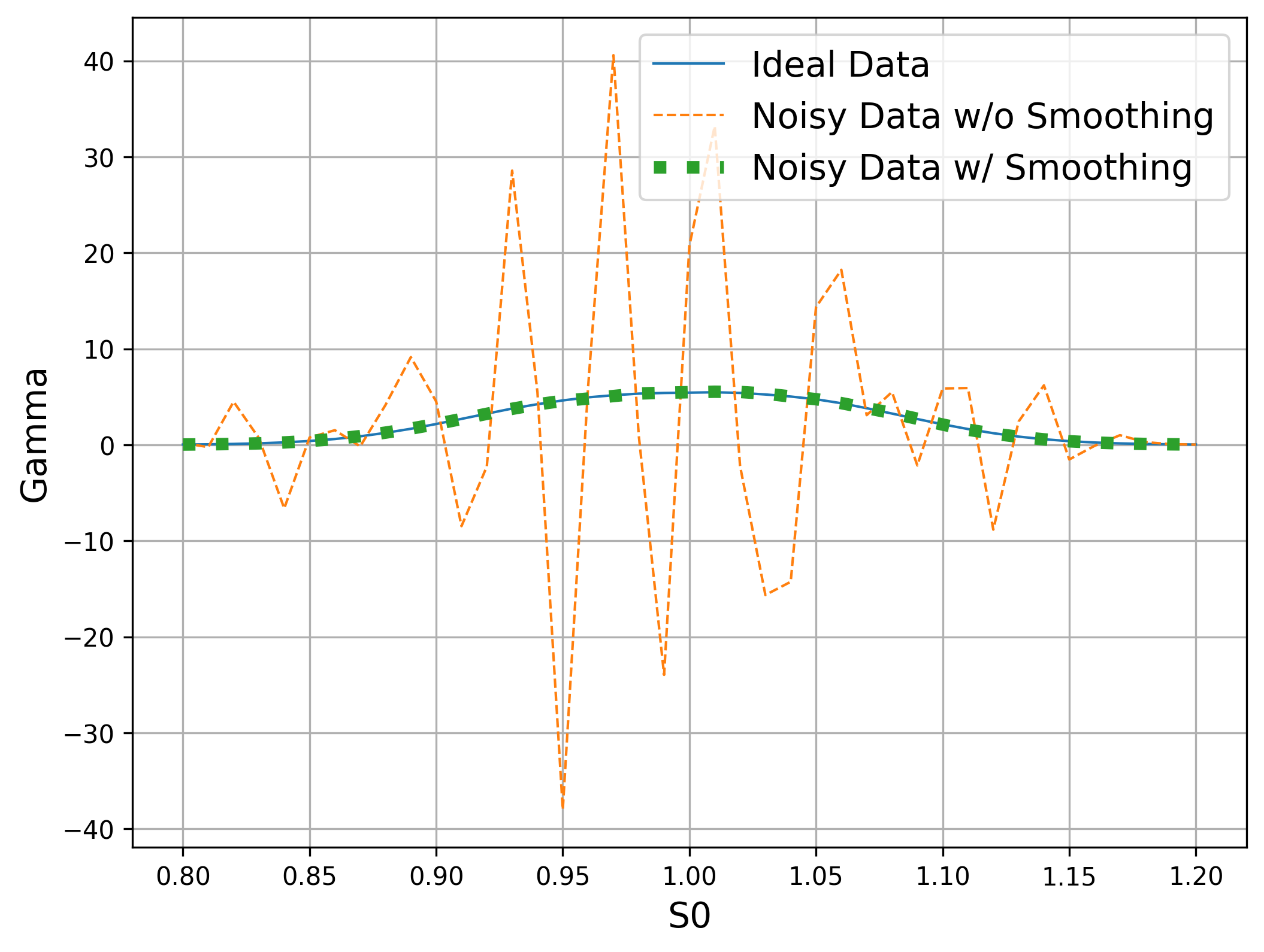}
  \caption{Gamma Value}
  \label{fig:SVI_Gamma_smoothing_AC}
\end{minipage}
\end{figure}

\begin{figure}[H]
\centering
\begin{minipage}{0.45\textwidth}
  \centering
  \includegraphics[width=\linewidth]{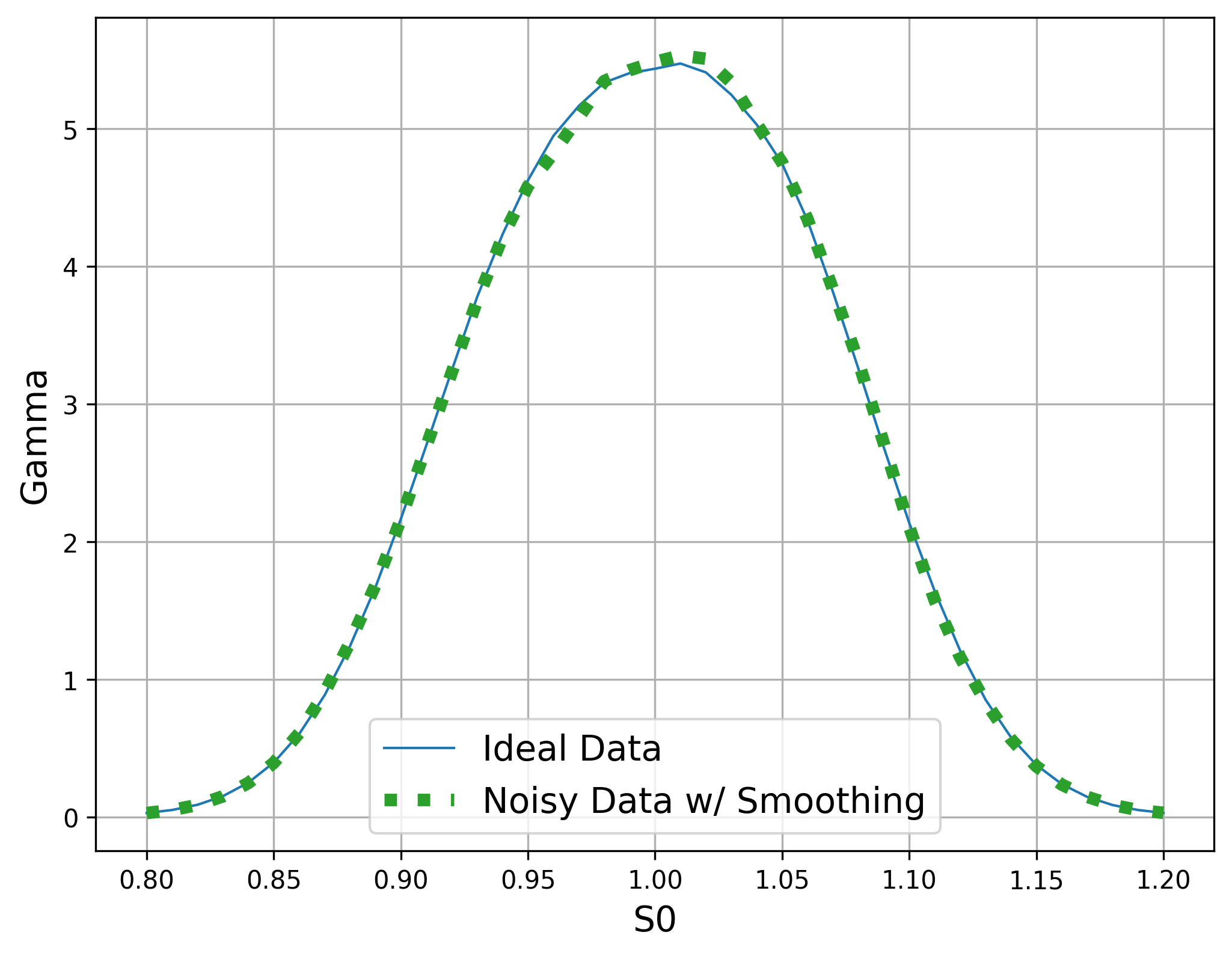}
\end{minipage}
  \caption{Gamma Value Zoomed in}
\label{fig:SVI_Gamma_smoothing_AC2}
\end{figure}

\subsection{Real-World W-Shaped Market}

We will now demonstrate that our method can process extreme and complicated market scenarios with high accuracy while maintaining smoothness. In \citep{alexiou2023pricing}, several extreme market scenarios are mentioned as examples of how the option market reacts to earnings announcement uncertainty. To test our method, we select the W-shaped IV surface of AAPL on October 28, 2013, one day before its quarterly earnings announcement, with data sourced from \textit{OptionMetrics}.

\subsubsection{Calibrate with only the F-D Based LV Fitting Method}

In this IV surface, the shortest time to maturity is four days. We take this IV curve from the IV surface to illustrate the issues. As demonstrated in Figure \ref{fig:AH_method_directly_LV}, if we apply only the F-D based LV fitting method on this IV curve, the calibrated LV function is spiky. Though as demonstrated in Figure \ref{fig:AH_method_directly_IV}, if we price European call options using this spiky LV function, and then invert the Black-Scholes formula to get the model IV, the model IV is closed align with the market IV and inside the spread.

\begin{figure}[H]
\centering
\begin{minipage}{0.45\linewidth}
  \centering
  \includegraphics[width=\linewidth]{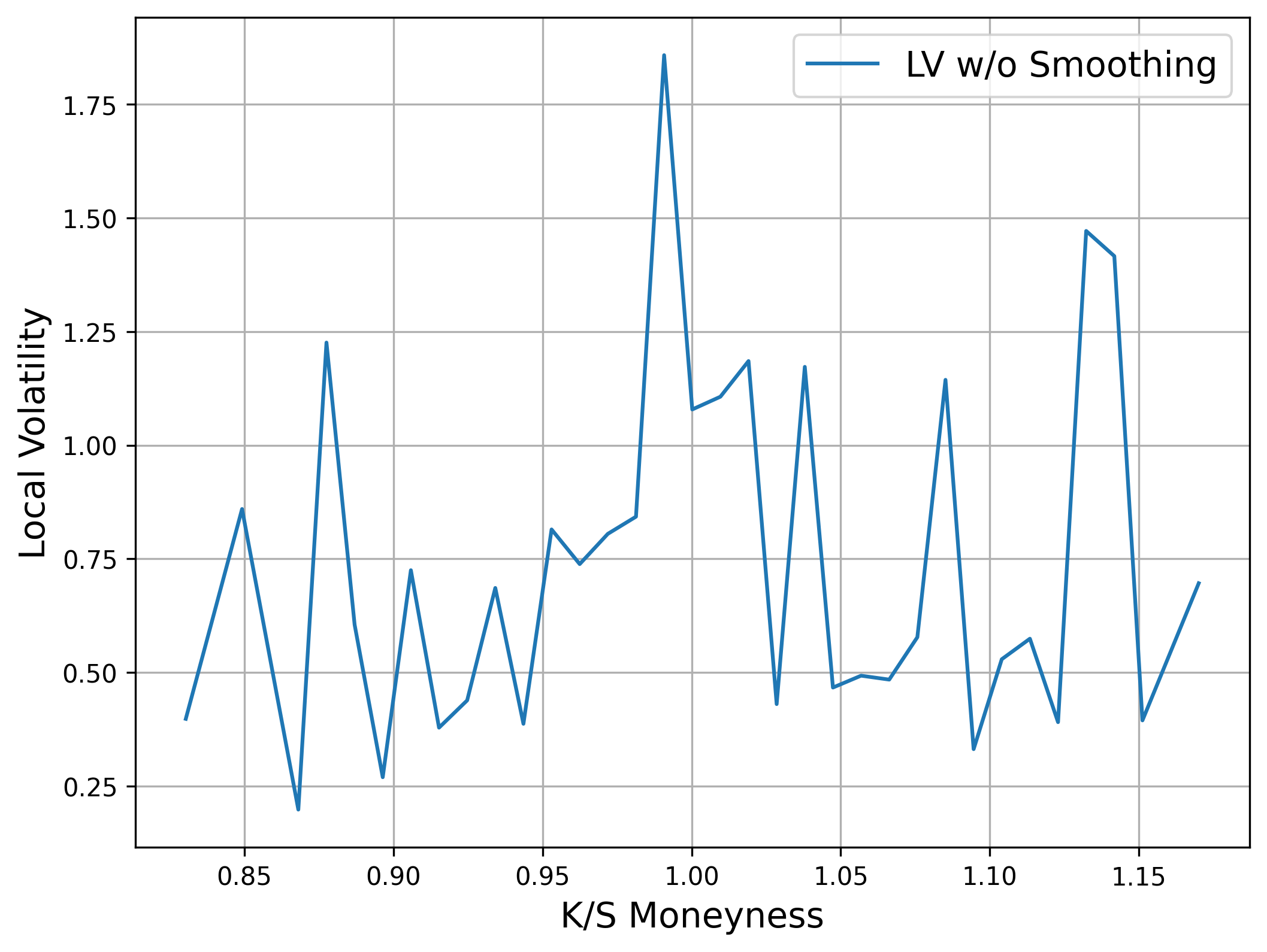}
  \caption{LV function from Direct LV Fitting Method}
  \label{fig:AH_method_directly_LV}
\end{minipage}
\hfill
\begin{minipage}{0.45\linewidth}
  \centering
  \includegraphics[width=\linewidth]{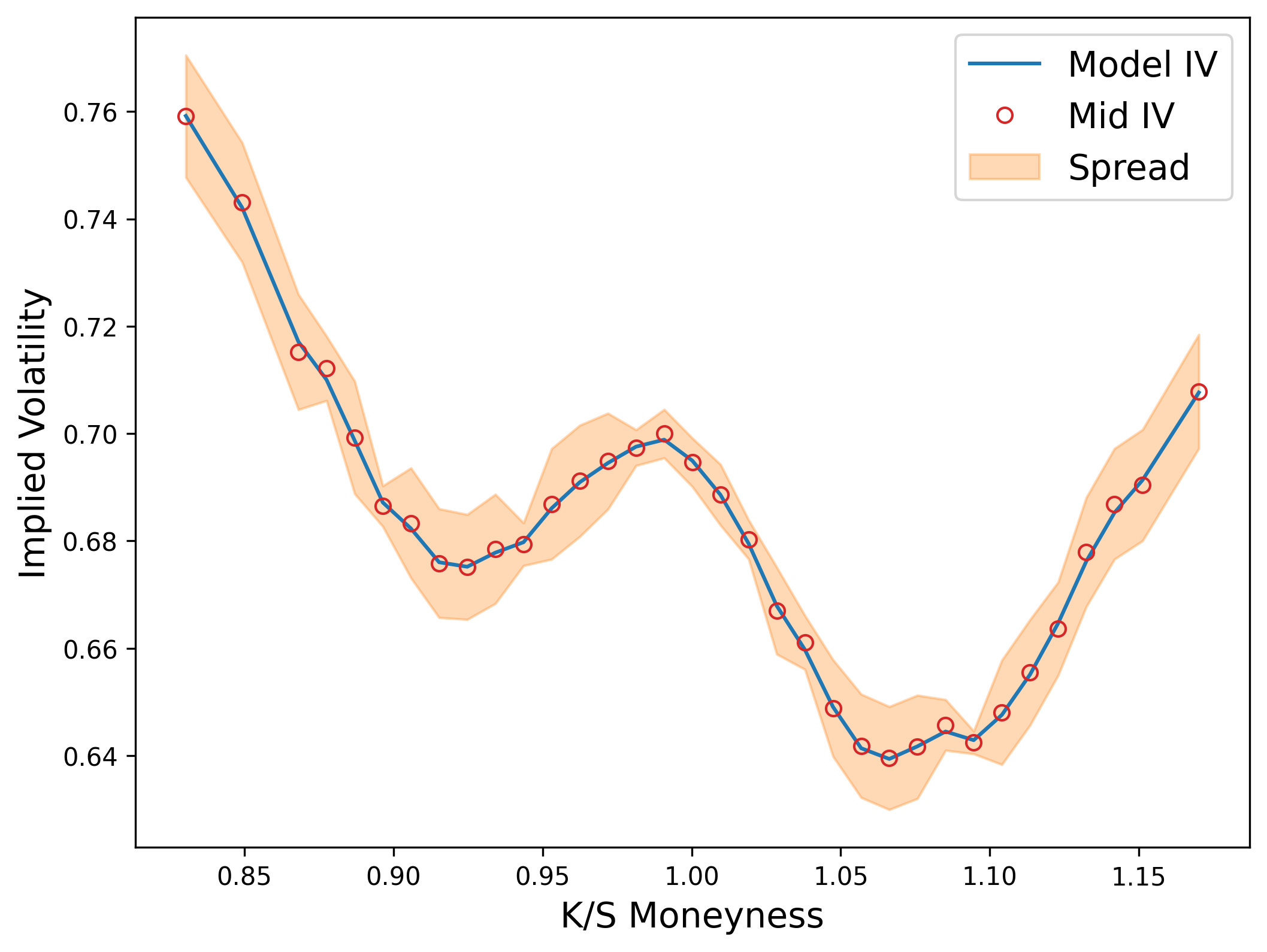}
  \caption{IV Curve from Direct LV Fitting Method}
  \label{fig:AH_method_directly_IV}
\end{minipage}
\end{figure}

\subsubsection{Calibrate with Regularization Method}

One popular option for enhancing the smoothness of the direct LV fitting method is to add a regularization term into the optimization function. As we demonstrated in previous sections, their objective function could usually be written as:

\begin{equation}
F(\theta) = \sum_{i} \left\| C(\theta, K_i) - C_{\text{market}}(K_i) \right\|^2 
+ \lambda \sum_{i} \left\| \frac{\partial^2 C}{\partial K^2} \Big|_{K=K_i} \right\|^2
\end{equation}

Intuitively, achieving the optimal trade-off between smoothness and accuracy is complex and may be unattainable due to the intricate nature of market conditions. As we previously highlighted, producing a smooth LV function necessitates a smooth second derivative, not merely a small second derivative. The distinction between smoothing and minimizing is crucial here; minimizing the second derivative does not inherently lead to a sufficiently smooth LV curve.

As in Figures \ref{fig:regularized_1e6_1}, \ref{fig:regularized_1e6_2}, \ref{fig:regularized_1e5_1}, \ref{fig:regularized_1e5_2}, \ref{fig:regularized_1e4_1}, and \ref{fig:regularized_1e4_2} below, when we increase the regularization strength $\lambda$, the reproduced IV curve deviates from the target market IV. This deviation occurs because the optimization no longer seeks the LV function that best matches the market IV, but rather balances between fit accuracy and smoothness.

\begin{figure}[H]
\centering
\begin{minipage}{0.45\linewidth}
    \centering
    \includegraphics[width=\linewidth]{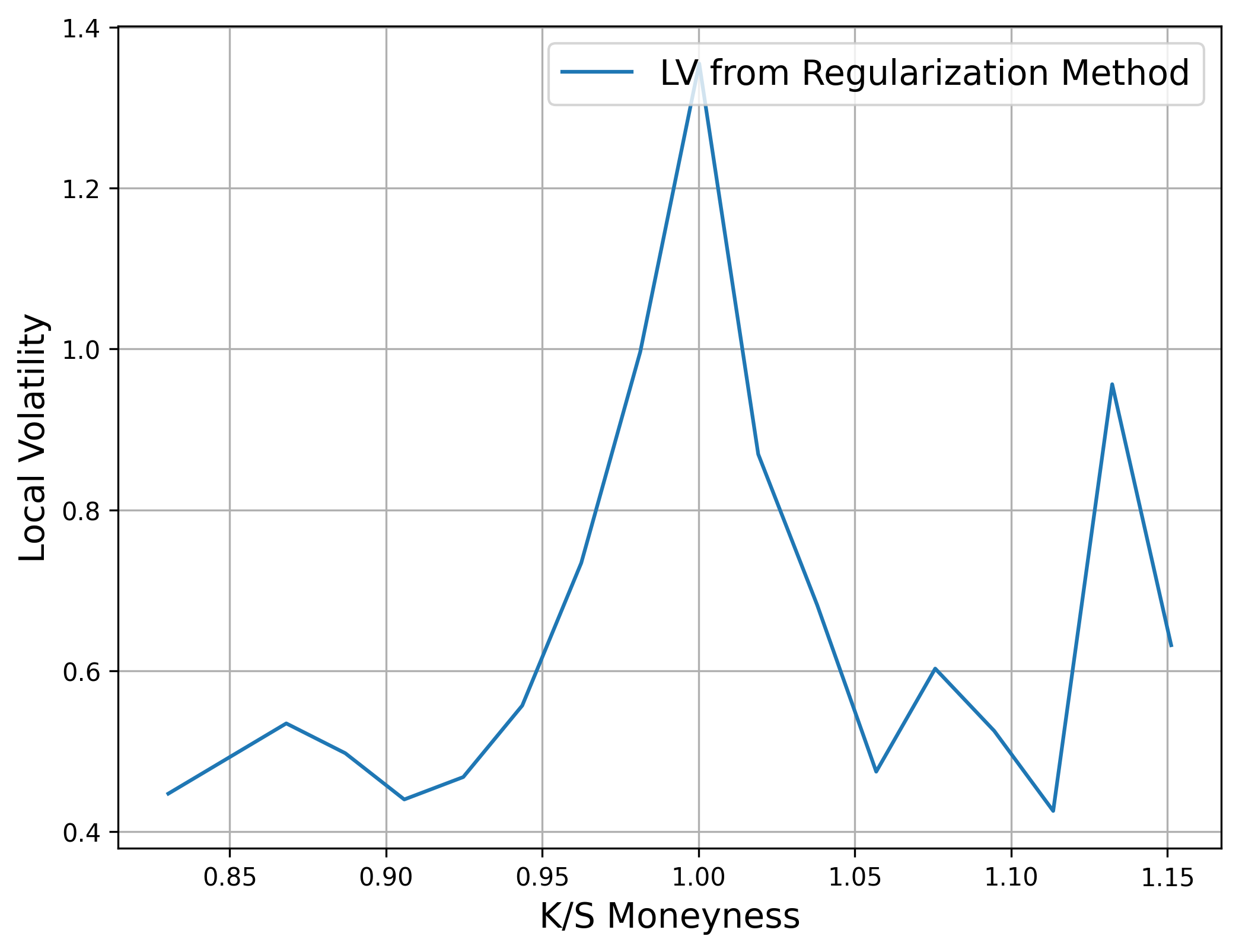}
    \caption{LV curve when $\lambda=1e-6$}
    \label{fig:regularized_1e6_1}
\end{minipage}
\hfill
\begin{minipage}{0.45\linewidth}
    \centering
    \includegraphics[width=\linewidth]{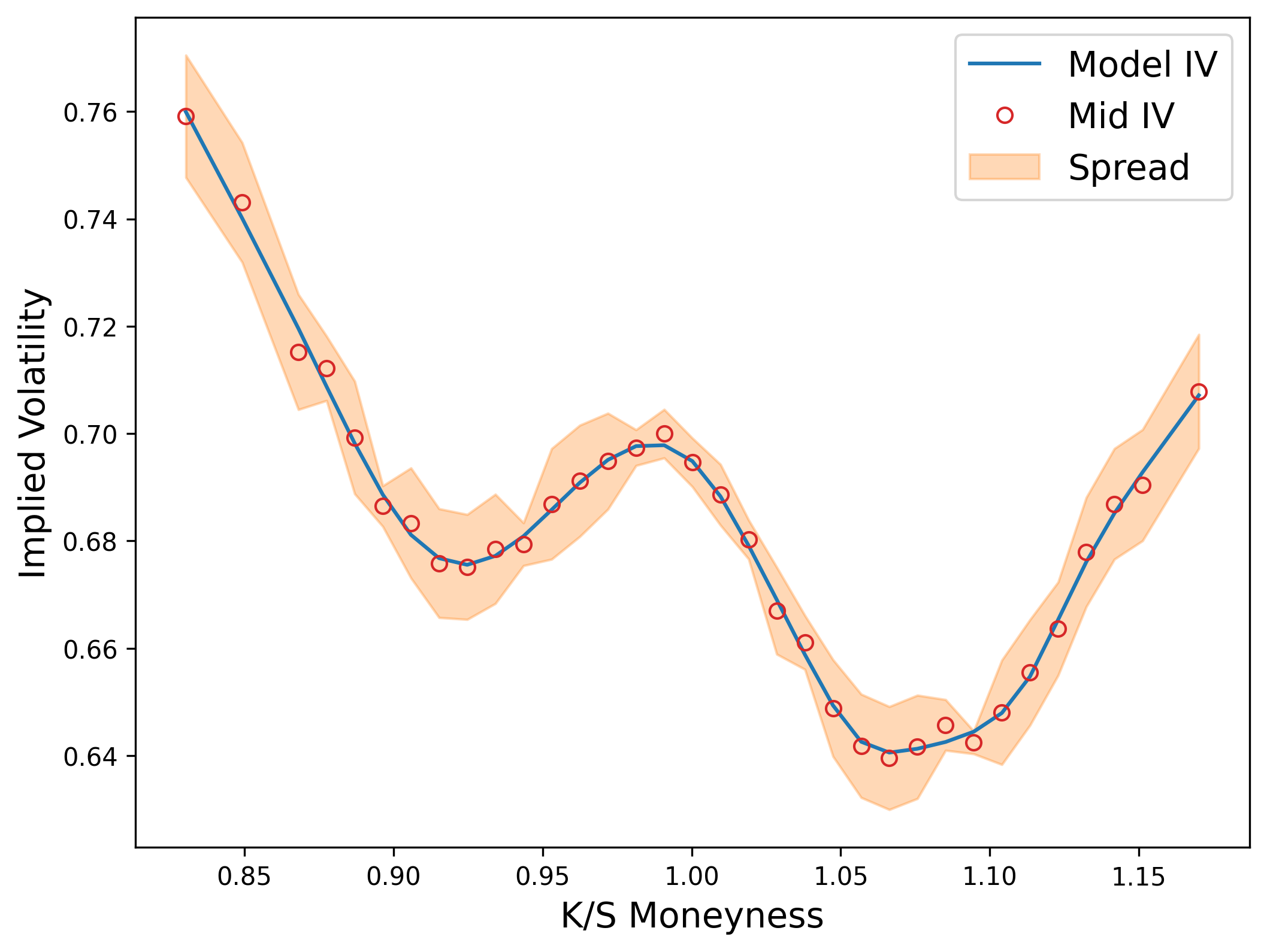}
    \caption{IV Function when $\lambda=1e-6$}
    \label{fig:regularized_1e6_2}
\end{minipage}
\end{figure}

\begin{figure}[H]
\centering
\begin{minipage}{0.45\linewidth}
    \centering
    \includegraphics[width=\linewidth]{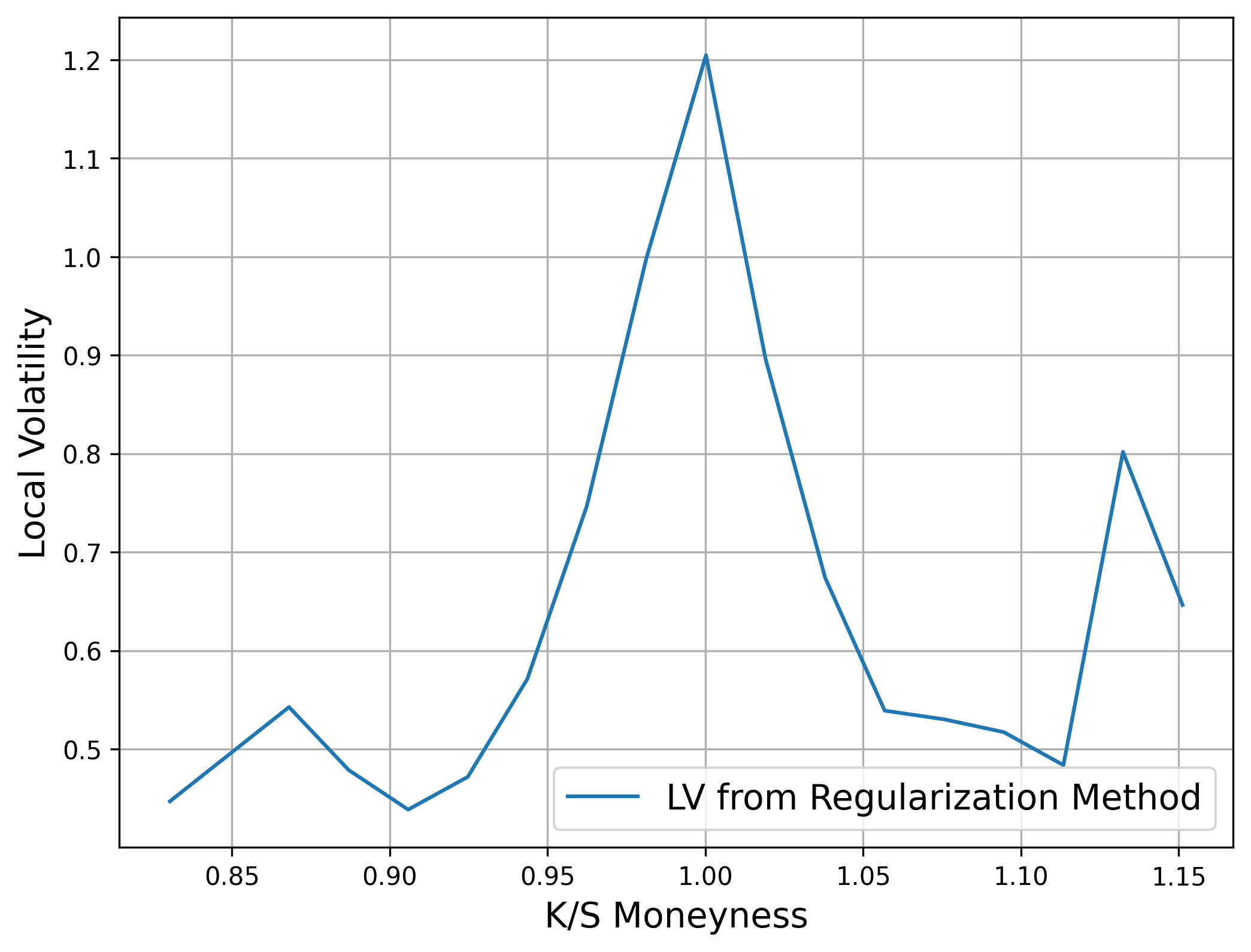}
    \caption{LV curve when $\lambda=1e-5$}
    \label{fig:regularized_1e5_1}
\end{minipage}
\hfill
\begin{minipage}{0.45\linewidth}
    \centering
    \includegraphics[width=\linewidth]{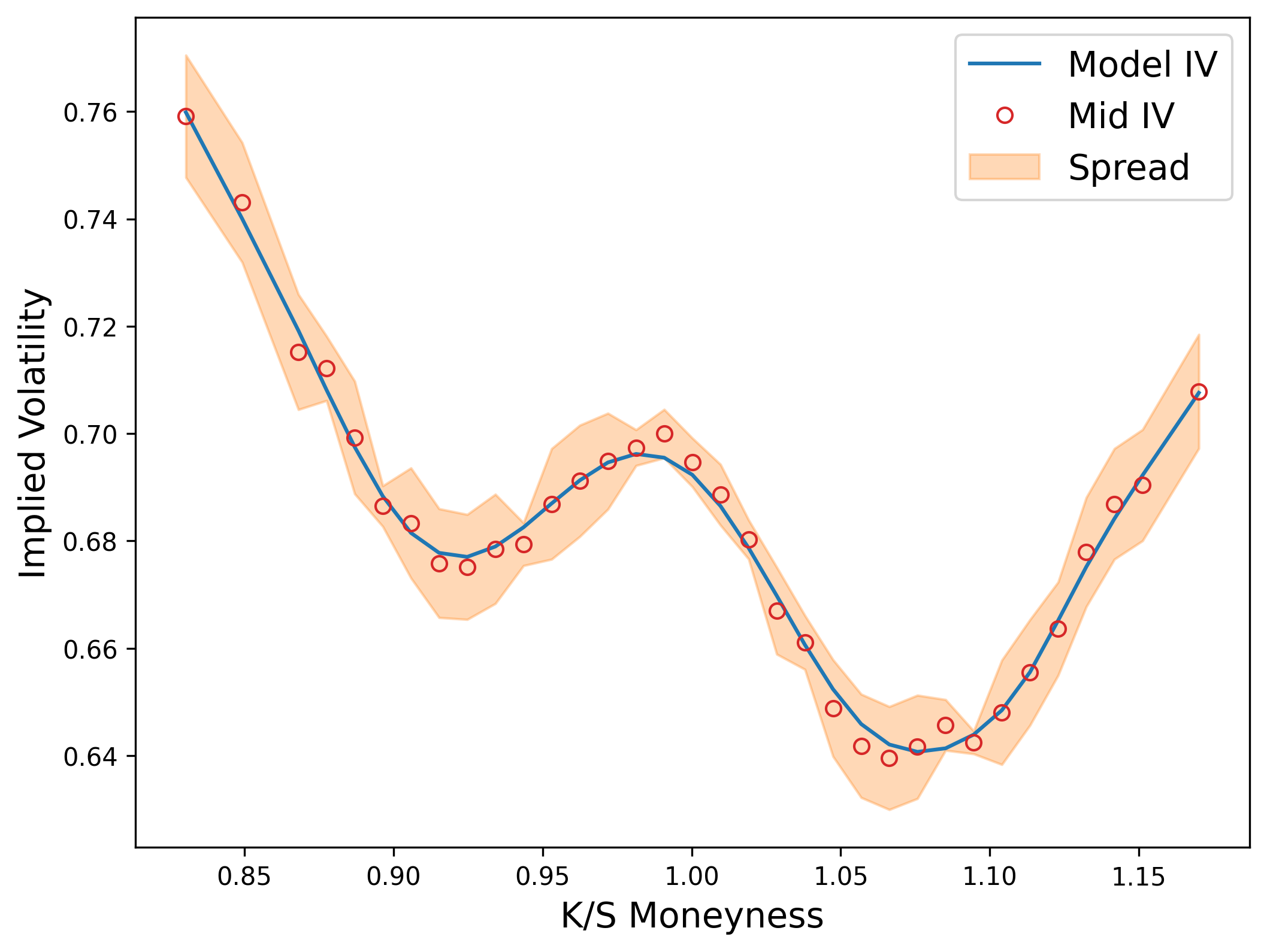}
    \caption{IV Function when $\lambda=1e-5$}
    \label{fig:regularized_1e5_2}
\end{minipage}
\end{figure}

\begin{figure}[H]
\centering
\begin{minipage}{0.45\linewidth}
    \centering
    \includegraphics[width=\linewidth]{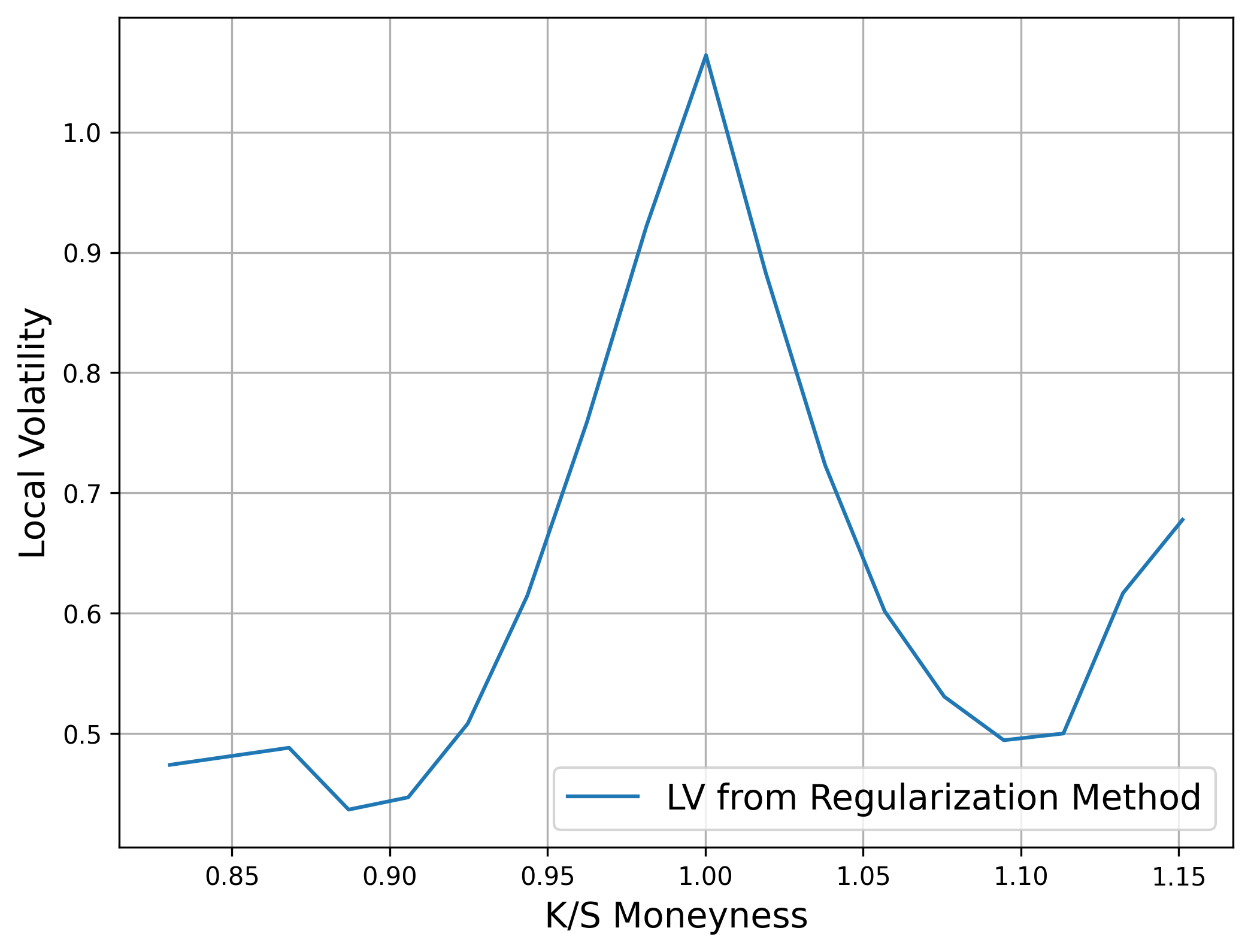}
    \caption{LV curve when $\lambda=1e-4$}
    \label{fig:regularized_1e4_1}
\end{minipage}
\hfill
\begin{minipage}{0.45\linewidth}
    \centering
    \includegraphics[width=\linewidth]{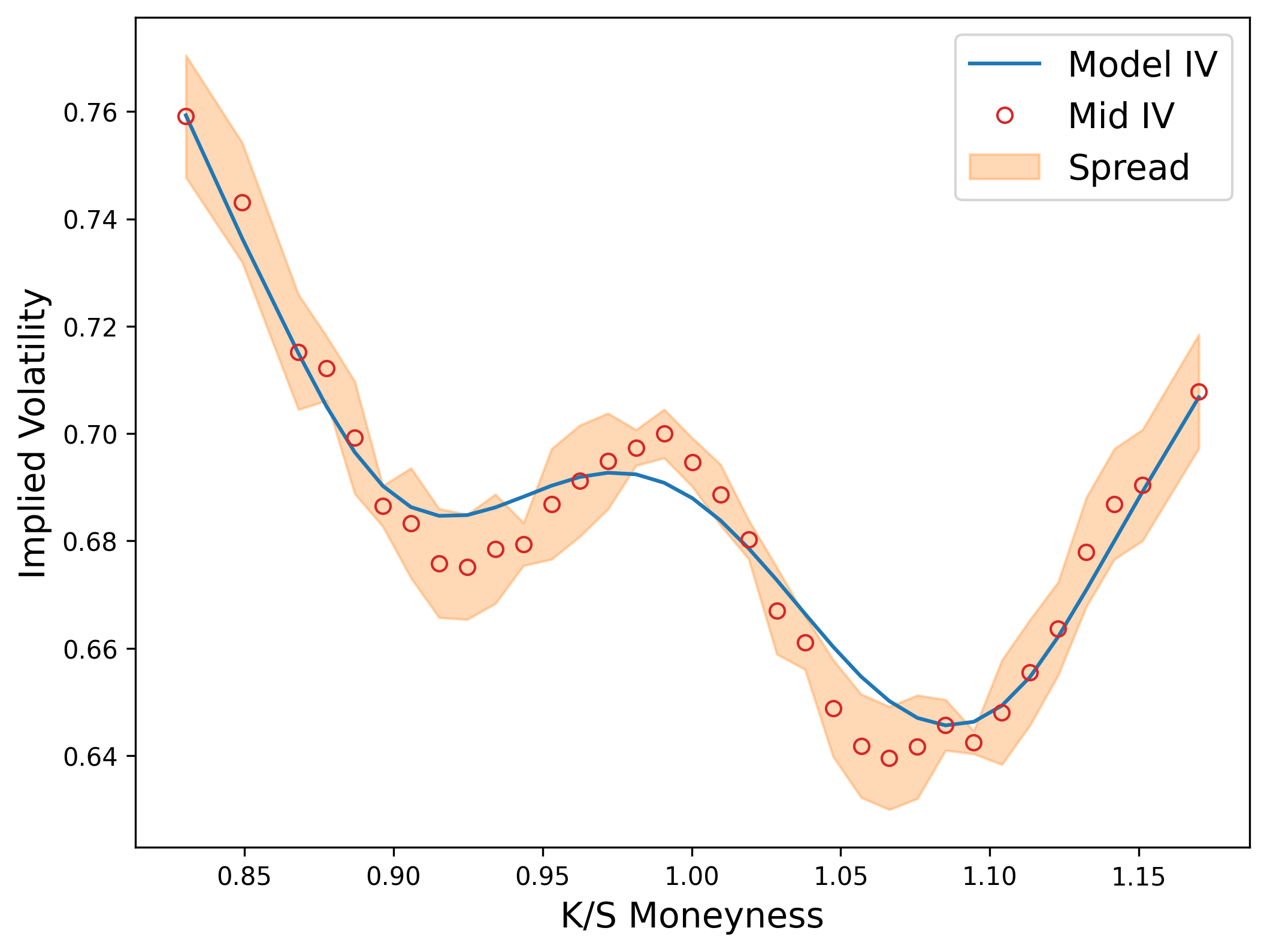}
    \caption{IV Function when $\lambda=1e-4$}
    \label{fig:regularized_1e4_2}
\end{minipage}
\end{figure}

\subsubsection{Calibrate with the Proposed Method}

However, if we apply the proposed method to process the market data prior to the LV fitting, as demonstrated in Figure \ref{fig:Proposed_method_directly_LV} and \ref{fig:Proposed_method_directly_IV}, we can observe that the LV function is smoother, without compromising the fitting quality.

\begin{figure}[H]
\centering
\begin{minipage}{0.45\linewidth}
  \centering
  \includegraphics[width=\linewidth]{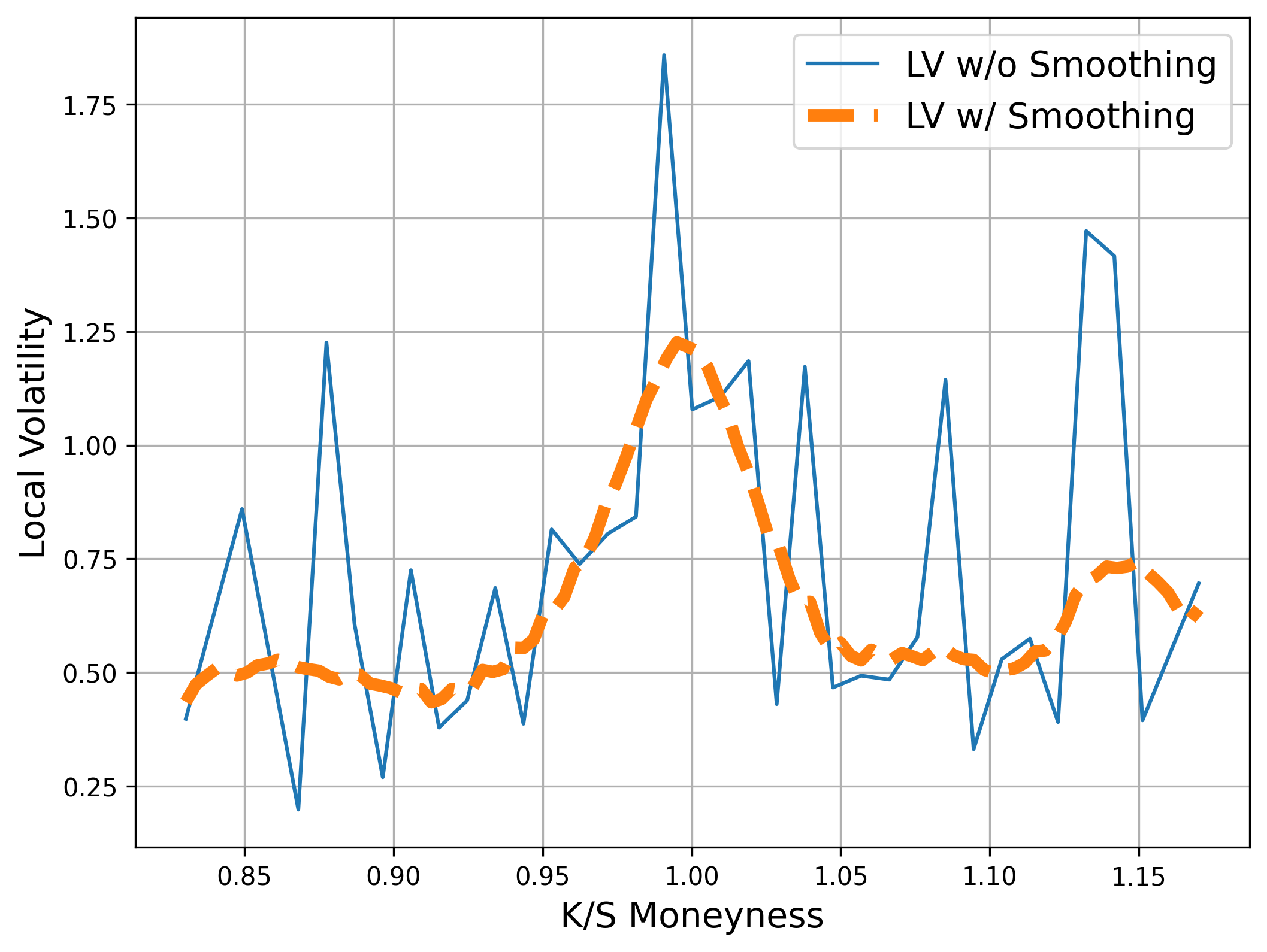}
  \caption{LV function from Proposed Method}
  \label{fig:Proposed_method_directly_LV}
\end{minipage}
\hfill
\begin{minipage}{0.45\linewidth}
  \centering
  \includegraphics[width=\linewidth]{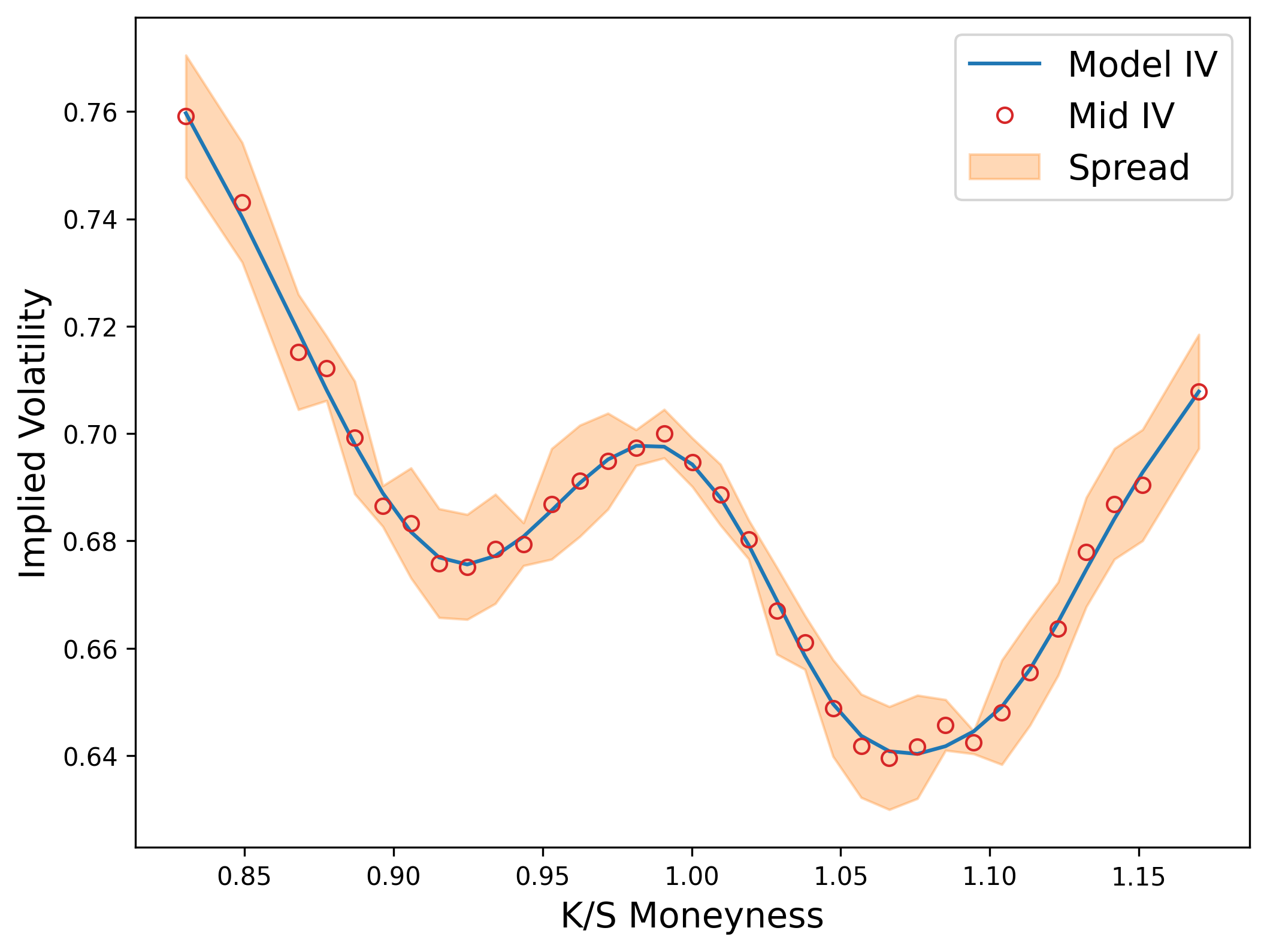}
  \caption{IV Curve from Proposed Method}
  \label{fig:Proposed_method_directly_IV}
\end{minipage}
\label{fig:LP_method}
\end{figure}

Iteratively, as demonstrated in Figure \ref{fig:Proposed_method_directly_LV_Surface}, we construct the LV function for the entire surface. Moreover, through standard F-D approach, we can price the European call options for different maturities, and invert the Black-Scholes formula to get the IV value for different strikes and maturities. This IV surface is demonstrated at Figure \ref{fig:Proposed_method_directly_IV_Surface}, the model IV surface has been shifted up $0.1$ for better viewing.

\begin{figure}[H]
\centering
\begin{minipage}{0.45\linewidth}
  \centering
  \includegraphics[width=\linewidth]{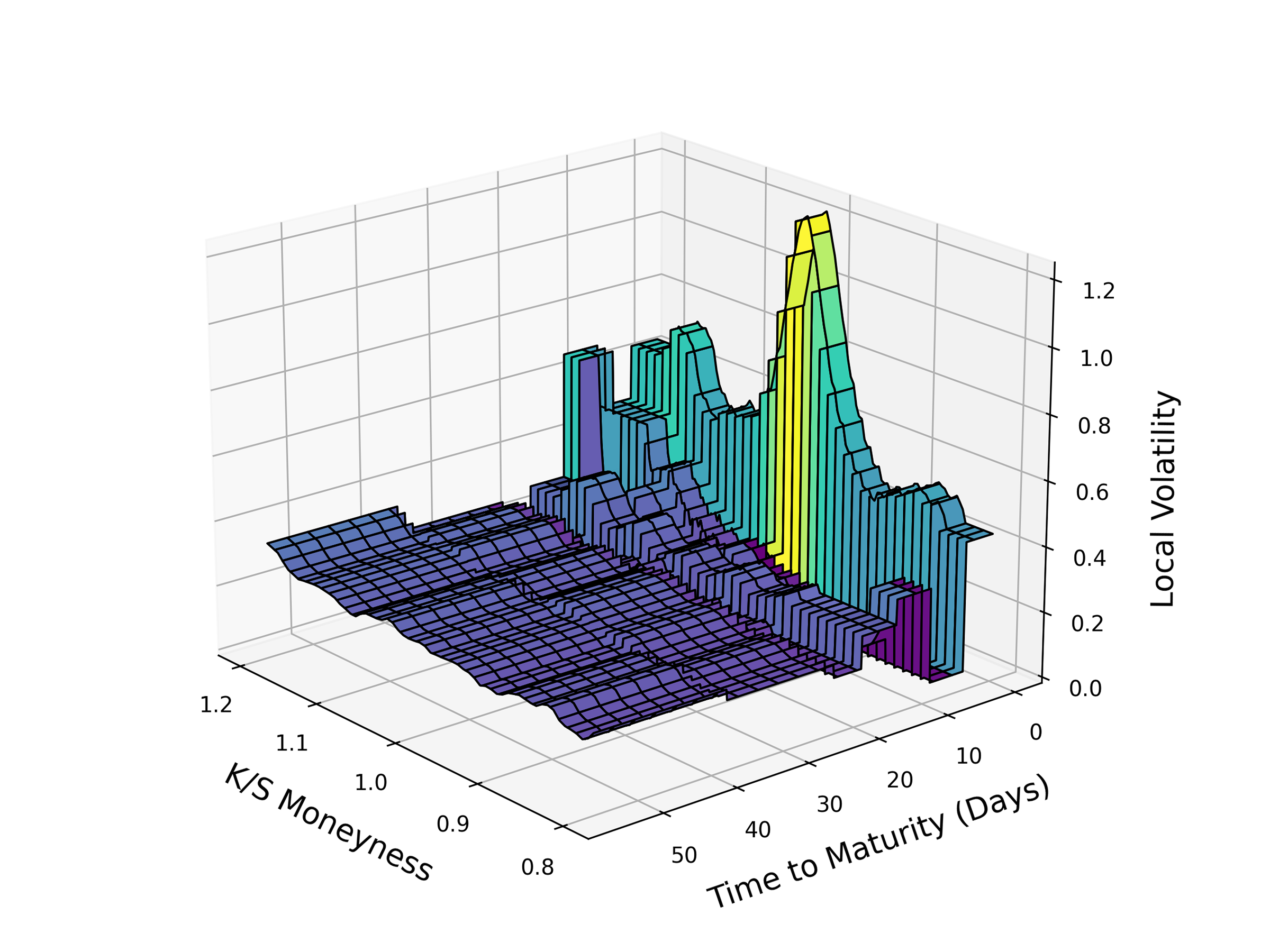}
  \caption{LV function from Proposed Method}
  \label{fig:Proposed_method_directly_LV_Surface}
\end{minipage}
\hfill
\begin{minipage}{0.45\linewidth}
  \centering
  \includegraphics[width=\linewidth]{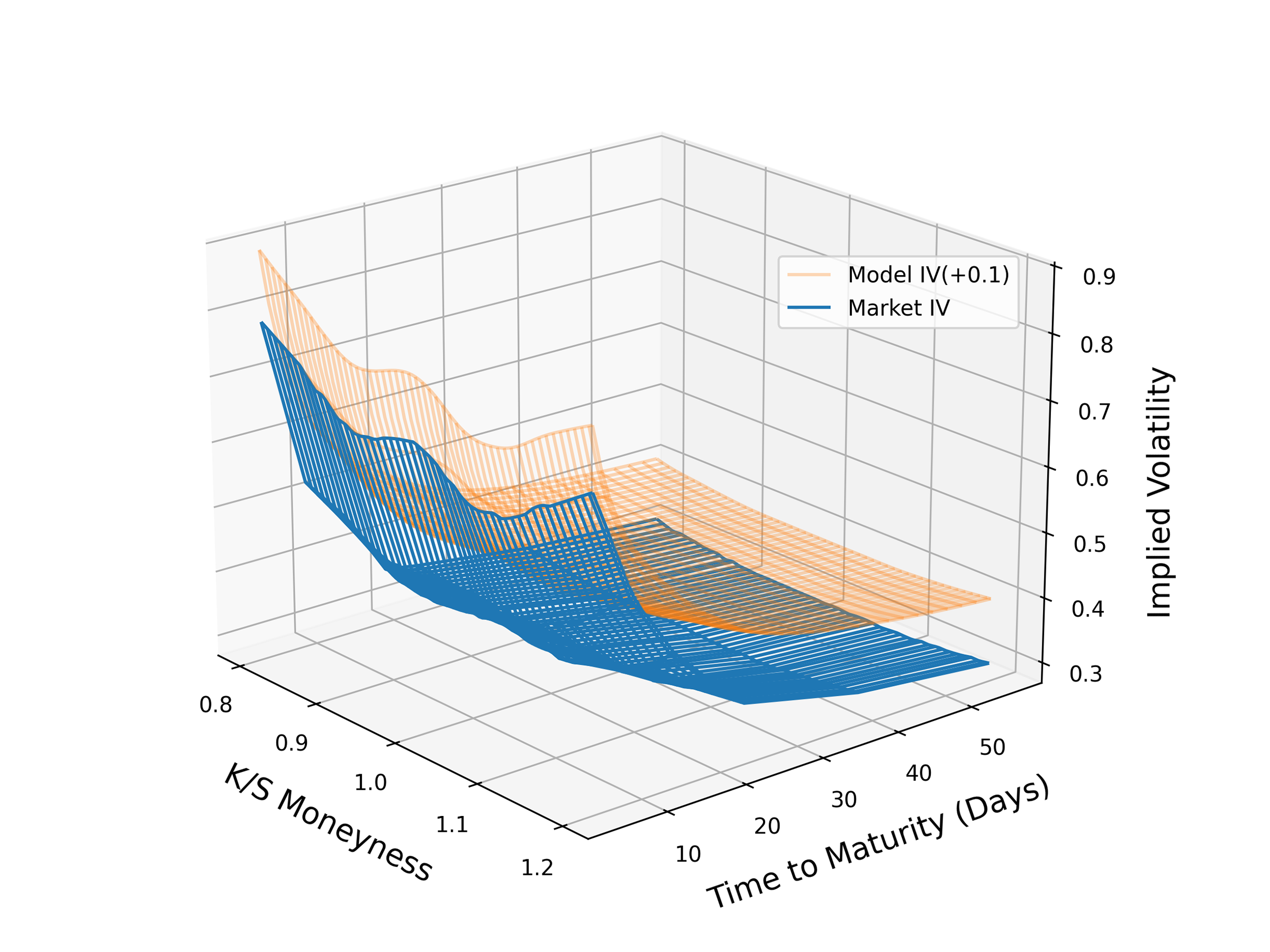}
  \caption{IV Surface from Proposed Method}
  \label{fig:Proposed_method_directly_IV_Surface}
\end{minipage}
\label{fig:LP_method}
\end{figure}

To quantitatively assess the quality of the fit, we report the Fail Ratio of proposed method—the proportion of model prices that fall outside the bid-ask spread:

\begin{equation}
\text{fail ratio} = \frac{\# \text{ model price points outside bid/ask}}{\# \text{ of all points}} \approx 1\%
\end{equation}

\begin{table}[H]
\centering
\textbf{Absolute Calibration Error in \% for AAPL Surface} \\[1ex]
\begin{tabular}{|c|c|c|c|}
\hline
             & Moneyness $< 0.95$ &  $0.95 <$ Moneyness $< 1.05$ & $1.05<$ Moneyness \\ \hline
Maturity $<1$ month  &          0.25\%      &                   0.28\%         &               0.31\% \\ \hline
$1$ month $<$ Maturity $<2$ month  &          0.22\%      &                   0.2\%         &               0.25\% \\ \hline
$2$ month $<$Maturity $<3$ month  &          0.25\%      &                   0.13\%         &               0.24\% \\ \hline
\end{tabular}
\caption {Absolute Calibration Error in \% for AAPL Surface}
\end{table}

\section{Conclusion}

In this paper, we address a fundamental challenge in exotic derivative pricing: the ill-posed nature of Local Volatility (LV) calibration that creates spiky surfaces and unstable Greeks despite the model's superior ability to match market observables across all strikes and maturities. We introduce a novel pre-calibration smoothing method that seamlessly integrates into existing LV calibration workflows to resolve these stability issues.

Our approach employs local regression that automatically minimizes asymptotic conditional mean squared error to pre-process market observables, generating denoised inputs for subsequent LV calibration. This method distinguishes itself by requiring no manual parameter tuning while efficiently filtering market noise and preserving essential market shape information. The resulting framework produces arbitrage-free LV functions that maintain the model's calibration accuracy while delivering the smooth, stable surfaces necessary for reliable hedging.
The empirical study demonstrates that our pre-calibration smoothing yields significantly smoother LV surfaces and greatly improves Greek stability for exotic options with negligible additional computational cost. Moreover, these improvements come without compromising the LV model's defining characteristic: its ability to fit market observables with high fidelity across all strikes and maturities.

Our method's design as a preprocessing step makes it universally applicable to any LV calibration workflow, offering practitioners a practical solution to the persistent problem of unstable Greeks in finite-difference-based exotic derivative valuation. By solving the spiky LV surface problem while preserving calibration accuracy, this framework enhances the reliability and stability of derivative pricing and risk management in practice.

\section{Appendix}

\subsection{Proof for proposition \ref{prop:ACMSE convergece in algo}}

\begin{proof}
We first establish that the sequence $\{\hat{Z}_k(p_n, h_n)\}_{n=0}^{\infty}$ generated by the iterative algorithm is monotonically non-increasing and bounded below, guaranteeing its convergence.

For any iteration $n \geq 0$, we have:
\begin{align}
p_{n+1} &= \arg\min_p \hat{Z}_k(p, h_n)\\
\Rightarrow \hat{Z}_k(p_{n+1}, h_n) &\leq \hat{Z}_k(p_n, h_n)
\end{align}

Similarly, for the second optimization step:
\begin{align}
h_{n+1} &= \arg\min_h \hat{Z}_k(p_{n+1}, h)\\
\Rightarrow \hat{Z}_k(p_{n+1}, h_{n+1}) &\leq \hat{Z}_k(p_{n+1}, h_n)
\end{align}

By transitivity:
\begin{align}
\hat{Z}_k(p_{n+1}, h_{n+1}) \leq \hat{Z}_k(p_{n+1}, h_n) \leq \hat{Z}_k(p_n, h_n)
\end{align}

Thus, $\{\hat{Z}_k(p_n, h_n)\}_{n=0}^{\infty}$ is monotonically non-increasing.

By definition of global minimum, the sequence $\{\hat{Z}_k(p_n, h_n)\}_{n=0}^{\infty}$ is bounded below by $\hat{Z}^*$.

By the Monotone Convergence Theorem, since $\{\hat{Z}_k(p_n, h_n)\}_{n=0}^{\infty}$ is monotonically non-increasing and bounded below, it converges to some limit $L$:
\begin{align}
\lim_{n \to \infty}\hat{Z}_k(p_n, h_n) = L \geq Z^*
\end{align}

We now prove that $L = \hat{Z}^*$ by contradiction.

Assume $L > \hat{Z}^*$. Let $\varepsilon = \frac{L - \hat{Z}^*}{3} > 0$.

By the continuity of $g(\cdot)$, $f^{(p+1)}(\cdot)$, and $\tau^2(\cdot)$ in a neighborhood of $k$, one can ensure that $\hat{Z}_k$ is continuous for both variables. For any point $(p, h) \not\in \hat{Ph}^*$ with $\hat{Z}_k(p, h) \leq L - \varepsilon$, at least one of the following must hold:
\begin{align}
\min_{p'} \hat{Z}_k(p', h) &< \hat{Z}_k(p, h)\\
\min_{h'} \hat{Z}_k(p, h') &< \hat{Z}_k(p, h)
\end{align}

For any $(p_n, h_n)$ satisfying $L - \varepsilon < \hat{Z}_k(p_n, h_n) < L - \varepsilon + \delta$ for some sufficiently small $\delta > 0$:
\begin{itemize}
    \item In the $p$-optimization step: $\hat{Z}_k(p_{n+1}, h_n) \leq \min_{p'} \hat{Z}_k(p', h_n) < \hat{Z}_k(p_n, h_n) - \delta_1$ for some $\delta_1 > 0$
    \item In the $h$-optimization step: $\hat{Z}_k(p_{n+1}, h_{n+1}) \leq \min_{h'} \hat{Z}_k(p_{n+1}, h') < \hat{Z}_k(p_{n+1}, h_n) - \delta_2$ for some $\delta_2 > 0$
\end{itemize}

This implies:
\begin{align}
\hat{Z}_k(p_{n+1}, h_{n+1}) < \hat{Z}_k(p_n, h_n) - \delta
\end{align}
where $\delta = \min(\delta_1, \delta_2) > 0$.

Since $\{\hat{Z}_k(p_n, h_n)\}_{n=0}^{\infty}$ converges to $L$, there exists $N$ such that for all $n \geq N$:
\begin{align}
|\hat{Z}_k(p_n, h_n) - L| < \varepsilon \Rightarrow \hat{Z}_k(p_n, h_n) > L - \varepsilon
\end{align}

However, following our argument, starting from any point where $\hat{Z}_k(p_N, h_N) < L - \varepsilon + \delta$, we would get:
\begin{align}
\hat{Z}_k(p_{N+1}, h_{N+1}) < \hat{Z}_k(p_N, h_N) - \delta < L - \varepsilon
\end{align}

This contradicts our established convergence property that $\hat{Z}_k(p_n, h_n) > L - \varepsilon$ for all $n \geq N$.

Therefore, our assumption that $L > \hat{Z}^*$ is false, and we must have $L = \hat{Z}^*$.

Lastly, we show that the algorithm eventually produces optimal parameter pairs in $\hat{Ph}^*$.

Since $\lim_{n \to \infty} \hat{Z}_k(p_n, h_n) = \hat{Z}^*$, for any $\epsilon > 0$, there exists $N_1$ such that for all $n \geq N_1$, $|\hat{Z}_k(p_n, h_n) - \hat{Z}^*| < \epsilon$.

Choose $\epsilon$ sufficiently small such that for any $(p, h)$ with $\hat{Z}_k(p, h) < Z^* + \epsilon$, if $(p, h) \not\in Ph^*$, then at least one of the optimization steps yields a strict improvement.

Then, there exists $N_2 \geq N_1$ such that for all $n \geq N_2$, $(p_n, h_n) \in \hat{Ph}^*$. Otherwise, we would continually improve beyond $\hat{Z}^*$, which is impossible.

Under this basis, if the algorithm includes the stopping criterion $|\hat{Z}_k(p_{n+1}, h_{n+1}) - \hat{Z}_k(p_n, h_n)| < \epsilon$, then for $n \geq N_2$, we have $\hat{Z}_k(p_{n+1}, h_{n+1}) = \hat{Z}_k(p_n, h_n) = \hat{Z}^*$, satisfying the stopping criterion and causing the algorithm to terminate.

This completes the proof.

\end{proof}
\subsection{Proof for proposition \ref{prop:tau convergence rate}}
\subsubsection{Establish of the Pseudo-Nadaraya-Watson Variance Estimator}
We apply assumption \ref{assup:convergence analysis} through our analysis. Under our settings, for a fixed point $K$, we have $E\left[\varepsilon_i \mid K_i=K\right]=0$ and $\operatorname{Var}\left(\varepsilon_i \mid K_i=K\right)=\tau^2(K)$. To estimate the variance function $\tau^2(K)$ for a given strike $K$, we use the squared residuals from the mean fit. Define the residuals $\hat{\varepsilon}_j := \sigma_j - \hat{f}(K_j)$ as an estimation to the true residual on $j$-th index $\varepsilon_j$. We form a local variance estimator $\hat{\tau}^2(K)$ using a local polynomial of degree 0 by solving a weighted least squares problem: 
$$\hat{\tau}^2(K) := \arg\min_c \sum_{j=1}^n \kappa\left(\frac{K_j-K}{h}\right)(\hat{\varepsilon}_j^2-c)^2$$

This is indeed the locally weighted average of squared residuals around $K$. After simple algebra, one obtains:
$$\hat{\tau}^2(K) = \frac{\sum_{j=1}^n \kappa\left(\frac{K_j-K}{h}\right) \hat{\varepsilon}_j^2}{\sum_{j=1}^n \kappa\left(\frac{K_j-K}{h}\right)}$$

For simplicity, we denote $\hat{\tau}_i^2 := \left.\hat{\tau}^2(K)\right|_{K=K_i}$.

Recall the classic Nadaraya-Watson estimator for $\tau^2$ in our case:
\begin{align*}
\tilde{\tau}^2(K) = \frac{\sum_{j=1}^n \kappa\left(\frac{K_j-K}{h}\right) \varepsilon_j^2}{\sum_{j=1}^n \kappa\left(\frac{K_j-K}{h}\right)}
\end{align*}
where $\varepsilon_j$ is the true error residual for the $j$-th index. However, our estimator $\hat{\tau}^2(K)$ is based on the estimated implied volatility function $\hat{f}\left(K_j\right)$. This gives us $\hat{\varepsilon}_j$ instead of $\varepsilon_j$. The convergence of our estimator to the true variance $\tau^2(K)$ therefore requires verification. Our goal is to establish that for each fixed point $K$, the estimator $\hat{\tau}^2(K)$ converges to the true variance $\tau^2(K)$ with probability 1 as $n \rightarrow \infty$.

\subsubsection{Convergence analysis for Nadaraya-Watson Variance Estimator $\tilde{\tau}^2(K)$}

The estimator $\tilde{\tau}^2(K)$ is a Nadaraya-Watson kernel regression estimate of $E[\varepsilon^2 | K]$. Standard results of kernel regression confirms its unbiasedness up to smoothing bias. The residual term, however, requires separate analysis beyond classic results.

We first recall the Strong Law of Large Numbers (SLLN) theorem:
\begin{theorem}
Let $\{Z_i\}_{i=1}^n$ be an independent and identically distributed sequence of random variables with finite expectation $E[|Z_1|] < \infty$. Then
$$
\frac{1}{n} \sum_{i=1}^n Z_i \xrightarrow{\text{a.s.}} E[Z_1] \quad \text{as } n \rightarrow \infty
$$
\end{theorem}

We denote that:
\begin{align*}
    \tilde{\tau}^2(K) &= \frac{N(K)}{D(K)} = \frac{\sum_{j=1}^n \kappa\left(\frac{K_j-K}{h}\right) \varepsilon_j^2}{\sum_{j=1}^n \kappa\left(\frac{K_j-K}{h}\right)},\\
    Z_j &= \kappa\left(\frac{K_j-K}{h}\right) \varepsilon_j^2    
\end{align*}

Since $\{(K_j, \varepsilon_j)\}$ are i.i.d. and kernel $\kappa$ is symmetric with bounded support, the $Z_j$ are i.i.d. with finite expectation. Using the law of iterated expectations:
$$
E[Z_j] = E\left[\kappa\left(\frac{K_j-K}{h}\right) \varepsilon_j^2\right] = \int \kappa\left(\frac{t-K}{h}\right) E[\varepsilon_j^2 | K_j=t] g(t) dt = \int \kappa\left(\frac{t-K}{h}\right) \tau^2(t) g(t) dt
$$

By the SLLN,
$$
\frac{1}{n} N(K) = \frac{1}{n} \sum_{j=1}^n Z_j \xrightarrow{\text{a.s.}} \int \kappa\left(\frac{t-K}{h}\right) \tau^2(t) g(t) dt
$$

Similarily for the Denominator $D(K)$, we have:
\begin{align*}
    \xi_j &= \kappa\left(\frac{K_j-K}{h}\right),\quad E[\xi_j] = \int \kappa\left(\frac{t-K}{h}\right) g(t) dt,\\
    \Rightarrow&\left[\frac{1}{n} D(K) = \frac{1}{n} \sum_{j=1}^n \xi_j \xrightarrow{\text{a.s.}} \int \kappa\left(\frac{t-K}{h}\right) g(t) dt\right]
\end{align*}

Recall the Continuous Mapping theorem:
\begin{theorem}
\label{theom:cont map}
If $a_n \rightarrow a$ and $b_n \rightarrow b$ almost surely, and if $G$ is continuous at $b$ (with $b \neq 0$), then
$$
G(a_n, b_n) \rightarrow G(a, b) \quad \text{almost surely.}
$$
\end{theorem}

Taking
$$
a_n = \frac{1}{n} N(K) \quad \text{and} \quad b_n = \frac{1}{n} D(K)
$$
with
$$
a = \int \kappa\left(\frac{t-K}{h}\right) \tau^2(t) g(t) dt \quad \text{and} \quad b = \int \kappa\left(\frac{t-K}{h}\right) g(t) dt
$$

Under assumption \ref{assup:convergence analysis} and theorem \ref{theom:cont map}, we have $b \neq 0$ and 
$\hat{\tau}^2(K) = \frac{\frac{1}{n} N(K)}{\frac{1}{n} D(K)}\xrightarrow[n\to\infty]{a.s.}
\frac{\int \kappa\left(\frac{t-K}{h}\right) \tau^2(t) g(t) dt}{\int \kappa\left(\frac{t-K}{h}\right) g(t) dt}$.

Next, we analyze the convergence of its expectation and variance.

Let us denote:
\begin{align*}
\mu_N &= E[N(K)] = E\left[\sum_{j=1}^n \kappa\left(\frac{K_j-K}{h}\right) \varepsilon_j^2\right] = n\int \kappa\left(\frac{t-K}{h}\right) \tau^2(t) g(t) dt = O(nh)\\
\mu_D &= E[D(K)] = E\left[\sum_{j=1}^n \kappa\left(\frac{K_j-K}{h}\right)\right] = n\int \kappa\left(\frac{t-K}{h}\right) g(t) dt = O(nh)
\end{align*}

We do a second-order Taylor expansion of the function $G(a, b) = \frac{a}{b}$ around the point $(\mu_N, \mu_D)$:

\begin{align*}
\tilde{\tau}^2(K) &= \frac{\mu_N}{\mu_D} + \frac{1}{\mu_D}(N(K) - \mu_N) - \frac{\mu_N}{\mu_D^2}(D(K) - \mu_D) + \frac{1}{2}\frac{\partial^2 G}{\partial a^2}(\mu_N, \mu_D)(N(K) - \mu_N)^2 \\
&+ \frac{\partial^2 G}{\partial a \partial b}(\mu_N, \mu_D)(N(K) - \mu_N)(D(K) - \mu_D) + \frac{1}{2}\frac{\partial^2 G}{\partial b^2}(\mu_N, \mu_D)(D(K) - \mu_D)^2 + R_n
\end{align*}
where $\mathcal{R}_n$ represents higher-order remainder terms.

Taking the expectation, and noting that $E[N(K) - \mu_N] = E[D(K) - \mu_D] = 0$, we get:
\begin{align*}
E[\tilde{\tau}^2(K)] &= \frac{\mu_N}{\mu_D} + \frac{1}{2}\frac{\partial^2 G}{\partial a^2}(\mu_N, \mu_D)E[(N(K) - \mu_N)^2] \\
&+ \frac{\partial^2 G}{\partial a \partial b}(\mu_N, \mu_D)E[(N(K) - \mu_N)(D(K) - \mu_D)] \\
&+ \frac{1}{2}\frac{\partial^2 G}{\partial b^2}(\mu_N, \mu_D)E[(D(K) - \mu_D)^2] + E[\mathcal{R}_n]
\end{align*}

where
\begin{align*}
\frac{\partial^2 G}{\partial a^2}(a, b) = 0,\quad \frac{\partial^2 G}{\partial a \partial b}(a, b) = -\frac{1}{b^2} ,\quad \frac{\partial^2 G}{\partial b^2}(a, b) = \frac{2a}{b^3}
\end{align*}

The expectation of $Z_j$ can be calculated using the law of iterated expectations:
\begin{align}
E[Z_j] &= E\left[\kappa\left(\frac{K_j-K}{h}\right)^2 \varepsilon_j^2\right] \\
&= \int \kappa\left(\frac{t-K}{h}\right)^2 E[\varepsilon_j^2|K_j=t] g(t) dt \\
&= \int \kappa\left(\frac{t-K}{h}\right)^2 \tau^2(t) g(t) dt\\
&=\int \kappa(u)^2 \tau^2(K+hu) g(K+hu) h du=O(h)
\end{align}

where the last step holds because $\kappa$ is a positive symmetric kernel with bounded support, and both $\tau^2$ and $g$ are continuous functions.

Therefore, $E\left[\sum_{j=1}^n Z_j\right] = n \cdot E[Z_j] = n \cdot O(h) = O(nh)$.

By law of total variance:$
\operatorname{Var}\left(Z_j\right)=E\left[\operatorname{Var}\left(Z_j \mid K_j\right)\right]+\operatorname{Var}\left(E\left[Z_j \mid K_j\right]\right)
$, one can express the conditional moments:
$E\left[Z_j \mid K_j=t\right]=K\left(\frac{t-K}{h}\right)^2 \tau^2(t)$, and $\operatorname{Var}\left(Z_j \mid K_j=t\right)=K\left(\frac{t-K}{h}\right)^4 \cdot \operatorname{Var}(\varepsilon_j^2 | K_j=t)$

Since the fourth moment of $\varepsilon_j$ exists by assumption, $\operatorname{Var}(\varepsilon_j^2 | K_j=t) = O(1)$. After change of variables $u=\frac{t-K}{h}$ and integration:
$\operatorname{Var}\left(Z_j\right)= O(h)$ and this gives $\text{Var}(N(K))=O(nh)$.

Similarly, one gets:
$$
\begin{gathered}
\operatorname{Var}(D(K))=\operatorname{Var}\left(\sum_{j=1}^n \xi_j\right)=\sum_{j=1}^n \operatorname{Var}\left(\xi_j\right)=O(n h) \\
\operatorname{Cov}(N(K), D(K))=\operatorname{Cov}\left(\sum_{j=1}^n Z_j, \sum_{j=1}^n \xi_j\right)=\sum_{j=1}^n \operatorname{Cov}\left(Z_j, \xi_j\right)=O(n h)
\end{gathered}
$$

Now, substituting into our expansion:
\begin{align*}
E[\tilde{\tau}^2(K)] &= \frac{\mu_N}{\mu_D} + \left(-\frac{1}{(\mu_D)^2}\right)O(nh) + \frac{1}{2}\left(\frac{2\mu_N}{(\mu_D)^3}\right)O(nh) + E[R_n]\\
&= \frac{\mu_N}{\mu_D} + O\left(\frac{nh}{(nh)^2}\right)= \frac{\mu_N}{\mu_D} + O\left(\frac{1}{nh}\right)
\end{align*}

In terms of bias, we perform a Taylor expansion of $\tau^2(t)$ around $K$:
\begin{align*}
\tau^2(K+hu) &= \tau^2(K) + hu\tau^{2\prime}(K) + \frac{1}{2}h^2u^2\tau^{2\prime\prime}(K) + O(h^3)
\end{align*}

By making a change of variables $t = K + hu$, one see:
\begin{align*}
\int \kappa\left(\frac{t-K}{h}\right) \tau^2(t) g(t) dt &= \int \kappa(u) \tau^2(K+hu) g(K+hu) h du\\
&= h\tau^2(K)g(K)\int \kappa(u) du + h^2\tau^{2\prime}(K)g(K)\int u\kappa(u) du\\
&+ \frac{h^2}{2}\tau^2(K)g'(K)\int u\kappa(u) du + O(h^3)\\
\int \kappa\left(\frac{t-K}{h}\right) g(t) dt &= hg(K)\int \kappa(u) du + \frac{h^2}{2}g'(K)\int u\kappa(u) du + O(h^3)
\end{align*}

Since $\kappa$ is a symmetric kernel, $\int u\kappa(u) du = 0$, and $\int \kappa(u) du = 1$, we get:
\begin{align*}
\frac{\mu_N}{\mu_D} &= \frac{h\tau^2(K)g(K) + O(h^3)}{hg(K) + O(h^3)}\\
&= \tau^2(K) + O(h^2)
\end{align*}

Combining the results and we obtain:
\begin{align*}
E[\tilde{\tau}^2(K)] &= \tau^2(K) + O(h^2) + O\left(\frac{1}{nh}\right)
\end{align*}

Next, we see the variance is computed as:
\begin{align*}
\text{Var}(\tilde{\tau}^2(K)) &= \text{Var}\left(\frac{\mu_N}{\mu_D} + \frac{1}{\mu_D}(N(K) - \mu_N) - \frac{\mu_N}{\mu_D^2}(D(K) - \mu_D) + \mathcal{R}_2\right) \\
&= \text{Var}\left(\frac{1}{\mu_D}(N(K) - \mu_N) - \frac{\mu_N}{\mu_D^2}(D(K) - \mu_D) + \mathcal{R}_2\right) \\
&= \frac{1}{\mu_D^2}\text{Var}(N(K)) + \frac{\mu_N^2}{\mu_D^4}\text{Var}(D(K)) - \frac{2\mu_N}{\mu_D^3}\text{Cov}(N(K), D(K)) + \text{Var}(R_2) \\
&+ \frac{2}{\mu_D}\text{Cov}(N(K), \mathcal{R}_2) - \frac{2\mu_N}{\mu_D^2}\text{Cov}(D(K), \mathcal{R}_2)
\end{align*}

From the properties of Taylor remainders and under our assumption \ref{assup:convergence analysis}, the terms involving $\mathcal{R}_2$ are of smaller order than the main terms. Thus:
\begin{align*}
\text{Var}(\tilde{\tau}^2(K)) &= \frac{1}{(\mu_D)^2}\text{Var}(N(K)) + \frac{(\mu_N)^2}{(\mu_D)^4}\text{Var}(D(K)) - \frac{2\mu_N}{(\mu_D)^3}\text{Cov}(N(K), D(K)) + o\left(\frac{1}{nh}\right) \\
&= \frac{1}{(O(nh))^2} O(nh) + \frac{(O(nh))^2}{(O(nh))^4} O(nh) - \frac{2O(nh)}{(O(nh))^3}O(nh) + o\left(\frac{1}{nh}\right) \\
&= O\left(\frac{1}{nh}\right) + O\left(\frac{1}{nh}\right) - O\left(\frac{1}{nh}\right) + o\left(\frac{1}{nh}\right) = O\left(\frac{1}{nh}\right)
\end{align*}

Now, we obtain the following lemma:
\begin{lemma}[Convergence rate of Nadaraya-Watson Variance Estimator]
\label{lemma:convergence_ND}
Under the conditions that $h \to 0$ and $nh \to \infty$ as $n \to \infty$, $\tilde{\tau}^2(K)$ converges to $\tau^2(K)$ in rate
$$O_p\left(\frac{1}{\sqrt{n h}}\right)+O\left(h^2\right)+O\left(\frac{1}{n h}\right)$$
where $O$ notation represents the exact rate at which a fixed sequence converges, and $O_p$ notation represents the rate at which a random sequence converges in probability.
\end{lemma}
\begin{proof}
To analyze the deviation of $\tilde{\tau}^2(K)$ from its expectation, we use Chebyshev's inequality:

\begin{align*}
P\left(|\tilde{\tau}^2(K) - E[\tilde{\tau}^2(K)]| > \varepsilon\right) \leq \frac{\text{Var}(\tilde{\tau}^2(K))}{\varepsilon^2} = O\left(\frac{1}{nh\varepsilon^2}\right)
\end{align*}

Setting $\varepsilon = C/\sqrt{nh}$ for some constant $C$ gives:

\begin{align*}
P\left(|\tilde{\tau}^2(K) - E[\tilde{\tau}^2(K)]| > \frac{C}{\sqrt{nh}}\right) \leq \frac{O(1/(nh))}{C^2/(nh)} = O\left(\frac{1}{C^2}\right)
\end{align*}

This shows:
\begin{align*}
\tilde{\tau}^2(K) - E[\tilde{\tau}^2(K)] = O_p\left(\frac{1}{\sqrt{nh}}\right)
\end{align*}

Combining with our bias result:
\begin{align*}
\tilde{\tau}^2(K) - \tau^2(K) &= [\tilde{\tau}^2(K) - E[\tilde{\tau}^2(K)]] + [E[\tilde{\tau}^2(K)] - \tau^2(K)] \\
&= O_p\left(\frac{1}{\sqrt{nh}}\right) + O(h^2) + O\left(\frac{1}{nh}\right)
\end{align*}

Note that one can obtain the result more directly via bias-variance decomposition. 

The decompose the estimation error can be written as:
\begin{align}
\tilde{\tau}^2(K) - \tau^2(K) = \underbrace{[\tilde{\tau}^2(K) - E[\tilde{\tau}^2(K)]]}_{\text{stochastic error}} + \underbrace{[E[\tilde{\tau}^2(K)] - \tau^2(K)]}_{\text{bias}}
\end{align}

where the stochastic error term converges at rate $O_p(1/\sqrt{nh})$ by Chebyshev's inequality, while the bias term exhibits rate $O(h^2) + O(1/(nh))$ through Taylor expansion.  In the following proof sections, we directly apply the bias-variance decomposition in convergence analysis.

\end{proof}

\subsubsection{Decomposition of the Estimated Residual Variance}

Now consider the actual estimator $\hat{\tau}^2(K)$ which uses the estimated residuals $\hat{\varepsilon}_j = \sigma_j - \hat{f}(K_j)$ instead of the true $\varepsilon_j$. We can decompose it as follows:

$$\hat{\varepsilon}_j^2 = (\sigma_j - \hat{f}(K_j))^2 = (\sigma_j - f(K_j) + f(K_j) - \hat{f}(K_j))^2 = \varepsilon_j^2 + 2\varepsilon_j[f(K_j) - \hat{f}(K_j)] + [f(K_j) - \hat{f}(K_j)]^2$$

Plugging this into the formula for $\hat{\tau}^2(K)$, we get:
\begin{align*}
\hat{\tau}^2(K) &= \frac{\sum_{j=1}^n \kappa\left(\frac{K_j-K}{h}\right) \varepsilon_j^2}{\sum_{j=1}^n \kappa\left(\frac{K_j-K}{h}\right)} + \frac{\sum_{j=1}^n \kappa\left(\frac{K_j-K}{h}\right) 2\varepsilon_j[f(K_j) - \hat{f}(K_j)]}{\sum_{j=1}^n \kappa\left(\frac{K_j-K}{h}\right)} \\
&+ \frac{\sum_{j=1}^n \kappa\left(\frac{K_j-K}{h}\right)[f(K_j) - \hat{f}(K_j)]^2}{\sum_{j=1}^n \kappa\left(\frac{K_j-K}{h}\right)}
\end{align*}

Denote:
\begin{align*}
\text{ND} &= \frac{\sum_{j=1}^n \kappa\left(\frac{K_j-K}{h}\right) \varepsilon_j^2}{\sum_{j=1}^n \kappa\left(\frac{K_j-K}{h}\right)},\quad
A = \frac{\sum_{j=1}^n \kappa\left(\frac{K_j-K}{h}\right) 2\varepsilon_j[f(K_j) - \hat{f}(K_j)]}{\sum_{j=1}^n \kappa\left(\frac{K_j-K}{h}\right)} \\
B &= \frac{\sum_{j=1}^n \kappa\left(\frac{K_j-K}{h}\right)[f(K_j) - \hat{f}(K_j)]^2}{\sum_{j=1}^n \kappa\left(\frac{K_j-K}{h}\right)}
\end{align*}

We need to show that $A$ and $B$ converge to zero as $n \rightarrow \infty$, $h \rightarrow 0$, and $nh \rightarrow \infty$.

We first introduce an asymptotic convergence theorem from \citep{fan1996study} :
\begin{theorem}
\label{theorem:fan}
Assume that $g\left(k\right)>0$ and that $g(\cdot), f^{(p+1)}(\cdot)$ and $\sigma^2(\cdot)$ are continuous in a neighborhood of $x_0$. Further, assume that $h \rightarrow 0$ and $n h \rightarrow \infty$. Then the asymptotic conditional variance of $\widehat{f}_\nu\left(k\right)$ is given by

$$
\begin{gathered}
\operatorname{Var}\left\{\hat{f}_\nu\left(x_0\right) \mid  \mathbb{K}\right\}=e_{\nu+1}^T S^{-1} S^* S^{-1} e_{\nu+1} \frac{\nu!^2 \tau^2\left(k\right)}{g\left(k\right) n h^{1+2 \nu}} \\
+o_P\left(\frac{1}{n h^{1+2 \nu}}\right)
\end{gathered}
$$

The asymptotic conditional bias for $p-\nu$ odd is given by

$$
\begin{gathered}
\operatorname{Bias}\left\{\widehat{f}_\nu\left(k\right) \mid  \mathbb{K}\right\}=e_{\nu+1}^T S^{-1} c_p \frac{\nu!}{(p+1)!} f^{(p+1)}\left(k\right) h^{p+1-\nu} \\
+o_P\left(h^{p+1-\nu}\right)
\end{gathered}
$$

Further, for $p-\nu$ even the asymptotic conditional bias is

$$
\begin{aligned}
\operatorname{Bias}\left\{\widehat{f}_\nu\left(x_0\right) \mid\mathbb{K}\right\}= & e_{\nu+1}^T S^{-1} \tilde{c}_p \frac{\nu!}{(p+2)!}\left\{f^{(p+2)}\left(k\right)\right. \\
& \left.+(p+2) f^{(p+1)}\left(k\right) \frac{g^{\prime}\left(k\right)}{g\left(k\right)}\right\} h^{p+2-\nu} \\
& +o_P\left(h^{p+2-\nu}\right)
\end{aligned}
$$

provided that $g^{\prime}(\cdot)$ and $f^{(p+2)}(\cdot)$ are continuous in a neighborhood of $k$ and $n h^3 \rightarrow \infty$.    
\end{theorem}

By leveraging this theorem, we obtain:

\begin{corollary}[Convergence of A]
\label{corollary:convergence_a}
Under the conditions that $h \to 0$ and $nh \to \infty$ as $n \to \infty$, term $A$ converges to zero in probability at rate:
$$A = O_p\left(h^{p+1} + \frac{1}{n^{1/2}h^{1/2}}\right)$$
\end{corollary}

\begin{proof}
Let us denote the numerator and denominator of $A$ as:
\begin{align}
N_A &= \sum_{j=1}^n \kappa\left(\frac{K_j-K}{h}\right) 2\varepsilon_j[f(K_j) - \hat{f}(K_j)] ,\quad
D = \sum_{j=1}^n \kappa\left(\frac{K_j-K}{h}\right)
\end{align}

Since $\varepsilon_j$ and $f(K_j) - \hat{f}(K_j)$ are not independent, we apply the Cauchy-Schwarz inequality:
\begin{align}
|N_A| &\leq 2\left(\sum_{j=1}^n \kappa\left(\frac{K_j-K}{h}\right)^2 \varepsilon_j^2\right)^{1/2} \left(\sum_{j=1}^n \kappa\left(\frac{K_j-K}{h}\right)^2 [f(K_j) - \hat{f}(K_j)]^2\right)^{1/2}
\end{align}

For the first term involving $\varepsilon_j^2$, we have $
\sum_{j=1}^n \kappa\left(\frac{K_j-K}{h}\right)^2 \varepsilon_j^2 = O_p(nh)$.

For the second term involving $[f(K_j)-\hat{f}(K_j)]^2$, we utilize the asymptotic properties from Fan's theorem \ref{theorem:fan}. Specifically, the local polynomial estimator $\hat{f}(x)$ has:
\begin{align}
\text{Bias}[\hat{f}(x)] &= 
\begin{cases}
O(h^{p+1}), & \text{for $p$ odd} \\
O(h^{p+2}), & \text{for $p$ even}
\end{cases}, \quad
\text{Var}[\hat{f}(x)] = O\left(\frac{1}{nh}\right)
\end{align}

WLOG, we assume $p$ odd for proposition \ref{prop:tau convergence rate} and proposition \ref{prop:ACMSE to TRUE MSE}. Similar analysis can be done for $p$ even by leveraging theorem \ref{theorem:fan}.

Next we apply the standard decomposition of the mean squared error:
\begin{align}
f(K_j)-\hat{f}(K_j) = O_p\left(h^{p+1} + \frac{1}{\sqrt{nh}}\right)
\end{align}

The second term in our Cauchy-Schwarz bound becomes:
\begin{align}
\sum_{j=1}^n \kappa\left(\frac{K_j-K}{h}\right)^2 [f(K_j) - \hat{f}(K_j)]^2 &= O_p\left(nh \cdot \left(h^{p+1} + \frac{1}{\sqrt{nh}}\right)^2\right) \\
&= O_p\left(nh \cdot \left(h^{2(p+1)} + \frac{2h^{p+1}}{\sqrt{nh}} + \frac{1}{nh}\right)\right) \\
&= O_p\left(nh^{2p+3} + 2n^{1/2}h^{p+3/2} + 1\right) \\
&= O_p\left(nh^{2p+3} + 1\right)
\end{align}

Substituting these bounds back into our Cauchy-Schwarz application:
\begin{align}
|N_A| &= O_p\left((nh)^{1/2} \cdot (nh^{2p+3} + 1)^{1/2}\right) \\
&= O_p\left((nh)^{1/2} \cdot ((nh^{2p+3})^{1/2} + 1)\right) \\
&= O_p\left(n^{1/2}h^{1/2} \cdot (n^{1/2}h^{p+3/2} + 1)\right) \\
&= O_p\left(nh^{p+2} + n^{1/2}h^{1/2}\right)
\end{align}

From our earlier analysis of the denominator, $D = O_p(nh)$. Therefore:
\begin{align}
|A| = \frac{|N_A|}{D} = \frac{O_p(nh^{p+2} + n^{1/2}h^{1/2})}{O_p(nh)} = O_p\left(h^{p+1} + \frac{1}{n^{1/2}h^{1/2}}\right)
\end{align}

Since $h \to 0$ and $nh \to \infty$ as $n \to \infty$, both terms converge to zero, proving that $A \xrightarrow{p} 0$.
\end{proof}

\begin{corollary}[Convergence of B]
\label{corollary:convergence_b}
Under the conditions that $h \to 0$ and $nh \to \infty$ as $n \to \infty$, term $B$ converges to zero in probability at rate:
$$B = O_p\left(h^{2p+2} + \frac{1}{nh}\right)$$
\end{corollary}

\begin{proof}
Let us denote the numerator and denominator of $B$ as:
\begin{align}
N_B &= \sum_{j=1}^n \kappa\left(\frac{K_j-K}{h}\right)[f(K_j) - \hat{f}(K_j)]^2, \quad
D = \sum_{j=1}^n \kappa\left(\frac{K_j-K}{h}\right)
\end{align}

From the asymptotic properties of the local polynomial estimator in theorem \ref{theorem:fan}
\begin{align}
[f(K_j) - \hat{f}(K_j)]^2 = O_p\left(\left(h^{p+1} + \frac{1}{\sqrt{nh}}\right)^2\right) = O_p\left(h^{2(p+1)} + \frac{2h^{p+1}}{\sqrt{nh}} + \frac{1}{nh}\right)
\end{align}

Using the law of iterated expectations:
\begin{align}
E[N_B] &= E\left[\sum_{j=1}^n \kappa\left(\frac{K_j-K}{h}\right)E\left[[f(K_j) - \hat{f}(K_j)]^2 | K_j\right]\right]
\end{align}

Since the expected squared error equals the squared bias plus the variance:
\begin{align}
E\left[[f(K_j) - \hat{f}(K_j)]^2 | K_j\right] = O\left(h^{2(p+1)} + \frac{1}{nh}\right)
\end{align}

Therefore:
\begin{align}
E[N_B] &= O\left(h^{2(p+1)} + \frac{1}{nh}\right) \sum_{j=1}^n E\left[\kappa\left(\frac{K_j-K}{h}\right)\right] \\
&= O\left(h^{2(p+1)} + \frac{1}{nh}\right) \cdot n \cdot O(h)= O\left(nh^{2p+3} + 1\right)
\end{align}

By the Law of Large Numbers:
\begin{align}
N_B = O_p\left(nh^{2p+3} + 1\right)
\end{align}

From our earlier analysis, $D = O_p(nh)$. Therefore:
\begin{align}
B = \frac{N_B}{D} = \frac{O_p\left(nh^{2p+3} + 1\right)}{O_p(nh)} = O_p\left(h^{2p+2} + \frac{1}{nh}\right)
\end{align}

Since $h \to 0$ and $nh \to \infty$ as $n \to \infty$, both terms converge to zero, proving that $B \xrightarrow{p} 0$.
\end{proof}

Recall:
\begin{align*}
\tilde{\tau}^2(K) - \tau^2(K) &= O_p\left(\frac{1}{\sqrt{nh}}\right) + O(h^2) + O\left(\frac{1}{nh}\right),\quad
A = O_p\left(h^{p+1} + \frac{1}{n^{1/2}h^{1/2}}\right) \\
B &= O_p\left(h^{2p+2} + \frac{1}{nh}\right)
\end{align*}

Since $\hat{\tau}^2(K) = \tilde{\tau}^2(K) + A + B$, we have:
\begin{align*}
|\hat{\tau}^2(K) - \tau^2(K)| &= |(\tilde{\tau}^2(K) - \tau^2(K)) + A + B| \\
&\leq |\tilde{\tau}^2(K) - \tau^2(K)| + |A| + |B| \\
&= O_p\left(\frac{1}{\sqrt{nh}}\right) + O(h^2) + O\left(\frac{1}{nh}\right) + O_p\left(h^{p+1} + \frac{1}{n^{1/2}h^{1/2}}\right) + O_p\left(h^{2p+2} + \frac{1}{nh}\right)
\end{align*}

Note that $h^{2p+2} = o(h^{p+1})$ since $h \to 0$ and $p \geq 0$, and $\frac{1}{nh} = o\left(\frac{1}{n^{1/2}h^{1/2}}\right)$ since $\frac{n^{1/2}h^{1/2}}{nh} = \frac{1}{n^{1/2}h^{1/2}} \to 0$ as $nh \to \infty$.

Therefore, simplifying and keeping only the dominant terms:
\begin{align*}
|\hat{\tau}^2(K) - \tau^2(K)| &= O_p\left(\frac{1}{\sqrt{nh}}\right) + O_p(h^{p+1}) + O(h^2) + O\left(\frac{1}{nh}\right)
\end{align*}

In conclusion, we have established the consistency and convergence rate of the Nadaraya-Watson variance estimator $\hat{\tau}^2(K)$ used in the Order Adaptive Local Regression algorithm, which ensures that our approach provides reliable error variance estimates for implied volatility modeling.

\subsection{Proof for proposition \ref{prop:hp are global optimal}}
\label{proposition: cvp}
Following notations from \citep{fan2018local}, we write

$$
C_{\nu, p}(K)=\left[\frac{(p+1)!^2(2 \nu+1) \int K_\nu^{* 2}(t) d t}{2(p+1-\nu)\left\{\int t^{p+1} K_\nu^*(t) d t\right\}^2}\right]^{1 /(2 p+3)}
$$

We also denote $S=\left(\mu_{j+\ell}\right)_{0 \leq j, \ell \leq p}$ as the moment matrix of the kernel, $S^*=\left(\nu_{j+\ell}\right)_{0 \leq j, \ell \leq p}$ as the second moment matrix of the kernel, $c_p=\left(\mu_{p+1}, \ldots, \mu_{2 p+1}\right)^T$, $C_1(p)=\left(e_1^T S^{-1} c_p \frac{1}{(p+1)!}\right)^2$, $C_2(p)=e_1^T S^{-1} S^* S^{-1} e_1$

For the special case where $\nu=0$ (estimating the function itself), we have:$
C_{0, p}(K)=\left[\frac{C_2(p)}{2(p+1) C_1(p)}\right]^{1 /(2 p+3)}
$

Here is the proof:

\begin{proof}
For $\int K_\nu^{* 2}(t) d t$, we transform this into an equivalent form by the definition of the equivalent kernel $K_\nu^*$ from \citep{fan2018local}:$K_\nu^*(t)=e_{\nu+1}^T S^{-1}\left(1, t, \ldots, t^p\right)^T K(t)$. We square this expression and integrate:$
\int K_\nu^{* 2}(t) d t=\int\left[e_{\nu+1}^T S^{-1}\left(1, t, \ldots, t^p\right)^T K(t)\right]^2 d t
$. This can be rewritten by noting that for two vectors $a$ and $b,\left(a^T b\right)^2=a^T\left(b b^T\right) a$ :$
\int K_\nu^{* 2}(t) d t=\int e_{\nu+1}^T S^{-1}\left(1, t, \ldots, t^p\right)^T K(t) \cdot K(t)\left(1, t, \ldots, t^p\right) S^{-1} e_{\nu+1} d t
$. Since $e_{\nu+1}^T S^{-1}$ and $S^{-1} e_{\nu+1}$ are independent of $t$, we can move them outside the integral:

$$
e_{\nu+1}^T S^{-1}\left[\int\left(1, t, \ldots, t^p\right)^T\left(1, t, \ldots, t^p\right) K^2(t) d t\right] S^{-1} e_{\nu+1}
$$

The matrix inside the integral has elements $\int t^{i+j} K^2(t) d t$, which is precisely the definition of $S^*$.

Therefore:

$$
\int K_\nu^{* 2}(t) d t=e_{\nu+1}^T S^{-1} S^* S^{-1} e_{\nu+1}
$$

For $\int t^{p+1} K_\nu^*(t) d t$, we transform this term into an equivalent form by expanding with the equivalent kernel:$
\int t^{p+1} K_\nu^*(t) d t=\int t^{p+1} e_{\nu+1}^T S^{-1}\left(1, t, \ldots, t^p\right)^T K(t) d t
$. Since $e_{\nu+1}^T S^{-1}$ is independent of $t$, i.e. $
\int t^{p+1} K_\nu^*(t) d t=e_{\nu+1}^T S^{-1}\left[\int t^{p+1}\left(1, t, \ldots, t^p\right)^T K(t) d t\right]
$, the vector inside the brackets has components $\int t^{p+1+j} K(t) d t$ for $j=0,1, \ldots, p$, which corresponds exactly to the vector $c_p=\left(\mu_{p+1}, \mu_{p+2}, \ldots, \mu_{2 p+1}\right)^T$ defined in the text, where $\mu_j=\int u^j K(u) d u$.

Therefore, one have $
\int t^{p+1} K_\nu^*(t) d t=e_{\nu+1}^T S^{-1} c_p
$. Substituting the established identities for the integrals, one get:

$$
\begin{gathered}
C_{0, p}(K)=\left[\frac{(p+1)!^2 \cdot e_1^T S^{-1} S^* S^{-1} e_1}{2(p+1)\left\{e_1^T S^{-1} c_p\right\}^2}\right]^{1 /(2 p+3)} \\
=\left[\frac{(p+1)!^2 \cdot C_2(p)}{2(p+1) \cdot\left\{e_1^T S^{-1} c_p\right\}^2}\right]^{1 /(2 p+3)}
\end{gathered}
$$

Since $C_1(p)=\left(e_1^T S^{-1} c_p \frac{1}{(p+1)!}\right)^2$, we have $(p+1)!^2 \cdot C_1(p)=\left\{e_1^T S^{-1} c_p\right\}^2$
Therefore:

$$
C_{0, p}(K)=\left[\frac{C_2(p)}{2(p+1) C_1(p)}\right]^{1 /(2 p+3)}
$$
\end{proof}
This completes the proof for our claim.

For a fixed polynomial order $p$, we take the derivative of $\hat{Z}_k(p, h)$ with respect to $h$ :

$$
\frac{\partial \hat{Z}_k(p, h)}{\partial h}=2(p+1) C_1(p) \cdot h^{2 p+1} \cdot\left[\hat{f}^{(p+1)}(k)\right]^2-\frac{C_2(p) \cdot \hat{\tau}^2(k)}{n h^2 \cdot \hat{g}(k)}
$$

Setting this equal to zero and solving for $h$, we get:

$$h=\left[\frac{C_2(p)}{2(p+1) C_1(p)}\right]^{\frac{1}{2 p+3}} \cdot\left[\frac{\hat{\tau}^2(k)}{\left[\hat{f}^{(p+1)}(k)\right]^2 \cdot \hat{g}(k)}\right]^{\frac{1}{2 p+3}} \cdot n^{-\frac{1}{2 p+3}}$$

By our previous claim, the constant term $\left[\frac{C_2(p)}{2(p+1) C_1(p)}\right]^{\frac{1}{2 p+3}}$ corresponds to $C_{v, p}$ in our formula, confirming that:

$$
\hat{h}_{o p t}(p)=C_{v, p}\left[\frac{\hat{\tau}^2(k)}{\left[\hat{f}^{(p+1)}(k)\right]^2 \hat{g}(k)}\right]^{\frac{1}{2 p+3}} n^{-\frac{1}{2 p+3}}
$$

Further, we notice that by taking the second order derivative:

$$
\frac{\partial^2 \hat{Z}_k(p, h)}{\partial h^2}=2(p+1)(2 p+1) C_1(p) \cdot h^{2 p} \cdot\left[\hat{f}^{(p+1)}(k)\right]^2+\frac{2 C_2(p) \cdot \hat{\tau}^2(k)}{n h^3 \cdot \hat{g}(k)}
$$

At any point $h>0$, this second derivative is strictly positive because:
$p>0$, $C_1(p)>0$ and $C_2(p)>0$, where $C_2(p)$ is a quadratic form with a positive regular kernel ,thus producing a positive value.$\left[\hat{f}^{(p+1)}(k)\right]^2>0$ , $\hat{\tau}^2(k)>0$, $\hat{g}(k)>0$. Therefore, the function $\hat{Z}_k(p, h)$ is strictly convex in $h$ for fixed $p$, and the critical point $\hat{h}_{\text {opt }}(p)$ is guaranteed to be a global minimum with respect to $h$.

\subsection{Proof for proposition \ref{prop:ACMSE to TRUE MSE}}
\begin{proof}
To start with, we first analyze the residual terms for estimated functions $\hat{f}^{(p+1)}(\cdot)$ and $\hat{g}(\cdot)$. WLOG, we assume polynominal $p$ odd following same settings as proposition \ref{prop:tau convergence rate}.

Recall that we have:
\begin{align}
\operatorname{Bias}\{\hat{f}_\nu(x_0) | \mathbb{X}\} &= e_{\nu+1}^T S^{-1} c_p \frac{\nu!}{(p+1)!} f^{(p+1)}(x_0) h^{p+1-\nu} + o_P(h^{p+1-\nu}) \\
\operatorname{Var}\{\hat{f}_\nu(x_0) | \mathbb{X}\} &= e_{\nu+1}^T S^{-1} S^* S^{-1} e_{\nu+1} \frac{\nu!^2 \tau^2(x_0)}{g(x_0) n h^{1+2\nu}} + o_P\left(\frac{1}{nh^{1+2\nu}}\right)
\end{align}

The MSE is directly obtained as $\text{MSE} = (\text{Bias})^2 + \text{Var}$. Denoting $C_1(p) = \left(e_{\nu+1}^T S^{-1} c_p \frac{\nu!}{(p+1)!}\right)^2$ and $C_2(p) = e_{\nu+1}^T S^{-1} S^* S^{-1} e_{\nu+1}$, for $\nu = 0$ we have:
\begin{align}
Z_k(p,h) = C_1(p) \cdot h^{2(p+1)} \cdot [f^{(p+1)}(k)]^2 + \frac{C_2(p) \cdot \tau^2(k)}{nh \cdot g(k)}+o_P\left(\frac{1}{n h^{1+2 \nu}}\right)+o_P\left(h^{p+1}\right)
\end{align}

For our estimators, standard asymptotic expansions yield the decomposition:
\begin{align}
\hat{f}^{(p+1)}(k) &= f^{(p+1)}(k) + B_f(p,h) + \xi_f(p,h) \\
\hat{\tau}^2(k) &= \tau^2(k) + B_\tau(p,h) + \xi_\tau(p,h) \\
\hat{g}(k) &= g(k) + B_g(h) + \xi_g(h)
\end{align}
where $B$ represent the deterministic bias terms, and $\xi$ represent the stochastic errors. 

Starting from $\hat{f}^{(p+1)}$, following the assumption of pilot estimations, a $p+a(a>1)$ order polynominal represents the true $f$ function. Leveraging theorem \ref{theorem:fan}, we see that
\begin{enumerate}
    \item $B_f(p,h)=O_P(h^a)$, when $a$ is even.
    \item $B_f(p,h)=O_P(h^{a+1})$, when $a$ is odd.
    \item $\xi_f(p, h)=O_P\left(\frac{1}{\sqrt{n h^{2 p+3}}}\right)$
\end{enumerate}
WLOG, we assume that $a$ is even, and the case for odd $a$ can be analyzed similarily.

Next, we analyze the estimation for design function $g$. By assumption, the estimation of design function shares the same asymptotic behavior of kernel density estimator $\hat{g}(x)=\frac{1}{\sum_{i=1}^n V\left(K_i\right) } \sum_{i=1}^n \frac{V\left(K_i\right)}{h}K\left(\frac{x-X_i}{h}\right)$ with a second order regular kernel. Here, $\sum_{i=1}^n V\left(K_i\right)$ refers to the total volume of $n$ different strikes. Clearly, one can see $\sum_{i=1}^n V\left(K_i\right)\geq n$, and for simplicity, we denote $N_V=\sum_{i=1}^n V\left(K_i\right)$.

Leveraging kernel estimation theorem, one can see its variance has asymptotic form: $\operatorname{Var}\{\hat{g}(x)\}=\frac{g(x) \int K^2(u) d u}{N_V h}+o\left(\frac{1}{N_V h}\right)$. On this basis, one can write
$$\hat{g}(k)=g(k)+\underbrace{\frac{h^2}{2} \mu_2(K) g^{\prime \prime}(k)+o\left(h^2\right)}_{B_g(h)}+\underbrace{O_p\left(\frac{1}{\sqrt{N_V h}}\right)}_{\xi_g(h)}$$

where $\mu_2(K)=\int u^2 K(u) d u$ is the second moment of the kernel function
and the bias term is specifically $O\left(h^2\right)$.

In brief, we have the following(under assumption a even adn p odd):
\begin{enumerate}
    \item $B_f(p, h)=O_P\left(h^a\right)$
    \item $B_g(h)=O(h^2)$
    \item $\xi_g(h)=O_p\left(\frac{1}{\sqrt{N_V h}}\right), N_V \geq n.$
    \item $\xi_f(p, h)=O_P\left(\frac{1}{\sqrt{n h^{2 p+3}}}\right)$
    \item $\left|\hat{\tau}(K)^2-\tau(K)^2\right|=O_p\left(\frac{1}{\sqrt{n h}}\right)+O_p\left(h^{p+1}\right)+O\left(h^2\right)+O\left(\frac{1}{n h}\right)$
\end{enumerate}

Substituting these into the MSE formula yields our practical MSE function:
\begin{align}
\hat{Z}_k(p,h) = C_1(p) \cdot h^{2(p+1)} \cdot [\hat{f}^{(p+1)}(k)]^2 + \frac{C_2(p) \cdot \hat{\tau}^2(k)}{nh \cdot \hat{g}(k)} +o_P\left(\frac{1}{n h^{1+2 \nu}}\right)+o_P\left(h^{p+1}\right)
\end{align}

We start our analysis from a fixed pair $(p,h)$:

We expand $\left[\hat{f}^{(p+1)}(k)\right]^2$ using the relationship $\hat{f}^{(p+1)}(k)=f^{(p+1)}(k)+B_f(p, h)+$ $\xi_f(p, h)$ :

$$
\left[\hat{f}^{(p+1)}(k)\right]^2=\left[f^{(p+1)}(k)\right]^2+2 f^{(p+1)}(k) B_f(p, h)+2 f^{(p+1)}(k) \xi_f(p, h)+\left[B_f(p, h)\right]^2+2 B_f(p, h) \xi_f(p, h)+\left[\xi_f(p, h)\right]^2
$$

The difference in the first term of $\hat{Z}_k(p, h)$ and $Z_k(p, h)$ is:

$$
C_1(p) \cdot h^{2(p+1)} \cdot\left\{2 f^{(p+1)}(k) B_f(p, h)+2 f^{(p+1)}(k) \xi_f(p, h)+\left[B_f(p, h)\right]^2+2 B_f(p, h) \xi_f(p, h)+\left[\xi_f(p, h)\right]^2\right\}
$$

after simplification, we get:

$$
C_1(p) \cdot\left\{O_P\left(h^{2 p+2+a}\right)+O_P\left(\frac{h^{p+\frac{1}{2}}}{\sqrt{n}}\right)+O_P\left(h^{2 p+2+2 a}\right)+O_P\left(\frac{h^{a+p+\frac{1}{2}}}{\sqrt{n}}\right)+O_P\left(\frac{1}{n h}\right)\right\}.
$$

For the term $\frac{1}{\hat{g}(k)}$, denote $\delta(k)=\frac{B_g(h)+\xi_g(h)}{g(k)}$, and since $\frac{1}{\hat{g}(k)}=\frac{1}{g(k)(1+\delta(k))}=\frac{1}{g(k)} \cdot \frac{1}{1+\delta(k)}$.

For $|\delta(k)|<1$, we can use the Taylor series expansion for $\frac{1}{1+x}: \frac{1}{1+x}=$ $\sum_{i=0}^{\infty}(-1)^i x^i=1-x+x^2-x^3+\ldots$

Therefore, we see Taylor expansion $\frac{1}{\hat{g}(k)}=\frac{1}{g(k)}\left[1-\delta(k)+\frac{2 \delta(k)^2}{(1+\rho)^3}\right]$, where $\rho$ is the middle value between 0 and $\delta(k)$:

$$
\frac{1}{\hat{g}(k)}=\frac{1}{g(k)+B_g(h)+\xi_g(h)}=\frac{1}{g(k)}\left(1-\frac{B_g(h)+\xi_g(h)}{g(k)}+\frac{\left(B_g(h)+\xi_g(h)\right)^2}{g(k)^2}-\cdots\right)
$$

Given that $\delta(k)=\frac{B g(h)+\xi g(h)}{g(k)}=O\left(h^2\right)+O_P\left(\frac{1}{\sqrt{n h}}\right)$, and $(1+\rho)^3$ is bounded for small $\rho$, we have: $\frac{1}{\hat{g}(k)}=$

$$
\frac{1}{g(k)}\left[1-\left(O\left(h^2\right)+O_P\left(\frac{1}{\sqrt{n h}}\right)\right)+O\left(\left(O\left(h^2\right)+O_P\left(\frac{1}{\sqrt{n h}}\right)\right)^2\right)\right]
$$

For simplicity, I denote $\mathcal{R}_1(n,h)=O\left(\left(O\left(h^2\right)+O_P\left(\frac{1}{\sqrt{n h}}\right)\right)^2\right)]$ for higher order terms that will not influence the convergence behavior.

For $\hat{\tau}^2(k)$, we know: $
\hat{\tau}^2(k)=\tau^2(k)+O_P\left(\frac{1}{\sqrt{n h}}\right)+O_P\left(h^{p+1}\right)+O\left(h^2\right)+O\left(\frac{1}{n h}\right).
$

The difference in the second term becomes:

$$
\frac{C_2(p)}{n h} \cdot\left\{\frac{\tau^2(k)}{g(k)}\left(-\frac{B_g(h)+\xi_g(h)}{g(k)}+\mathcal{R}_1(n,h)\right)+\frac{O_P\left(\frac{1}{\sqrt{n h}}\right)+O_P\left(h^{p+1}\right)+O\left(h^2\right)+O\left(\frac{1}{n h}\right)}{g(k)}+\mathcal{R}_2(n,h)\right\}
$$

where $\mathcal{R}_2(n,h)$ is some higher order terms generated by $(-\frac{B_g(h)+\xi_g(h)}{g(k)}+\mathcal{R}_1(n, h))$ times $O_P\left(\frac{1}{\sqrt{n h}}\right)+O_P\left(h^{p+1}\right)+O\left(h^2\right)+O\left(\frac{1}{n h}\right)$

After simplifying:

$$
C_2(p) \cdot\left\{O\left(\frac{h}{n}\right)+O_P\left(\frac{1}{nh \sqrt{n h}}\right)+O_P\left(\frac{h^p}{n}\right)+O\left(\frac{1}{n^2 h^2}\right)+\frac{\mathcal{R}_2(n,h)}{nh}\right\}
$$

Therefore, the overall difference $\hat{Z}_k(p, h)-Z_k(p, h)$ is:

\begin{align*}
\hat{Z}_k(p,h)-Z_k(p,h)
&= O_P\bigl(h^{2p+2+a}\bigr)
  + O_P\!\left(\frac{h^{p+\frac12}}{\sqrt{n}}\right)
  + O_P\bigl(h^{2p+2+2a}\bigr) \\
&\quad+ O_P\!\left(\frac{h^{a+p+\frac12}}{\sqrt{n}}\right)
  + O_P\!\left(\frac{1}{nh}\right)
  + O\!\left(\frac{h}{n}\right) \\
&\quad+ O_P\!\left(\frac{1}{nh\sqrt{nh}}\right)
  + O_P\!\left(\frac{h^p}{n}\right)
  + O\!\left(\frac{1}{n^2h^2}\right)+\frac{\mathcal{R}_2(n,h)}{nh}.
\end{align*}

Under the conditions $n \rightarrow \infty, n h \rightarrow \infty$, and $h \rightarrow 0$, since $a\geq 2$, $p\geq 1$, and by Talor expansion $\lim _{\substack{n \rightarrow \infty \\ n h \rightarrow \infty \\ h \rightarrow 0}}\frac{\mathcal{R}_2(n,h)}{nh}\rightarrow 0$, one can see every term in this expression converges to zero. Therefore:

$$
\lim _{\substack{n \rightarrow \infty \\ n h \rightarrow \infty \\ h \rightarrow 0}}\left[\hat{Z}_k(p, h)-Z_k(p, h)\right]=0
$$

Next, we establish the convergence of $\hat{Z}_k^*$ to $Z_k^*$.

Let's denote:

$\hat{P h}^*=\left\{(p, h) \mid \hat{Z}_k(p, h)=\hat{Z}_k^*\right\}$ as the set of parameter pairs achieving the minimum estimated ACMSE

$P h^*=\left\{(p, h) \mid Z_k(p, h)=Z_k^*\right\}$ as the set of parameter pairs achieving the minimum true MSE

We proceeds the proof by contradiction.

Suppose that $\lim _{\substack{n \rightarrow \infty \\ nh \rightarrow \infty \\ h \rightarrow 0}} \hat{Z}_k^* \neq Z_k^*$.
This implies that there exists some $\epsilon>0$ and a subsequence of $n$, denoted as $\left\{n_j\right\}_{j=1}^{\infty}$, such that:

$$
\left|\hat{Z}_{k, n_j}^*-Z_k^*\right|>\epsilon
$$

for all $j$, where $\hat{Z}_{k, n_j}^*$ explicitly indicates the dependence on sample size $n_j$.

For each $n_j$, let $\left(p_{n_j}^*, h_{n_j}^*\right)$ be a parameter pair that achieves the optimal estimated ACMSE, i.e., $\left(p_{n_j}^*, h_{n_j}^*\right) \in \hat{Ph}^*_{n_j}$
By the definition of minimum, for any $(p, h) \in P h$ and in particular for any $(p, h) \in P h^*$:

$$
\hat{Z}_{k, n_j}\left(p_{n_j}^*, h_{n_j}^*\right) \leq \hat{Z}_{k, n_j}(p, h)
$$

Similarly, for any $(p, h) \in P h$ and specifically for $\left(p_{n_j}^*, h_{n_j}^*\right)$ :

$$
Z_k\left(p^*, h^*\right) \leq Z_k\left(p_{n_j}^*, h_{n_j}^*\right)
$$

where $\left(p^*, h^*\right) \in P h^*$ is any parameter pair achieving the optimal true MSE.

From our point-wise convergence results from first part, we know that for any fixed pair $(p, h)$ :

$$
\lim _{\substack{n \rightarrow \infty \\ n h \rightarrow \infty \\ h \rightarrow 0}} \hat{Z}_{k, n}(p, h)=Z_k(p, h)
$$

In particular, for $\left(p^*, h^*\right) \in P h^*$ :

$$
\lim _{j \rightarrow \infty} \hat{Z}_{k, n_j}\left(p^*, h^*\right)=Z_k\left(p^*, h^*\right)=Z_k^*
$$

Therefore, for any $\delta>0$, there exists $J_1$ such that for all $j>J_1$ :

$$
\left|\hat{Z}_{k, n_j}\left(p^*, h^*\right)-Z_k^*\right|<\delta
$$

Now, since $h_{n_j}^*=C_{\nu, p_{n_j}^*}(\kappa)\left[\frac{\hat{\tau}^2(k)}{\left\{\hat{f}^{\left(p_{n_j}^*+1\right)}(k)\right\}^2 \hat{g}(k)}\right]^{\frac{1}{2 p_{n_j}^*+3}} n_j^{-\frac{1}{2 p_{n_j}^*+3}}$, we know that as $n_j \rightarrow$ $\infty, h_{n_j}^* \rightarrow 0$

Furthermore, due to the constraint $nh \rightarrow \infty$, we have $n_j h_{n_j}^* \rightarrow \infty$ as $n_j \rightarrow \infty$.

This implies that the sequence of parameter pairs $\left(p_{n_j}^*, h_{n_j}^*\right)$ satisfies the conditions required for the pointwise convergence. Thus, there exists $J_2$ such that for all $j>J_2$ :

$$
\left|Z_k\left(p_{n_j}^*, h_{n_j}^*\right)-\hat{Z}_{k, n_j}\left(p_{n_j}^*, h_{n_j}^*\right)\right|<\delta
$$

Now, for $j>\max \left(J_1, J_2\right)$, we have:

$$
Z_k^* \leq Z_k\left(p_{n_j}^*, h_{n_j}^*\right)<\hat{Z}_{k, n_j}\left(p_{n_j}^*, h_{n_j}^*\right)+\delta=\hat{Z}_{k, n_j}^*+\delta
$$

$$
\hat{Z}_{k, n_j}^*=\hat{Z}_{k, n_j}\left(p_{n_j}^*, h_{n_j}^*\right) \leq \hat{Z}_{k, n_j}\left(p^*, h^*\right)<Z_k^*+\delta
$$

Therefore:

$$
Z_k^*<\hat{Z}_{k, n_j}^*+\delta \text { and } \hat{Z}_{k, n_j}^*<Z_k^*+\delta
$$

This implies: $
\left|\hat{Z}_{k, n_j}^*-Z_k^*\right|<\delta
$. Choosing $\delta=\epsilon / 2$, we arrive at a contradiction to our initial assumption that $\mid \hat{Z}_{k, n_j}^*-$ $Z_k^* \mid>\epsilon$ for all $j$.

Therefore, one can obtain
\begin{align}
    \lim _{\substack{n \rightarrow \infty \\ n h \rightarrow \infty \\ h \rightarrow 0}} \hat{Z}_k^*=Z_k^*
\end{align}

This completes the proof.
\end{proof}

\section*{Disclaimer}
This paper was prepared for informational purposes in part by the Quantitative Research Group of JPMorgan Chase \& Co. This paper is not a product of the Research Department of JPMorgan Chase \& Co. or its affiliates. Neither JPMorgan Chase \& Co. nor any of its affiliates makes any explicit or implied representation or warranty and none of them accept any liability in connection with this paper, including, without limitation, with respect to the completeness, accuracy, or reliability of the information contained herein and the potential legal, compliance, tax, or accounting effects thereof. This document is not intended as investment research or investment advice, or as a recommendation, offer, or solicitation for the purchase or sale of any security, financial instrument, financial product or service, or to be used in any way for evaluating the merits of participating in any transaction.

\bibliographystyle{unsrtnat}

\end{document}